\documentclass{article}
\usepackage{array} 
\usepackage[english]{babel}
\usepackage[utf8]{inputenc}
\usepackage{authblk}
\usepackage{multicol}
\usepackage{setspace}
\usepackage[margin=1.25in]{geometry}
\usepackage{pgfplots}
\usepgfplotslibrary{fillbetween} 
\pgfplotsset{compat=newest}
\usepackage{graphicx}
\graphicspath{ {./figures/} }
\usepackage{subcaption}
\usepackage{amsmath}
\usepackage{amsthm}
\usepackage{amssymb}
\usepackage{braket}
\usepackage{dsfont}
\usepackage{xcolor}
\usepackage{mathrsfs} 
\usepackage{mathtools} 
\usepackage{tikz}
\usetikzlibrary{calc}
\usetikzlibrary{arrows}
\usepackage{tcolorbox}
\usepackage{enumitem}
\newlist{steps}{enumerate}{1}
\setlist[steps, 1]{label = Step \arabic*:}
\usepackage[colorinlistoftodos,prependcaption]{todonotes}
\usepackage{booktabs}
\usepackage{thmtools,thm-restate}
\usepackage[
backend=bibtex,
style=alphabetic,
giveninits=true, 
maxbibnames=99,
minalphanames=3,
date=year
]{biblatex} 
\usepackage{pdfpages} 

\usepackage{hyperref} 
\usepackage{xurl} 
\usepackage{doi} 
\hypersetup{
colorlinks = true,
citecolor= blue,
urlcolor= blue,
linkcolor = blue
}

\renewbibmacro{in:}{}

\AtEveryBibitem{\clearfield{issue}} 

\newbibmacro{string+doiurl}[1]{%
	\iffieldundef{doi}{%
		\iffieldundef{url}{
			#1%
		}{%
			\href{\thefield{url}}{#1}%
		}%
	}{%
		\href{http://doi.org/\thefield{doi}}{#1}%
	}%
}
\DeclareFieldFormat{title}{\usebibmacro{string+doiurl}{\mkbibemph{#1}}}
\DeclareFieldFormat[article,thesis,incollection,inproceedings,manual]{title}%
{\usebibmacro{string+doiurl}
	{#1} 
}
\usepackage{algorithm}
\usepackage[noend]{algpseudocode}

\newtheorem{theorem}{Theorem}[section]

\newtheorem{rem}[theorem]{Remark}

\newtheorem{prop}[theorem]{Proposition}

\theoremstyle{definition} 
\newtheorem{defn}[theorem]{Definition}

\newcounter{algsubstate}

\newenvironment{algsubstates}
  {\setcounter{algsubstate}{0}%
   \renewcommand{\State}{%
     \refstepcounter{algsubstate}%
     \Statex {\normalsize\arabic{ALG@line}.\arabic{algsubstate}.}\kern5pt}
     }
  {}

\newcommand{\ketbra}[2]{\ket{#1}\!\bra{#2}}
\newcommand{\vect}[1]{\mathbf{#1}}

\newcommand\norm[1]{\lVert#1\rVert}

\newcommand{\eps}{\epsilon} 
\newcommand{\freq}{\operatorname{freq}}
\newcommand{\term}[1]{\textup{\textbf{#1}}} 
\newcommand{\Renyi}{R\'{e}nyi}
\newcommand{\dtol}{\delta_\mathrm{tol}}

\title{Bounds on Petz-Rényi Divergences and their Applications for Device-Independent Cryptography}

\author[1]{Thomas A. Hahn}
\author[2]{Ernest Y.-Z. Tan}
\author[3]{Peter Brown}

\affil[1]{The Center for Quantum Science and Technology, Department of Physics of Complex Systems, Weizmann Institute of Science, Rehovot, Israel}
\affil[2]{
University of Waterloo, Waterloo, Ontario N2L 3G1, Canada}
\affil[3]{T\'{e}l\'{e}com Paris-LTCI, Institut Polytechnique de Paris, 19 Place Marguerite Perey, 91120 Palaiseau, France}

\date{}

\begin{document}

\maketitle

\begin{abstract}
    Variational techniques have been recently developed to find tighter bounds on the von Neumann entropy in a completely device-independent (DI) setting. This, in turn, has led to significantly improved key rates of DI protocols, in both the asymptotic limit as well as in the finite-size regime. In this paper, we discuss two approaches towards applying these variational methods for Petz-Rényi divergences instead. We then show how this can be used to further improve the finite-size key rate of DI protocols, utilizing a fully-Rényi entropy accumulation theorem developed in a partner work. Petz-Rényi divergences can also be applied to study DI advantage distillation, in which two-way communication is used to improve the noise tolerance of quantum key distribution (QKD) protocols. We implement these techniques to derive increased noise tolerances for DIQKD protocols, which surpass all previous known bounds.
\end{abstract}

\section{Introduction}
The goal of quantum key distribution (QKD) protocols is to generate a secure key that seems uniformly random to an eavesdropper, but is known to all honest parties that applied the protocol. Quantum entanglement appears to be particularly suitable for this task. Two parties, Alice and Bob, can measure shared, highly entangled states, such as Bell states, to produce correlated bits. At the same time, their outcomes will be largely unknown to an adversary. One potential method for Alice and Bob to verify they have shared entanglement is by checking whether their devices produced non-local output distributions, which violated a Bell inequality~\cite{10.1093/oso/9780198788416.001.0001,brunner2013bell}. 
Basing the security of a QKD protocol on such Bell violations lays the foundation for what is known as device-independent quantum key distribution (DIQKD)~\cite{pironio2009deviceindependent,scarani2013deviceindependent}. 

Bell inequalities depend solely on the measurement statistics, and DIQKD security proofs thus do not require any additional assumptions on the initial quantum state or measurement devices, other than that they act on separate (generally finite-dimensional) Hilbert spaces and do not secretly leak information~\cite{pironio2009deviceindependent}.\footnote{Alternatively to requiring separate Hilbert spaces, one could assume that the measurements simply commute~\cite{brown2023deviceindependent,navascues2008convergent,Pironio_2010}.} As such, DIQKD provides security guarantees under what is considered to be the weakest possible set of assumptions~\cite{Arnon-Friedman:2018aa,pironio2009deviceindependent}.

One of the strongest tools in device-independent cryptography are so-called entropy accumulation theorems (EATs)~\cite{DFR20,MFSR22}. They relate the smooth min-entropy that is amassed over all rounds of a protocol to single-round quantities. As the former quantity characterizes the finite-size secure key length in protocols such as quantum key distribution (QKD), while the latter is comparatively simple to analyze, this provides a convenient framework for DI security proofs~\cite{Arnon-Friedman:2018aa,Metger_2023}. In~\cite{DFR20,MFSR22} this single-round quantity is the von Neumann entropy, and a lot of work has gone into bounding it in a completely device-independent setting~\cite{pironio2009deviceindependent,Woodhead_2021,Sekatski_2021,Tan_2021,brown2021computing,masini2022simple,brown2023deviceindependent}. For general DIQKD set-ups, the tightest existing bounds are usually achieved by \cite{brown2023deviceindependent}, where they approximate the von Neumann entropy by a variational optimization problem. 

However, we advance the state of the art on finite-size security proofs in a partner work \cite{arx_AHT24}, where we prove a new type of EAT that is expected to generally yield tighter bounds on the overall accumulated entropy. In that result, the single-round quantity is related to {\Renyi} divergences, rather than the von Neumann entropy. As such, to analyze this {\Renyi} EAT for DIQKD protocols, one would ideally like to be able to generate variational bounds for {\Renyi} divergences, analogous to \cite{brown2023deviceindependent}. 

\begin{rem}
In this work we focused on discussing our final bounds in terms of the smooth min-entropy of the final state, in order to provide a clearer comparison to previous works. However, EAT-based approaches more fundamentally give a bound on the {\Renyi} entropy of that state as an intermediate step, and for the purposes of practical applications, better finite-size keyrates are obtained by directly combining that bound with a {\Renyi} privacy amplification theorem from~\cite{arx_Dup21}. Under that approach, our methods would yield security proofs entirely based on {\Renyi} entropies. This has the dual advantages of being simpler to implement (as it avoids introducing intermediate quantities that need to be optimized, such as smoothing parameters) and 
also yielding better keyrates. (The concept of fully-{\Renyi} security proofs has been previously considered in a framework termed \term{quantum probability estimation}~\cite{ZFK20}, but we believe that our bounds should be simpler to evaluate.)
\end{rem}

Furthermore, a drawback of DIQKD is that these protocols are less robust against noise~\cite{Tan_2020}. In device-dependent QKD, noise tolerances can be significantly increased by including advantage distillation, i.e. extra post-processing steps which rely on two-way communication~\cite{gottesman2002proof}. Advantage distillation has also been studied in the context of device-independent cryptography~\cite{Tan_2020,Hahn_2022,stasiuk2022quantum}, under the assumption of collective attacks. In \cite{stasiuk2022quantum}, they prove a security condition, that can be used to relate the Quantum Chernoff Coefficient to the achievable noise tolerances of DIQKD protocols which include advantage distillation. Although they show that advantage distillation also improves the noise robustness of DIQKD protocols, their results are most likely not optimal. This is because they could not tightly approximate the Quantum Chernoff Coefficient in a device-independent setting, and had to resort to using ``Fuchs-van de Graaf"-like inequalities, which relate this quantity to the trace distance~\cite{stasiuk2022quantum,Audenaert_2007}.  

In this work, we provide tight bounds on Petz-{\Renyi} divergences, using the techniques from \cite{brown2023deviceindependent}. That is, for all $m \in \mathbb{N}$, we construct polynomials, $P_{m}\ \text{and}\ Q_{m}$, such that 
\begin{eqnarray} \label{Eq: DivergenceBoundGoal}
    D_{\alpha} \left(\rho||\sigma\right) \leq  \sup_{Z} \left(\frac{1}{\alpha-1} \log \left( \operatorname{Tr}\left[\rho P_{m}(Z)\right]+\operatorname{Tr}\left[\sigma Q_{m}(Z)\right] 
\right)\right) \ \ \ \text{for} \ \alpha \in (0,1) \cup (1,2)\; ,
\end{eqnarray}
  whenever $\operatorname{Tr}\left[\rho\right]=1$.\footnote{Whenever this state is not normalized, one gains an additional term that only depends on $\operatorname{Tr}\left[\rho\right]$.} Moreover, this approximation can be made arbitrarily tight, as the right-hand-side converges to $D_{\alpha} \left(\rho||\sigma\right)$ for $m \to \infty$. We discuss two approaches towards constructing these polynomials in Section \ref{Section: TightBoundsPetzRenyiVariational}. It was first mentioned in \cite{brown2023deviceindependent} that their results, combined with L{\"o}wner's theorem, should be able to provide such bounds. Section \ref{SubSection: App2} analyzes this approach in full detail. In Section \ref{SubSection: App1}, we derive a different method towards achieving Eq.\ (\ref{Eq: DivergenceBoundGoal}). This latter approach is slightly more general, and the resulting polynomials possess the additional property of being positive semi-definite whenever $\alpha <1$, which those in Section \ref{SubSection: App2} do not generally have.\footnote{For $\alpha > 1$, $Q_m$ is negative semi-definite in both approaches. Moreover, for the latter approach, it is often possible to construct $P_m$ such that the terms which explicitly depend on the operator $Z$ are also negative semi-definite.}

In Section \ref{Section: GREATBounds}, we show how Eq.\ (\ref{Eq: DivergenceBoundGoal}) can be used to provide device-independent lower bounds on the single-round contribution from the {\Renyi} EAT. For the extended CHSH set-up~\cite{pironio2009deviceindependent}, we explicitly calculate the resulting bounds on the accumulated smooth min-entropy. These surpass the best previously known results achieved in~\cite{Liu_2021}.  In \cite{arx_AHT24}, we discuss why the {\Renyi} EAT is generally expected to provide tighter bounds, and our results show that this improvement also holds in a fully device-independent setting.

Eq.\ (\ref{Eq: DivergenceBoundGoal}) can also be used to derive improved noise tolerances for device-independent advantage distillation. In Section \ref{Section: AdvantageDistill}, we first prove a security condition that is based on the pretty good fidelity, and equivalent to the Quantum Chernoff Coefficient condition from \cite{stasiuk2022quantum}. The pretty good fidelity can be expressed as a function of $D_{\alpha} \left(\rho||\sigma\right)$, for $\alpha = 1/2$~\cite{Iten_2017}. As such, one can directly bound it using Eq.\ (\ref{Eq: DivergenceBoundGoal}), rather than resorting to the approach taken in \cite{stasiuk2022quantum}. We consider 
two explicit DIQKD set-ups and, in 
both cases, our noise tolerance bounds improve upon the results from \cite{stasiuk2022quantum}.

\section{Notation}
\label{sec:notation}

\begin{table}[h]
\caption{List of notation}\label{tab:notation}
\def\arraystretch{1.5} 
\setlength\tabcolsep{.28cm}
\begin{tabular}{c l}
\toprule
\textit{Symbol} & \textit{Definition} \\
\toprule
$\log$ & Base-$2$ logarithm \\
\hline
$\norm{\cdot}_p$ & Schatten $p$-norm \\
\hline
$A \perp B$ & A and B are orthogonal; $AB = BA = 0$ \\
\hline
$X\geq Y$ (resp.~$X>Y$) & $X-Y$ is positive semi-definite (resp.~positive definite)\\
\hline
$X \ll Y$ & $\ker(Y) \subseteq \ker(X)$ \\
\hline
$S_{=}(A)$ (resp.~$S_{\leq}(A)$) & Set of normalized (resp.~subnormalized) states on register $A$ \\
\hline
$\mathds{1}_A$ & Identity operator on register $A$ \\
\toprule
\end{tabular}
\def\arraystretch{1}
\end{table}
In this work, we mainly consider the smooth min-entropy, the Petz-{\Renyi} divergence, and the Petz-{\Renyi} entropy. We use the notation from~\cite{Tomamichel2015QuantumIP}.
\begin{defn}
For any bipartite, subnormalized density matrix, $\rho\in S_{\leq}(AB)$, the \term{min-entropy of $A$ conditioned on $B$} is given by
\begin{align}
H_{\operatorname{min}}(A|B)_\rho &:= 
-\log 
\min_{\substack{\sigma \in S_{\leq}(B) \ \text{s.t.}\\ \rho_B \ll \sigma_B}} 
\norm{\rho_{AB}^\frac{1}{2}
(\mathds{1}_A \otimes \sigma_{B})
^{-\frac{1}{2}}}_\infty^2 \; .
\end{align}
Moreover, given any $\eps\in \left[0,\sqrt{\operatorname{Tr}{\rho_{AB}}}\right)$,
the \term{$\eps$-smoothed min-entropy of $A$ conditioned on $B$} is defined as 
\begin{align}
H_{\operatorname{min}}^\eps(A|B)_\rho :=
\max_
{\substack{\tilde{\rho} \in S_{\leq}(AB) \ \text{s.t.} \\ P(\tilde{\rho},\rho)\leq\eps}}
H_{\operatorname{min}}(A|B)_{\tilde{\rho}} \; ,
\end{align}
where $P$ represents the purified distance from~\cite{Tomamichel2015QuantumIP}. 
\end{defn}

\begin{defn}
Given any two positive semi-definite operators $\rho , \sigma\in \text{Pos}(A)$ with $\operatorname{Tr}\left[\rho\right] > 0$, and ${\alpha\in(0,1)\cup (1,2)}$, the \term{Petz-{\Renyi} divergence} between $\rho$, $\sigma$ is given by:
\begin{align}
    D_\alpha(\rho||\sigma)=\begin{cases}
    \frac{1}{\alpha-1}\log\frac{\operatorname{Tr} \left[\rho^\alpha\sigma^{1-\alpha}\right]}{\operatorname{Tr}\left[\rho\right]} &\left(\alpha < 1\ \wedge\ \rho\not\perp\sigma\right)\vee \rho \ll \sigma \\
    +\infty & \text{otherwise} \; .
    \end{cases}  
\end{align}
\end{defn}

\begin{defn}
For any bipartite, normalized state $\rho\in
S_{=}(AB)
$, and $\alpha\in (0,1)\cup (1,2)$, the \term{Petz-{\Renyi} conditional entropy} is given by
\begin{align}
    \label{eq:cond_renyi}
    &H_\alpha(A|B)_\rho=-D_\alpha(\rho_{AB}||\mathds{1}_A\otimes\rho_B) \; .
\end{align}
\end{defn}
In Section \ref{Section: GREATBounds}, we require the following definitions when discussing the {\Renyi} EAT. Here, we use a similar notation to~\cite{arx_AHT24}.
\begin{defn}\label{def:freq}
(Frequency distributions) Given a string $z_1^n\in\mathcal{Z}^n$ on some alphabet $\mathcal{Z}$, $\freq_{z_1^n}$ denotes the following probability distribution on $\mathcal{Z}$:
\begin{align}
\freq_{z_1^n}(z) \coloneqq \frac{\text{number of occurrences of $z$ in $z_1^n$}}{n} .
\end{align}
\end{defn}


\pagebreak
\begin{defn}
(Conditioning on classical events) For a state $\rho \in S_{\leq}(CQ)$ classical on $C$, written in the form
$\rho_{CQ} = \sum_c \ketbra{c}{c} \otimes \omega_c$ 
for some $\omega_c \in S_{\leq}(Q)$,
and an event $\Omega$ defined on the register $C$, we will define a corresponding \term{partial state} and \term{conditional state} as, respectively,
\begin{align}
\rho_{\land\Omega} \coloneqq \sum_{c\in\Omega} \ketbra{c}{c} \otimes \omega_c, \qquad\qquad \rho_{|\Omega} \coloneqq \frac{\operatorname{Tr}[\rho]}{\operatorname{Tr}[\rho_{\land\Omega}]} \rho_{\land\Omega} = \frac{
\sum_{c} \operatorname{Tr}[\omega_c]
}{\sum_{c\in\Omega} \operatorname{Tr}[\omega_c]} \rho_{\land\Omega} ,
\end{align}
and as a slightly abbreviated version of this notation, we denote partial and conditional states for specific values of $c$ via
\begin{align}
\rho_{Q \land c} \coloneqq \omega_c, \qquad\qquad \rho_{Q | c} \coloneqq \frac{\operatorname{Tr}[\rho]}{\operatorname{Tr}[\omega_c]} \omega_c.
\end{align}
\end{defn}

\section{Bounds on \texorpdfstring{$D_{\alpha} \left(\rho||\sigma\right)$ for $\alpha \in (0,1) \cup (1,2)$}{Petz-{\Renyi} Entropy}} \label{Section: TightBoundsPetzRenyiVariational}

Our aim is to bound $\operatorname{Tr}\left[\rho^{\alpha}\sigma^{1-\alpha}\right]$ via a variational optimization problem. For $\alpha \in (0,1)$, $\rho$ and $\sigma$ may generally be allowed to denote arbitrary positive semi-definite operators. This is not the case when $\alpha \in (1,2)$, where one additionally requires that $\rho <\!\!<\sigma$ holds. We discuss two different approaches towards achieving this goal, one in Section \ref{SubSection: App1} and one in Section \ref{SubSection: App2}. In both cases, we are primarily interested in bounds of the form
 \begin{eqnarray}
    \operatorname{Tr}\left[\rho^{\alpha}\sigma^{1-\alpha}\right] \geq  \inf_{Z_1, Z_2, \dots} \operatorname{Tr}\left[\rho P_{m}(Z_1,Z_2,\dots)\right]+\operatorname{Tr}\left[\sigma Q_{m}(Z_1,Z_2,\dots)\right] 
\end{eqnarray}
for $\alpha \in (0,1)$ and
\begin{eqnarray} \label{Eq: IntroIdealUpperBoundAlphageq1}
    \operatorname{Tr}\left[\rho^{\alpha}\sigma^{1-\alpha}\right] \leq  \sup_{Z_1, Z_2, \dots} \operatorname{Tr}\left[\rho P_{m}(Z_1,Z_2,\dots)\right]+\operatorname{Tr}\left[\sigma Q_{m}(Z_1,Z_2,\dots)\right] 
\end{eqnarray}
for $\alpha \in (1,2)$. For $\alpha \in (0,1) \cup (1,2)$, the resulting bound on the Petz-{\Renyi} divergence is then given by
\begin{eqnarray} 
    D_{\alpha} \left(\rho||\sigma\right) \leq  \sup_{Z_1, Z_2, \dots} \left(\frac{1}{\alpha-1} \log \left( \operatorname{Tr}\left[\rho P_{m}(Z_1,Z_2,\dots)\right]+\operatorname{Tr}\left[\sigma Q_{m}(Z_1,Z_2,\dots)\right] 
\right)\right) - \frac{ \log \left(\operatorname{Tr}\left[\rho\right] \right)}{\alpha -1}\; ,
\end{eqnarray}
which reduces to $Eq.\ (\ref{Eq: DivergenceBoundGoal})$ whenever $\operatorname{Tr}\left[\rho\right]=1$.
All relevant proofs from this section may be found in Section \ref{Section: Appendix}.
\subsection{Approach 1} \label{SubSection: App1}
 Given any $x \geq 0$ and $\alpha \in \left(0,1\right)$, it is known that~\cite[Eq.~(V.4)]{BhatiaMatrixAnalysis}
\begin{align} 
    x^{\alpha} = \frac{\sin (\alpha \pi)}{\pi} \int_{0}^{\infty} \frac{x}{x+t}t^{\alpha-1}dt \; .
\end{align}
For the purposes of our analysis, we convert this improper integral into two definite integrals. The resulting expression is given in Lemma \ref{Lemma: Approach1SmallAlphaIntegralFormulation}.

\begin{restatable}{lem}{ApponeLemone}
\label{Lemma: Approach1SmallAlphaIntegralFormulation}
For any $x \geq 0$ and $\alpha \in \left(0,1\right)$,
    \begin{align} \label{Eq: Approach1SmallAlphaIntegralFormulation}
    x^{\alpha} = \frac{\sin (\alpha \pi)}{\pi} \left[\int_{0}^{1} \frac{x}{x+t}t^{\alpha-1}dt + \int_{0}^{1} \frac{x}{1+xt} t^{-\alpha}dt \right] \; .
\end{align}
\end{restatable}
We first consider $\operatorname{Tr}\left[\rho^{\alpha}\sigma^{1-\alpha}\right]$, where $\alpha \in (0,1) $ and $\rho$, $\sigma$ are density operators acting on a finite-dimensional Hilbert space $\mathcal{H}$. This value can be expressed in terms of the eigenvalues and eigenvectors of $\rho$ and $\sigma$. More precisely, 
\begin{align}
    \operatorname{Tr}\left[\rho^{\alpha}\sigma^{1-\alpha}\right] = \sum_{j,k} q_{k}^\alpha p_{j}^{1-\alpha}  \left| \braket{\psi_{j} | \phi_{k}}\right|^2 \; ,
\end{align}
where $\rho = \sum q_{k}\ketbra{\phi_{k}}{\phi_{k}}$ and $\sigma = \sum p_{j}\ketbra{\psi_{j}}{\psi_{j}}$, such that $q_{k}, p_{j} \geq 0$ and both $\{\ket{\psi_{j}}\}$ as well as $\{\ket{\phi_{k}}\}$ form an orthonormal basis of $\mathcal{H}$. Lemma \ref{Lemma: Approach1SmallAlphaIntegralFormulation} provides a method with which to further expand this expression. 
\begin{align}
    \operatorname{Tr}\left[\rho^{\alpha}\sigma^{1-\alpha}\right] &= \sum_{j,k} q_{k}^\alpha p_{j}^{1-\alpha}  \left| \braket{\psi_{j} | \phi_{k}}\right|^2  \\
    &= \sum_{\substack{j,k \ \mathrm{s.t.} \\ q_k,p_j >0}} q_{k}^\alpha p_{j}^{1-\alpha}  \left| \braket{\psi_{j} | \phi_{k}}\right|^2  \\
    &= \sum_{\substack{j,k \ \mathrm{s.t.} \\ q_k,p_j >0}} p_{j} \left(\frac{q_{k}}{p_{j}}\right)^{\alpha} \left| \braket{\psi_{j} | \phi_{k}}\right|^2 \\
    &= \frac{\sin (\alpha \pi)}{\pi} \sum_{\substack{j,k \ \mathrm{s.t.} \\ q_k,p_j >0}}  \left[\int_{0}^{1} \frac{q_k \cdot p_j}{q_k+p_j t}t^{\alpha-1}dt + \int_{0}^{1} \frac{q_k \cdot p_j}{p_j+q_k t} t^{-\alpha}dt \right] \left| \braket{\psi_{j} | \phi_{k}}\right|^2  \label{Eq: Approach1DoubleIntegral}\\
    &= \frac{\sin (\alpha \pi)}{\pi}   \int_{0}^{1} \left[ \sum_{\substack{j,k \ \mathrm{s.t.} \\ q_k,p_j >0}}\frac{q_k \cdot p_j}{q_k+p_j t} \left| \braket{\psi_{j} | \phi_{k}}\right|^2 t^{\alpha-1} +  \sum_{\substack{j,k \ \mathrm{s.t.} \\ q_k,p_j >0}}\frac{q_k \cdot p_j}{p_j+q_k t} \left| \braket{\psi_{j} | \phi_{k}}\right|^2 t^{-\alpha} \right] dt  \label{Eq: Approach1QuasiRelEntropy}
\end{align}
In the second line, we use the fact that one may remove all terms in the sum for which $q_k,p_j=0$, as they do not contribute. Moreover, removing them ensures that the subsequent fractions are well-defined. 
\begin{rem}
One may re-include these terms in the last line, as, for all $t\in (0,1]$,
\begin{align}
    \frac{q_k \cdot p_j}{q_k+p_j t} \ \text{and} \ \frac{q_k \cdot p_j}{p_j+q_k t}
\end{align}
can be continuously extended such that it is well-defined for all $q_k,p_j\geq 0$ and this extension yields a value of zero whenever either $q_k=0$ or $p_j=0$ (or both). 
    For $t=0$, this does not hold, as the former is simply equal to $p_j$, whereas the latter simplifies to $q_k$. As such, the former will be non-zero even if $q_k =0$ (the converse statement holds true for the latter). This effect at the endpoint will clearly not affect the outcome of the integral, but it may affect later results, where we apply quadratures that include said endpoint.
\end{rem}
The last line contains two terms of the form 
\begin{align}
    \sum_{\substack{j,k \ \mathrm{s.t.} \\ q_k,p_j >0}}f(q_k,p_j,t) \left| \braket{\psi_{j} | \phi_{k}}\right|^2 \; ,
\end{align}
for some function, $f$, which takes as input the eigenvalues, $q_k$ and $p_j$, as well as the variable of integration, $t$, and satisfy 
\begin{align} \label{Eq: IncludingAllSumsApproach1Part1}
    \sum_{\substack{j,k \ \mathrm{s.t.} \\ q_k,p_j >0}}f(q_k,p_j,t) \left| \braket{\psi_{j} | \phi_{k}}\right|^2  =  \sum_{j,k}f(q_k,p_j,t) \left| \braket{\psi_{j} | \phi_{k}}\right|^2\; ,
\end{align}
as long as $t \in (0,1]$. 
They are special forms of quasi-relative entropies, and can be expressed as a variational optimization problem, similarly to those in \cite{brown2023deviceindependent}.

\begin{restatable}{prop}{ApproachOnePartOneVarOpt}
\label{Prop: Approach1Part1VarOpt}
Let $\rho = \sum q_{k}\ketbra{\phi_{k}}{\phi_{k}}$ and $\sigma = \sum p_{j}\ketbra{\psi_{j}}{\psi_{j}}$ be positive semi-definite operators. For all $t \in \left(0,1\right]$,
\begin{align}
    \sum_{j,k}  \frac{q_k \cdot p_j}{q_k+p_j t} \left| \braket{\psi_{j} | \phi_{k}}\right|^2 &= \frac{1}{t} \inf_{Z} \left(\operatorname{Tr}\left[\rho Z^\dagger Z\right]+\operatorname{Tr}\left[\sigma\left( t \left(\mathds{1} + Z \right) \left(\mathds{1} + Z^\dagger \right)
        \right)\right]\right) \label{Eq: Approach1Part1VarOptEq1}\\
   \sum_{j,k}   \frac{q_k \cdot p_j}{p_j+q_k t} \left| \braket{\psi_{j} | \phi_{k}}\right|^2  &= \frac{1}{t} \inf_{Z} \left(\operatorname{Tr}\left[\rho\left( t \left(\mathds{1} + Z^\dagger \right) \left(\mathds{1} + Z \right)
        \right)\right] + \operatorname{Tr}\left[\sigma Z Z^\dagger\right]\right) \; . \label{Eq: Approach1Part1VarOptEq2}
\end{align}
\end{restatable}
We remark that for every value of $t>0$, Proposition \ref{Prop: Approach1Part1VarOpt} represents the terms in Eq.\ (\ref{Eq: Approach1QuasiRelEntropy}) in the form we require. The remaining issue is that $t$ is a variable of integration, and directly applying Proposition \ref{Prop: Approach1Part1VarOpt} would transform $\operatorname{Tr}\left[\rho^{\alpha}\sigma^{1-\alpha}\right]$ into a sum over uncountably many optimization problems. 

This problem can be dealt with by approximating both integrals in Eq.\ (\ref{Eq: Approach1DoubleIntegral}) with quadratures. It is important for our purposes, however, that the quadrature actually bounds the integral expression, rather than just approximating it. In our case, one can lower bound both integrals via Gauss-Jacobi (GJ) quadratures and upper bound them using Gauss-Radau-Jacobi (GRJ) quadratures. The conceptual difference between them is that the latter contains an endpoint (in our case we choose $t=0$) as a quadrature node, whereas the former does not. 

The errors stemming from these quadratures can be related to the derivatives of the functions one is integrating \cite{Hildebrand56}. In particular, our bounds directly follow from the fact that, when taking the derivative over $t$, both $\frac{q_k \cdot p_j}{q_k+p_j t}$ and $\frac{q_k \cdot p_j}{p_j+q_k t}$ have positive even derivatives and negative odd derivatives, for all $t \in [0,1]$. This method of achieving upper and lower bounds is not specific to the functions we are considering, and is in fact a special case of~\cite[Eq.\ (8.9.8) and Eq.\ (8.10.22)]{Hildebrand56}. We include Lemma \ref{Lem: BoundingIntviaQuad} for the sake of completeness.
\begin{restatable}{lem}{BoundingIntviaQuad}
\label{Lem: BoundingIntviaQuad}
Let $\beta>-1$ and $f(t)$ be a function such that 
$\frac{d^{2k} f(t)}{d t^{2k}} \geq 0$ and $\frac{d^{2k-1} f(t)}{d t^{2k-1}} \leq 0$ for all $k \in \mathbb{N}$ and $t \in [0,1]$. In this case, an m-point GJ quadrature generates $m$ nodes $\{t_i\}_{i \in \{1,\dots,m\}}$ and weights $\{w_i\}_{i \in \{1,\dots,m\}}$
such that
\begin{align}
    \int_0^1 f(t) t^\beta dt \geq \sum_{i=1}^{m} w_i f(t_i) \; .
\end{align}
Conversely, an m-point GRJ quadrature with endpoint $t_1=0$ generates $m$ nodes $\{t_i\}_{i \in \{1,\dots,m\}}$ and weights $\{w_i\}_{i \in \{1,\dots,m\}}$ such that 
\begin{align}
    \int_0^1 f(t) t^\beta dt \leq \sum_{i=1}^{m} w_i f(t_i) \; .
\end{align}
\end{restatable}
Applying first Lemma \ref{Lem: BoundingIntviaQuad} to Eq.\ (\ref{Eq: Approach1DoubleIntegral}) and then Proposition \ref{Prop: Approach1Part1VarOpt} gives us the desired bound for $\alpha \in (0,1) $.
\begin{restatable}[Approach 1 for $\alpha \in (0,1)$]{thm}{TheoremApponePartone}
\label{Theorem: App1Part1}
    Let $H$ be a finite-dimensional Hilbert space, and let $\rho, \sigma$ be two positive semi-definite operators on $H$. Then, for any $\alpha \in (0,1)$ and $m \in \mathbb{N}$, an m-point GJ quadrature generates nodes $t_1, t_1^\prime, \dots, t_{m}, t_m^\prime \in \left(0,1\right)$ and weights $w_1,w_1^\prime,\dots, w_{m},w_m^\prime >0$ such that
    \begin{align}
    \begin{split}
         \operatorname{Tr}\left[\rho^{\alpha}\sigma^{1-\alpha}\right] &\geq \frac{\sin(\alpha \pi )}{\pi}\sum_{i=1}^{m}
        \frac{w_i}{t_{i}} \inf_{Z} \left(\operatorname{Tr}\left[\rho Z^\dagger Z\right]+\operatorname{Tr}\left[\sigma\left( t_i \left(\mathds{1} + Z \right) \left(\mathds{1} + Z^\dagger \right)
        \right)\right]\right) \\
        &\ \   + \frac{\sin(\alpha \pi )}{\pi}\sum_{i=1}^{m}
        \frac{w_i^\prime}{t_{i}^\prime} \inf_{Z^\prime}\left(\operatorname{Tr}\left[\rho\left( t_i^\prime \left(\mathds{1} + Z^{\prime \dagger} \right) \left(\mathds{1} + Z^\prime \right)\right)\right]+\operatorname{Tr}\left[\sigma Z^\prime Z^{\prime \dagger} 
        \right]\right) \; .
        \end{split}
    \end{align}
Moreover, as $m \to \infty$, the right-hand-side converges to $\operatorname{Tr}\left[\rho^{\alpha}\sigma^{1-\alpha}\right]$.
\end{restatable}
\begin{rem}
    For each value $i$, these bounds generate $2$ additional matrix variables, and, similarly to \cite{brown2023deviceindependent}, it is possible to change the order of the sum and infima, s.t.\ 
    \begin{eqnarray}
    \operatorname{Tr}\left[\rho^{\alpha}\sigma^{1-\alpha}\right] \geq  \inf_{Z_1, Z_1^\prime, \dots} \operatorname{Tr}\left[\rho P_{m}(Z_1,Z_1^\prime,\dots)\right]+\operatorname{Tr}\left[\sigma Q_{m}(Z_1,Z_1^\prime,\dots)\right] \; .
\end{eqnarray}
    The polynomials $P_m$ and $Q_m$ have the additional property of always being positive semi-definite. We consider these bounds for $\alpha \in (0,1)$ in Section \ref{Section: AdvantageDistill}.

    In principle, one can use GRJ quadratures (with the endnode $t = 0$) to similarly provide upper bounds. However, for $t = 0$, Proposition \ref{Prop: Approach1Part1VarOpt} is not directly applicable, due to the factor $\frac{1}{t}$. Instead, one can directly deal with the sums that appear in this case, i.e.\
    \begin{align}
         \sum_{\substack{j,k \ \mathrm{s.t.} \\ q_k,p_j >0}}  p_j \left| \braket{\psi_{j} | \phi_{k}}\right|^2 \ \text{and} \ \sum_{\substack{j,k \ \mathrm{s.t.} \\ q_k,p_j >0}}    q_k \left| \braket{\psi_{j} | \phi_{k}}\right|^2 \; .
    \end{align}
    These are equivalent to $\operatorname{Tr}\left[\rho^0 \sigma\right] := \lim_{\epsilon \to 0^+} \operatorname{Tr}\left[\rho^\epsilon \sigma\right]$ and $\operatorname{Tr}\left[\rho \sigma^0\right]  := \lim_{\epsilon \to 0^+} \operatorname{Tr}\left[\rho \sigma^\epsilon\right] $, respectively. As such, one can upper bound them via 
    \begin{align}
        \operatorname{Tr}\left[\rho^0 \sigma\right] &\leq \operatorname{Tr}\left[\sigma\right] \\
        \operatorname{Tr}\left[\rho \sigma^0\right] &\leq \operatorname{Tr}\left[\rho\right] \; .
    \end{align}
    Their corresponding weights $w_1$ and $w_1^\prime$ approach $0$, as the number of nodes, $m$, increases~\cite[Eq.\ (3.10)]{GRWeights}. Therefore,
    this upper bound would also produce increasingly tight results.  
\end{rem}
In the case of $\alpha \in (1,2)$, we apply a similar approach. We again use Lemma \ref{Lemma: Approach1SmallAlphaIntegralFormulation} to express $\operatorname{Tr}\left[\rho^{\alpha}\sigma^{1-\alpha}\right]$ via integrals, and apply Lemma \ref{Lem: BoundingIntviaQuad} to bound them, using quadratures. For any $\rho$, $\sigma$ such that $\rho <\!\!<\sigma$,
\begin{align}
\operatorname{Tr}\left[\rho^{\alpha}\sigma^{1-\alpha}\right] &= \sum_{\substack{j,k \ \mathrm{s.t.} \\ q_k,p_j >0}} q_{k}^\alpha p_{j}^{1-\alpha}  \left| \braket{\psi_{j} | \phi_{k}}\right|^2  \\
    &= \sum_{\substack{j,k \ \mathrm{s.t.} \\ q_k,p_j >0}} q_{k} \left(\frac{q_{k}}{p_{j}}\right)^{\alpha-1} \left| \braket{\psi_{j} | \phi_{k}}\right|^2 \\
    &= \frac{\sin ((\alpha -1)\pi)}{\pi} \sum_{\substack{j,k \ \mathrm{s.t.} \\ q_k,p_j >0}}  q_k \left[\int_{0}^{1} \frac{q_k }{q_k+p_j t}t^{\alpha-2}dt + \int_{0}^{1} \frac{q_k }{p_j+q_k t} t^{1-\alpha}dt \right] \left| \braket{\psi_{j} | \phi_{k}}\right|^2  \label{Eq: Approach1Part2Integral}\\
    &\approx \frac{\sin ((\alpha -1)\pi)}{\pi} \sum_{\substack{j,k \ \mathrm{s.t.} \\ q_k,p_j >0}} \sum_{i}\left[w_i \frac{q_k^2 }{q_k+p_j t_i}\left| \braket{\psi_{j} | \phi_{k}}\right|^2  + w_i^\prime\frac{q_k^2 }{p_j+q_k t_i^\prime} \left| \braket{\psi_{j} | \phi_{k}}\right|^2  \right]  
    \; ,  \label{Eq: Approach1Part2QuasiRelEntropy}
\end{align}
where the ``$\approx$" in the fourth line can be replaced by either ``$\geq$" or ``$\leq$", depending on the quadrature.
While the remaining quasi-relative entropies do not directly have the structure required by Proposition \ref{Prop: Approach1Part1VarOpt}, they can still be expressed as an optimization problem. 
\begin{rem} \label{Rem: CompletingSumApp1Part2}
  For $\alpha \in (1,2)$, we generally assume that $\rho <\!\!<\sigma$. As such, for any $j$ such that $p_j=0$, either $q_k=0$ or $\left| \braket{\psi_{j} | \phi_{k}}\right| = 0$. For $t>0$, this allows us to expand the sum to include all $j,k$ without changing the total output value. 
\end{rem}
\begin{restatable}{prop}{ApproachOnePartTwoVarOpt}
\label{Prop: Approach1Part2VarOpt}
Let $\rho = \sum q_{k}\ketbra{\phi_{k}}{\phi_{k}}$ and $\sigma = \sum p_{j}\ketbra{\psi_{j}}{\psi_{j}}$ be positive semi-definite operators. 
 For all $t \in \left(0,1\right]$, 
\begin{align}
    \sum_{j,k}  \frac{q_k^2}{q_k+p_j t} \left| \braket{\psi_{j} | \phi_{k}}\right|^2 &= \operatorname{Tr}\left[\rho\right] - 
   \inf_{Z} \left(\operatorname{Tr}\left[\rho Z^\dagger Z\right]+t \operatorname{Tr}\left[\sigma \left(\mathds{1} + Z \right) \left(\mathds{1} + Z^\dagger \right)
       \right]\right) \label{Eq: Approach2VarOptEq1}
    \\
   \sum_{j,k}   \frac{q_k^2}{p_j+q_k t} \left| \braket{\psi_{j} | \phi_{k}}\right|^2  &= \frac{1}{t} \left[ \operatorname{Tr}\left[\rho\right] - \inf_{Z} \left(\operatorname{Tr}\left[\rho 
   \left(\mathds{1} + Z^\dagger \right) \left(\mathds{1} + Z \right)
       \right] + \frac{1}{t}\operatorname{Tr}\left[\sigma Z Z^\dagger \right]\right) \right]
   \; . \label{Eq: Approach2VarOptEq2}
\end{align}
\end{restatable}
\begin{rem} \label{Rem: EndNodeApp1Part2}
    When we use the GRJ quadrature with endpoint $t=0$, Proposition \ref{Prop: Approach1Part2VarOpt} is no longer applicable for $t=0$. This, however, is not an issue. 
    The first term 
    contributes
\begin{align}
    \sum_{\substack{j,k \ \mathrm{s.t.} \\ q_k,p_j >0}} q_k \left| \braket{\psi_{j} | \phi_{k}}\right|^2 &= \operatorname{Tr}\left[\rho  \sigma^0\right] \\
    &= \operatorname{Tr}\left[\rho\right] \; ,
\end{align}
where we use $\rho <\!\!<\sigma$ in the last line. Similarly, the contribution from the second term is equal to 
\begin{align}
    \sum_{\substack{j,k \ \mathrm{s.t.} \\ q_k,p_j >0}} \frac{q_k^2}{ p_j} \left| \braket{\psi_{j} | \phi_{k}}\right|^2 =\operatorname{Tr}\left[\rho^2 \sigma^{-1}\right]  \; ,
\end{align}
where $\sigma^{-1}$ represents the pseudo-inverse of the density matrix $\sigma$. Moreover, whenever $\rho <\!\!<\sigma$, this can be expressed as~\cite[Remark 3.6]{brown2023deviceindependent} 
\begin{align}
    \operatorname{Tr}\left[\rho^2 \sigma^{-1}\right] = - \inf_{Z} \left(\operatorname{Tr}\left[\rho\left( Z^\dagger  + Z \right)\right] + \operatorname{Tr}\left[\sigma Z Z^\dagger\right]\right) \; .
\end{align}
\end{rem}

\begin{restatable}[Approach 1 for $\alpha \in (1,2)$]{thm}{TheoremApponeParttwo}
\label{Theorem: App1Part2}
    Let $H$ be a finite-dimensional Hilbert space, and let $\rho, \sigma$ be two positive semi-definite operators on $H$ such that $\rho <\!\!<\sigma$. Then, for any $\alpha \in (1,2)$ and $m \in \mathbb{N}$, an m-point GRJ quadrature with endpoints $t_1 = t_1^\prime =0$ generates nodes 
$t_1, t_1^\prime, \dots, t_{m}, t_m^\prime \in \left[0,1\right)$ and weights $w_1,w_1^\prime,\dots, w_{m},w_m^\prime >0$ such that 
\begin{align}
    \begin{split}
        \operatorname{Tr}\left[\rho^{\alpha}\sigma^{1-\alpha}\right] &\leq C_{m} + \frac{\sin(\alpha \pi )}{\pi}\sum_{i=2}^{m} w_i \inf_{Z} \left(\operatorname{Tr} \left[\rho Z^\dagger Z\right]+\operatorname{Tr}\left[\sigma \left( t_i \left(\mathds{1} + Z \right) \left(\mathds{1} + Z^\dagger \right)
        \right)\right]\right) \\
         & \  \ + \frac{\sin(\alpha \pi )}{\pi} w_1^\prime \inf_{Z^\prime}\left(\operatorname{Tr}\left[\rho \left( Z^{\prime \dagger} + Z^\prime \right)\right]+\operatorname{Tr}\left[\sigma\ Z^\prime Z^{\prime \dagger}
        \right]\right) \\
        & \  \ + \frac{\sin(\alpha \pi )}{\pi}\sum_{i=2}^{m}
        \frac{w_i^\prime}{t_{i}^\prime} \inf_{Z^\prime}\left(\operatorname{Tr}\left[\rho \left(\mathds{1} + Z^{\prime \dagger} \right) \left(\mathds{1} + Z^\prime \right)\right]+\frac{1}{t_{i}^\prime}\operatorname{Tr}\left[\sigma Z^\prime Z^{\prime \dagger}  
        \right]\right) \; ,
        \end{split}
    \end{align}
where $C_m = \frac{\sin((\alpha-1) \pi )}{\pi}\operatorname{Tr}\left[\rho\right] \left(\frac{1}{\alpha-1} + \sum_{i=2}^{m} \frac{w_{i}^\prime}{t_{i}^\prime} \right)$. As $m \to \infty$, this bound converges to $\operatorname{Tr}\left[\rho^{\alpha}\sigma^{1-\alpha}\right]$.
\end{restatable}
\begin{rem}
    
Conversely, an m-point GJ quadrature generates nodes $t_1, t_1^\prime, \dots, t_{m}, t_m^\prime \in \left(0,1\right)$ and weights $w_1,w_1^\prime,\dots, w_{m},w_m^\prime >0$ such that
    \begin{align} \label{Eq: App1Part2LowerBounds}
    \begin{split}
        \operatorname{Tr}\left[\rho^{\alpha}\sigma^{1-\alpha}\right] &\geq C_{m} + \frac{\sin(\alpha \pi )}{\pi}\sum_{i=1}^{m} w_i \inf_{Z} \left(\operatorname{Tr} \left[\rho Z^\dagger Z\right]+\operatorname{Tr}\left[\sigma \left( t_i \left(\mathds{1} + Z \right) \left(\mathds{1} + Z^\dagger \right)
        \right)\right]\right) \\
        & \  \ + \frac{\sin(\alpha \pi )}{\pi}\sum_{i=1}^{m}
        \frac{w_i^\prime}{t_{i}^\prime} \inf_{Z^\prime}\left(\operatorname{Tr}\left[\rho \left(\mathds{1} + Z^{\prime \dagger} \right) \left(\mathds{1} + Z^\prime \right)\right]+\frac{1}{t_{i}^\prime}\operatorname{Tr}\left[\sigma Z^\prime Z^{\prime \dagger} 
        \right]\right) \; ,
        \end{split}
    \end{align}
    where $C_m = \frac{\sin((\alpha-1) \pi )}{\pi}\operatorname{Tr}\left[\rho\right] \left(\frac{1}{\alpha-1} + \sum_{i=1}^{m} \frac{w_{i}^\prime}{t_{i}^\prime} \right)$.
  The lower bounds we achieve here will again generate polynomials $P_m$ and $Q_m$ that are positive semi-definite and satisfy
    \begin{eqnarray}
    \operatorname{Tr}\left[\rho^{\alpha}\sigma^{1-\alpha}\right] \geq C_m + \frac{\sin(\alpha \pi )}{\pi} \inf_{Z_1, Z_1^\prime, \dots} \operatorname{Tr}\left[\rho P_{m}(Z_1,Z_1^\prime,\dots)\right]+\operatorname{Tr}\left[\sigma Q_{m}(Z_1,Z_1^\prime,\dots)\right] \; .
\end{eqnarray}
   This is not true for the upper bounds due to the contribution from $t_1^\prime = 0$. If one requires the polynomials to be positive semi-definite, the contribution from $t_1^\prime = 0$ can sometimes be dealt with separately. Whenever $\rho \leq \beta \sigma$, where $\beta $ is some positive constant, one can use the bound~\cite[Remark 3.12]{brown2023deviceindependent} 
   \begin{align}
       \operatorname{Tr}\left[\rho^2 \sigma^{-1}\right] \leq \beta \operatorname{Tr}\left[\rho \right] \; .
   \end{align}
   This bound no longer contains any variational expressions over which one has to optimize, and can be incorporated into $C_m$. Moreover, for an $m$-point GRJ quadrature, the corresponding weight $w_1^{\prime}$ will approach $0$ as $m$ increases~\cite[Eq.\ (3.10)]{GRWeights}.
   As such, for large enough $m$, this approximation will not greatly affect the final result.
   
   Absorbing the sine pre-factor into the polynomials converts the optimization into a maximization problem. Incorporating both the pre-factor, as well as the term $C_m$, into the polynomials then yields the desired bound from $Eq.\ (\ref{Eq: IntroIdealUpperBoundAlphageq1})$. We consider ${\alpha \in (1,2)}$ in Section \ref{Section: GREATBounds}, and require upper bounds, which are achieved via GRJ quadratures.
\end{rem}
\subsection{Approach 2} \label{SubSection: App2}

The existence of the bound from this section was first discussed in~\cite{brown2023deviceindependent}, and it relies on the fact that $x^{1-\alpha}$ and $-x^{1-\alpha}$ are operator monotones for $\alpha \in \left[0,1\right)$ and $\alpha \in \left(1,2\right]$, respectively. In particular, L{\"o}wner's theorem \cite{Lwner1934berMM} provides an integral representation for these quantities, similar to Lemma \ref{Lemma: Approach1SmallAlphaIntegralFormulation}. For all $\alpha \in (0,1) \cup (1,2)$ \cite{faust2023rational,brown2023deviceindependent},
\begin{align}
    x^{1-\alpha} = 1 + \frac{\sin \left(\alpha\pi\right)}{\pi }\int_{0}^{1} \frac{x-1}{t\left(x-1\right) +1} t^{\alpha-1}\left(1-t\right)^{1-\alpha}dt \; .
\end{align}
Given any density matrices $\rho = \sum q_{k}\ketbra{\phi_{k}}{\phi_{k}}$ and $\sigma = \sum p_{j}\ketbra{\psi_{j}}{\psi_{j}}$,\footnote{If $\alpha >1$, then we include the constraint $\rho <\!\!<\sigma$.}
\begin{align}
\operatorname{Tr}\left[\rho^{\alpha}\sigma^{1-\alpha}\right] &= \sum_{\substack{j,k \ \mathrm{s.t.} \\ q_k,p_j >0}} q_{k}^\alpha p_{j}^{1-\alpha}  \left| \braket{\psi_{j} | \phi_{k}}\right|^2  \\
    &= \sum_{\substack{j,k \ \mathrm{s.t.} \\ q_k,p_j >0}} q_{k} \left(\frac{p_{j}}{q_{k}}\right)^{1-\alpha} \left| \braket{\psi_{j} | \phi_{k}}\right|^2 \\
    &=\sum_{\substack{j,k \ \mathrm{s.t.} \\ q_k,p_j >0}}  q_k \left[1 + \frac{\sin \left(\alpha\pi\right)}{\pi }\int_{0}^{1} \frac{p_j-q_k}{t\left(p_j-q_k\right) +q_k} t^{\alpha-1}\left(1-t\right)^{1-\alpha}dt \right] \left| \braket{\psi_{j} | \phi_{k}}\right|^2  \; .\label{Eq: App2Integral}
\end{align}
This integral can be upper and lower bounded via GRJ quadratures, by including either the endpoint $t_1 = 0$ or $t_m = 1$~\cite[Theorem 1]{faust2023rational}.
\begin{prop} [Special case of \cite{faust2023rational}, Theorem 1] \label{Prop: App2QuadBounds}
    Let $\alpha \in (0,1) \cup (1,2)$. For any $m \in \mathbb{N}$, an m-point GRJ quadrature with node $t_1 = 0$, will generate $m$ nodes $\{t_i\}_{i \in \{1,\dots,m\}}$ and weights $\{w_i\}_{i \in \{1,\dots,m\}}$
such that
\begin{align}
    \int_{0}^{1} \frac{x-1}{t\left(x-1\right) +1} t^{\alpha-1}\left(1-t\right)^{1-\alpha}dt \leq \sum_{i=1}^{m} w_i \frac{x-1}{t_i\left(x-1\right) +1}  \; .
\end{align}
Conversely, an m-point GRJ quadrature with node $t_m=1$ will generate $m$ nodes $\{t_i\}_{i \in \{1,\dots,m\}}$ and weights $\{w_i\}_{i \in \{1,\dots,m\}}$ such that 
\begin{align}
    \int_{0}^{1} \frac{x-1}{t\left(x-1\right) +1} t^{\alpha-1}\left(1-t\right)^{1-\alpha}dt \geq \sum_{i=1}^{m} w_i  \frac{x-1}{t_i\left(x-1\right) +1}  \; .
\end{align}
\end{prop}
\begin{rem}
    We will generally be interested in the m-point quadrature which includes the end node $t_m =1$. For this node, we need to consider the contribution
    \begin{align}
        \sum_{\substack{j,k \ \mathrm{s.t.} \\ q_k,p_j >0}}  q_k \left[\frac{\sin \left(\alpha\pi\right)}{\pi } \frac{p_j-q_k}{p_j}  \right] \left| \braket{\psi_{j} | \phi_{k}}\right|^2 = \sum_{\substack{j,k \ \mathrm{s.t.} \\ p_j >0}}  \left[ \frac{\sin \left(\alpha\pi\right)}{\pi }q_k - \frac{\sin \left(\alpha\pi\right)}{\pi } \frac{q_k^2}{p_j}  \right] \left| \braket{\psi_{j} | \phi_{k}}\right|^2 \; .
    \end{align}
    Here, one runs into the issue that there does not seem to be a clean way to express
\begin{align}
    \sum_{\substack{j,k \ \mathrm{s.t.} \\ p_j >0}}  \left[ \frac{\sin \left(\alpha\pi\right)}{\pi } \frac{q_k^2}{p_j}  \right] \left| \braket{\psi_{j} | \phi_{k}}\right|^2 = \frac{\sin \left(\alpha\pi\right)}{\pi } \operatorname{Tr}\left[\rho^2  \sigma^{-1}\right]\; ,
\end{align}
as a variational problem for any $\rho, \sigma$, where $\sigma^{-1}$ should be viewed as the pseudo-inverse of $\sigma$. If $\rho <\!\!<\sigma$, then it is known that~\cite[Remark 3.6]{brown2023deviceindependent} 
\begin{align}
    \operatorname{Tr}\left[\rho^2 \sigma^{-1}\right] = - \inf_{Z} \left(\operatorname{Tr}\left[\rho\left( Z^\dagger  + Z \right)\right] + \operatorname{Tr}\left[\sigma Z Z^\dagger\right]\right) \; .
\end{align}
    However, this no longer holds true for general $\rho$, $\sigma$.\footnote{This can be seen by considering  
$
        \rho = \begin{pmatrix}
1 & 0 \\
0 & 0
\end{pmatrix} \ , \ \sigma = \begin{pmatrix}
0 & 0 \\
0 & 1
\end{pmatrix} 
$, which are orthogonal. In this case, one has that $\operatorname{Tr}\left[\rho^2  \sigma^{-1}\right]=0$, but
$
         -\inf_{Z} \left(\operatorname{Tr}\left[\rho\left( Z^\dagger  + Z \right)\right] + \operatorname{Tr}\left[\sigma Z Z^\dagger\right]\right)
$ diverges. To see this, simply consider $Z=-k \cdot \rho$, where $k$ is a positive real number. These are feasible points in the optimization problem and, as $k\to \infty$, the objective function diverges. The reason for this difference is that we consider the pseudo-inverse, whereas the variational optimization problem is equivalent to $\lim_{\epsilon \to 0^+}\operatorname{Tr}\left[\rho^2  \left(\sigma +\epsilon \mathds{1} \right)^{-1}\right]$. } To avoid this issue, we will from now on assume that $\rho <\!\!<\sigma$, even for $\alpha < 1$.
\end{rem}
For any $t\in \left(0,1\right]$, and density operators $\rho, \sigma$ such that $\rho <\!\!<\sigma$, 
\begin{align}
    \sum_{\substack{j,k \ \mathrm{s.t.} \\ q_k,p_j >0}}  q_k \left[ \frac{\sin \left(\alpha\pi\right)}{\pi }\frac{p_j-q_k}{t\left(p_j-q_k\right) +q_k}  \right] \left| \braket{\psi_{j} | \phi_{k}}\right|^2  = \lim_{\epsilon \to 0^+} \sum_{j,k}  q_k \left[\frac{\sin \left(\alpha\pi\right)}{\pi }\frac{(p_j+\epsilon)-q_k}{t\left((p_j+\epsilon)-q_k\right) +q_k}  \right] \left| \braket{\psi_{j} | \phi_{k}}\right|^2 \; , \label{Eq: App2CompletionSum}
\end{align}
due to the fact that for any $k,j$ such that $p_j=0$, either $q_k=0$ or $\left| \braket{\psi_{j} | \phi_{k}}\right| = 0$. For simplicity, and in accordance with~\cite{brown2023deviceindependent}, we will simply refer to this limit as 
\begin{align}
    \sum_{j,k} q_k \left[\frac{ \sin\left(\alpha\pi\right)}{\pi }\frac{p_j-q_k}{t\left(p_j-q_k\right) +q_k}\right] \left| \braket{\psi_{j} | \phi_{k}}\right|^2 \; .
\end{align}
These quasi-relative entropies are known to have variational expressions~\cite[Proposition~3.5]{brown2023deviceindependent}.

\begin{prop}[\cite{brown2023deviceindependent}, Proposition~3.5] \label{Prop: Approach2VarOpt}
Let $\rho = \sum q_{k}\ketbra{\phi_{k}}{\phi_{k}}$ and $\sigma = \sum p_{j}\ketbra{\psi_{j}}{\psi_{j}}$ be positive semi-definite operators. For all $t \in \left(0,1\right]$,
\begin{align}
    \sum_{j,k} q_k \left[\frac{ \sin\left(\alpha\pi\right)}{\pi }\frac{p_j-q_k}{t\left(p_j-q_k\right) +q_k}\right] \left| \braket{\psi_{j} | \phi_{k}}\right|^2
\end{align}
is equal to
\begin{align}
    \frac{\sin\left(\alpha\pi\right)}{\pi \cdot t}
   \inf_{Z} \left(\operatorname{Tr}\left[\rho \left(1+Z+Z^\dagger + (1-t)Z^\dagger Z\right)\right]+t\operatorname{Tr}\left[\sigma Z Z^\dagger 
      \right]\right)   \; .
\end{align}
\end{prop}
Proposition \ref{Prop: Approach2VarOpt} combined with the $t_m = 1$ GRJ quadrature are sufficient to provide the desired upper bounds on Petz-{\Renyi} divergences. 
\begin{restatable}[Approach 2]{thm}{TheoremApptwo}
\label{Theorem: App2}
    Let $H$ be a finite-dimensional Hilbert space, and let $\rho, \sigma$ be two positive semi-definite operators on $H$ such that $\rho <\!\!<\sigma$. Then, for any $\alpha \in (0,1)$ and $m \in \mathbb{N}$, the m-point GRJ quadrature with endpoint $t_m = 1$ generates nodes $t_1,\dots, t_{m} \in \left(0,1\right]$ and weights $w_1,\dots, w_{m} >0$ such that
    \begin{eqnarray}
         \operatorname{Tr}\left[\rho^{\alpha}\sigma^{1-\alpha}\right] \geq  C_{m} + \frac{\sin(\alpha \pi )}{\pi}\sum_{i=1}^{m} \frac{w_{i}}{t_i} 
        \inf_{Z} \left(\operatorname{Tr}\left[\rho\left(Z+Z^\dagger + \left(1-t_i\right)Z^\dagger Z\right)\right]+t_i\operatorname{Tr}\left[\sigma Z Z^\dagger\right]\right) \; ,
    \end{eqnarray}
    where $C_m = \operatorname{Tr}\left[\rho\right] \left(1 + \frac{\sin(\alpha \pi )}{\pi}\sum_{i=1}^{m} \frac{w_{i}}{t_i} \right)$. The inequality is reversed for $\alpha \in (1,2)$. Moreover, as $m \to \infty$, the right-hand-side converges to $\operatorname{Tr}\left[\rho^{\alpha}\sigma^{1-\alpha}\right]$.
\end{restatable}
\begin{rem} 
Theorem \ref{Theorem: App2} also generates polynomials $P_m$ and $Q_m$ such that
\begin{eqnarray}
    \operatorname{Tr}\left[\rho^{\alpha}\sigma^{1-\alpha}\right] \approx  C_m + \frac{\sin(\alpha \pi )}{\pi} \inf_{Z_1, Z_2,\dots} \operatorname{Tr}\left[\rho P_{m}(Z_1,Z_2,\dots)\right]+\operatorname{Tr}\left[\sigma Q_{m}(Z_1,Z_2,\dots)\right] \; .
\end{eqnarray}    
    However, unlike Theorem  \ref{Theorem: App1Part1}, $P_m$ is no longer generally positive semi-definite. Moreover, Theorem  \ref{Theorem: App2} assumes $\rho <\!\!<\sigma$. As such, it cannot directly be used when we consider device-independent advantage distillation in Section \ref{Section: AdvantageDistill}.

    To achieve the reverse bounds, one has to include the end node $t_1 = 0$~\cite[Theorem 1]{faust2023rational}. Here, one has to consider a sum of the form
    \begin{align}
    \sum_{\substack{j,k \ \mathrm{s.t.} \\ q_k,p_j >0}}  q_k \left[ \frac{ \sin \left(\alpha\pi\right)}{\pi }\frac{p_j-q_k}{q_k}  \right] \left| \braket{\psi_{j} | \phi_{k}}\right|^2 &=   \frac{ \sin \left(\alpha\pi\right)}{\pi }\operatorname{Tr}\left[ \sigma  \rho^{0}\right] -  \frac{ \sin \left(\alpha\pi\right)}{\pi }\operatorname{Tr}\left[\rho  \sigma^{0}\right] \\
    &=  \frac{ \sin \left(\alpha\pi\right)}{\pi }\operatorname{Tr}\left[ \sigma  \rho^{0}\right] -  \frac{ \sin \left(\alpha\pi\right)}{\pi }\operatorname{Tr}\left[\rho \right]\; ,
\end{align}
where the last line follows from the fact that $\rho <\!\!<\sigma$. Moreover, to achieve upper bounds on $\operatorname{Tr}\left[\rho^{\alpha}\sigma^{1-\alpha}\right]$ for $\alpha \in (0,1)$ and lower bounds for $\alpha \in (1,2)$, it is sufficient to apply the bound
\begin{align}
    \operatorname{Tr}\left[ \sigma  \rho^{0}\right] \leq \operatorname{Tr}\left[ \sigma\right] \; .
\end{align}
\end{rem}

\section{Improved DI Protocol Key Rates using the Generalized R\'enyi Entropy Accumulation Theorem} \label{Section: GREATBounds}
\subsection{Introduction}
We consider the standard set-up for DIQKD (also basically usable for \term{device-independent randomness expansion}, DIRE~\cite{Liu_2021}) with two parties, Alice and Bob, who each have access to a measurement device with multiple measurement settings. Apart from the set of possible inputs and outputs, the devices are not further characterized.
\begin{algorithm}[H]
\floatname{algorithm}{}
\caption{Set-up}
\begin{tabular}{lll} 
 Parties: &Alice &Bob\\
 Inputs: & $\mathcal{X} = \{0,\dots,M_{A}\}$  & $\mathcal{Y} = \{0,\dots,M_{B}\}$\\
 Key-Inputs: & $\mathcal{X}_{K} \subset \mathcal{X}$  & $\mathcal{Y}_{K} \subset \mathcal{Y}$\\
 Test-Inputs: & $\mathcal{X}_{T} \subset \mathcal{X}$  & $\mathcal{Y}_{T} \subset \mathcal{Y}$\\
 Outputs for Inputs $X,Y$: & $\mathcal{A}_{X} = \{0,\dots,m_{A,X}\}$  &$\mathcal{B}_{Y} = \{0,\dots,m_{B,Y}\}$\\
\end{tabular}
\end{algorithm} 
Each round of the protocol can be categorized into one of two groups; key-generation and test rounds. The former is used to create a raw key, using so-called key-generating measurements. Rounds of the latter constitute a type of non-local test which, if passed, verify that that the measurement outcomes are sufficiently random to an adversary, Eve. For each test round, $i$, the test data, $\bar{C}_i$, depends on the inputs $X_i,Y_i$ and outputs $A_i,B_i$, i.e. there exists a function $\bar{C}_i = f(A_i,B_i, X_i, Y_i)$.
One can additionally set $C_i = \perp$ for key-generating rounds, which simply means that no test data has been gathered. In this case, one can parameterize the test data for any round via a function
\begin{equation} \label{FunctionTestData}
    \bar{C}_i = f(A_i,B_i, X_i, Y_i,T_i) \; ,
\end{equation}
where $T_i \in \{0,1\}$ is a (random) bit that determines whether the $i$'th round will be a test- or key-generating round. Alice and Bob then verify that the frequency distribution on the stored $C_i$ values lies within some set of accepted distributions, $S_\Omega$. In the most general case, the data stored during test rounds is the combined inputs and outputs, and therefore 
\begin{eqnarray}
    \bar{C}_i \in \{ \perp\} \cup \{ (A,B,X,Y)|X \in \mathcal{X}_{T},Y \in \mathcal{Y}_{T}, A \in \mathcal{A}_{X},B \in \mathcal{B}_{Y}  \}
    \; .
\end{eqnarray}
Alternatively, one can use a simple ``pass/fail" condition for every test round, using a function $w(A,B,X,Y)$. Here, Alice and Bob pass the test if $w(A,B,X,Y)=1 $, fail if $w(A,B,X,Y)=0 $, and the data that can be stored in $\bar{C}_i$ reduces to
\begin{eqnarray}
    \bar{C}_i \in \{ \perp, 0,1\} \; .
\end{eqnarray}


\begin{algorithm}[H]
\floatname{algorithm}{}
\caption{Parameters \& Notation}
\begin{tabular}{lll} 
 Rounds: & $n$  & \\
 Prob.\ of test round: & $\gamma$  & \\
 Input $i$'th Round: & $X_i, Y_i$  & \\
 Output $i$'th Round: & $A_i, B_i$  & \\
Test/Key $i$'th Round: & $T_i$  & \\
  Test Data $i$'th round: & $\bar{C}_i$  & \\
 All Outputs: & ${X}_{1}^{n}, {Y}_{1}^{n}$  & \\
 All Inputs: & ${A}_{1}^{n}, {B}_{1}^{n}$  & \\
  All Test/Key Bits: & ${T}_{1}^{n}$  & \\
    All Test Data: & $\bar{{C}}_{1}^{n}$  & \\
\end{tabular}
\end{algorithm} 

\begin{algorithm}[H]
\floatname{algorithm}{}
\caption{Cryptographic Protocol}
\begin{algorithmic}[1] 
\State For all rounds, $i \in \left[n\right]$:
\begin{algsubstates}
\State Alice and Bob generate a common random bit $T_i$, such that $\operatorname{P}\left(T_i=0\right) = 1-\gamma$ and $\operatorname{P}\left(T_i=1\right) = \gamma$.
\State If $T_i=0$, both parties will choose 
key-generating inputs $\left(X_i,Y_i\right) \in \mathcal{X}_{K} \times \mathcal{Y}_{K}$ according to some distribution, obtaining outputs $A_i$ and $B_i$, respectively. They then set $ \bar{C}_i = \perp$. 
\State If $T_i=1$, both parties will choose test inputs $\left(X_i,Y_i\right) \in \mathcal{X}_{T} \times \mathcal{Y}_{T}$ according to some distribution, obtaining outputs $A_i$ and $B_i$, respectively. They set $\bar{C}_i$ according to the specified function of $\left(A_i,B_i,X_i,Y_i\right)$.\protect\footnotemark
\end{algsubstates}
\State Parameter estimation (a.k.a.~acceptance test): The protocol aborts if the observed frequency distribution $\freq_{\bar{c}_1^n}$ lies outside some predetermined set $S_\Omega$.
\State Classical postprocessing: Alice and Bob perform some additional classical operations, including for instance error correction in the case of DIQKD, ending with a privacy amplification step to generate their final keys.
\end{algorithmic}
\label{alg:OriginalProtocol}
\end{algorithm}
\footnotetext{For strict accuracy, in the physical protocol Alice and Bob should only announce the values required to compute $\bar{C}_i$ \emph{after} all the quantum states have been sent and measured, to prevent Eve updating her attack based on these announcements. The description here is more regarding a ``virtual'' process obtained by commuting some of these measurements and announcements; see~\cite{Arnon-Friedman:2018aa,TSB+22} for further discussion. Alternatively, this restriction can be avoided by working under the model of the generalized EAT, but this introduces a slight change of {\Renyi} parameters in the subsequent bounds, and we hence avoid it here for ease of presentation.}
In the above protocol, for simplicity, we have assumed Alice and Bob can jointly generate the bits $T_i$. This can be achieved using some initial seed randomness; for more detailed discussion, refer to~\cite{Liu_2021,Bhavsar2023} for the case of DIRE and~\cite{TSB+22} for the case of DIQKD. Furthermore, in the case of DIQKD, in order to compute the value of $\bar{C}_i$, at least one of Alice and/or Bob needs to know all the values $\left(A_i,B_i,X_i,Y_i\right)$ in test rounds. For the purposes of our analysis we shall suppose this is simply achieved by having Bob publicly announce $Y_i$ together with a value
\begin{align}\label{eq:testB}
\bar{B}_i = 
\begin{cases}
0 & \text{if } T_i = 0\\
B_i & \text{if } T_i = 1
\end{cases},
\end{align}
where we suppose the value $0$ lies in the alphabet of $B_i$ (if it is not, just replace it in the above with some fixed value from that alphabet), to slightly reduce the effects of some dimension-dependent bounds we discuss later. 
This allows Alice to compute $\bar{C}_i$ on her side and decide whether to accept or abort.
Slightly better finite-size keyrates can be obtained by instead implementing somewhat elaborate error-correction procedures~\cite{NDN+22}, but we do not consider them here --- our goal is just to compare our bounds on some ``basic'' smooth min-entropy terms against previous EAT results, and these variants do not affect this comparison.

After $n$ rounds, the global state is of the form $\rho_{{A}_{1}^{n} {B}_{1}^{n}
\bar{{C}}_{1}^{n}{X}_{1}^{n}{Y}_{1}^{n}{T}_{1}^{n}\mathbf{E}}$, which can be seen as a convex sum
\begin{eqnarray}
    \rho_{{A}_{1}^{n} {B}_{1}^{n}
\bar{{C}}_{1}^{n}{X}_{1}^{n}{Y}_{1}^{n}{T}_{1}^{n}\mathbf{E}} = p_{\Omega}\rho_{|\Omega} + \left(1-p_{\Omega}\right)\rho_{|\bar{\Omega}} \; ,
\end{eqnarray}
where $\rho_{|\Omega}$ contains all terms for which the distribution of the data stored in $\bar{{C}}_{1}^{n}$ lies in the set of accepted distributions, $S_{\Omega}$. 
The length of secure key that can be produced by the final privacy amplification step will depend on the smooth min-entropy\footnote{As discussed previously, recent work has shown~\cite{arx_GLT+22} that better keyrates are obtained by directly applying the {\Renyi} entropy in a privacy amplification theorem derived in~\cite{arx_Dup21}, but for ease of comparison to earlier EAT security proofs, here we just work with smooth min-entropy.} of Alice and/or Bob's raw outputs conditioned on the side-information that Eve holds just before that step. More specifically, for DIRE the quantity of interest would just be $H_{\operatorname{min}}^\epsilon \left({A}_{1}^{n}{B}_{1}^{n}|
{X}_{1}^{n}{Y}_{1}^{n}{T}_{1}^{n}\mathbf{E}\right)_{\rho_{|\Omega}}$ (assuming we use both Alice and Bob's outputs for key generation). For DIQKD the situation is more complex, but basically it suffices to bound $H_{\operatorname{min}}^\epsilon \left({A}_{1}^{n}|
{X}_{1}^{n}{Y}_{1}^{n}{T}_{1}^{n}\mathbf{E}\right)_{\rho_{|\Omega}}$ (or $H_{\operatorname{min}}^\epsilon \left({A}_{1}^{n}|\bar{{B}}_{1}^{n}{X}_{1}^{n}{Y}_{1}^{n}{T}_{1}^{n}\mathbf{E}\right)_{\rho_{|\Omega}}$, depending on the protocol) and then perform some additional computations to account for public announcements and/or error correction. Here we focus on explaining how to bound $H_{\operatorname{min}}^\epsilon \left({A}_{1}^{n}{B}_{1}^{n}|
{X}_{1}^{n}{Y}_{1}^{n}{T}_{1}^{n}\mathbf{E}\right)_{\rho_{|\Omega}}$, and then briefly sketch out how to handle the DIQKD version instead.

To do so, we can relate the smooth min-entropy to {\Renyi} entropies as follows.
First, for technical reasons we need to consider another variant of {\Renyi} entropy referred to as the \term{sandwiched {\Renyi} entropy}, which we denote as $\widetilde{H}^\uparrow_\alpha$; however, we will not need to make use of its actual definition in this work, so we just defer to~\cite{Tomamichel2015QuantumIP} for an explanation of its definition and properties (we use the same notation $\widetilde{H}^\uparrow_\alpha$ for it as in that book). For this entropy, we have for all $\alpha \in \left(1,2\right]$:
\begin{align}
    H_{\operatorname{min}}^{\epsilon}\left({A}_{1}^{n}{{B}}_{1}^{n}|
{X}_{1}^{n}{Y}_{1}^{n}{T}_{1}^{n}\mathbf{E}\right)_{\rho_{|\Omega}} 
&\geq \widetilde{H}_{\alpha}^{\uparrow}\left({A}_{1}^{n}{{B}}_{1}^{n}|
{X}_{1}^{n}{Y}_{1}^{n}{T}_{1}^{n}\mathbf{E}\right)_{\rho_{|\Omega}} - \frac{\log (2/\epsilon^{2})}{\alpha - 1} \label{Eq: RenyitoHmineps} \\
&= \widetilde{H}_{\alpha}^{\uparrow}\left({A}_{1}^{n}{{B}}_{1}^{n}\bar{{C}}_{1}^{n}|
{X}_{1}^{n}{Y}_{1}^{n}{T}_{1}^{n}\mathbf{E}\right)_{\rho_{|\Omega}} - \frac{\log (2/\epsilon^{2})}{\alpha - 1}  \label{Eq: InterimGREATStep}\\ 
\text{(optionally)} &\geq \widetilde{H}_{\alpha}^{\uparrow}\left({A}_{1}^{n}\bar{{C}}_{1}^{n}|
{X}_{1}^{n}{Y}_{1}^{n}{T}_{1}^{n}\mathbf{E}\right)_{\rho_{|\Omega}} - \frac{\log (2/\epsilon^{2})}{\alpha - 1} \; , \label{Eq: FinalGREATStep}
\end{align}
where the first line is \cite[Prop.~6.5 and Eq.~(6.92)]{Tomamichel2015QuantumIP} (originally shown in~\cite{Mull13}), the second line holds because $\bar{{C}}_{1}^{n}$ can be ``projectively reconstructed'' from ${A}_{1}^{n}{B}_{1}^{n}
{X}_{1}^{n}{Y}_{1}^{n}{T}_{1}^{n}$ in the sense described in~\cite[Lemma~B.7]{DFR20}, and the last line holds because classical registers have non-negative contributions to entropy~\cite[Lemma~5.3]{Tomamichel2015QuantumIP}.
That last inequality is not strictly necessary, as the {\Renyi} EAT can directly be applied to Eq.\ (\ref{Eq: InterimGREATStep}); however, we apply it in this work to obtain a more straightforward comparison to previous results in~\cite{Arnon-Friedman:2018aa,Liu_2021} (which also basically excluded the entropy contributions from Bob's outputs, as they relied on a closed-form expression from~\cite{pironio2009deviceindependent} for Alice's entropy only). The resulting SDPs in our approach are also simpler when considering this version.

The {\Renyi} EAT from~\cite[Theorem~5.1 and Lemma~6.1]{arx_AHT24} then gives a lower bound on $\widetilde{H}_{\alpha}^{\uparrow}\left({A}_{1}^{n}\bar{{C}}_{1}^{n}|
{X}_{1}^{n}{Y}_{1}^{n}{T}_{1}^{n}\mathbf{E}\right)_{\rho_{|\Omega}}$, where the first-order contribution only depends on a single-round quantity. This single round can be described by a mapping 
\begin{equation} \label{Eq: SingleRoundMapping}
\begin{aligned}
\mathcal{M}: \mathcal{H}_{Q_AQ_BE} &\to \mathcal{H}_{A\bar{C}XYTE} \\
    \sigma_{Q_AQ_BE}  &\mapsto  \sum_{x,y,t}  \operatorname{P}\left(X=x,Y=y,T=t\right)\mathcal{M}(\sigma)_{A\bar{C}E| X=x,Y=y,T=t } \otimes\ketbra{x,y,t}{x,y,t} \; ,
\end{aligned}
\end{equation}
with 
\begin{align}
    \mathcal{M}(\sigma)_{A\bar{C}E| X=x,Y=y,T=t } &:= \sum_{a,b} \ketbra{a,\bar{c}}{a,\bar{c}} \otimes \operatorname{Tr}_{Q_AQ_B}\left[\left(M_{a|x} \otimes N_{b|y} \otimes \mathds{1}_{E}\right) \sigma_{Q_AQ_BE} \right]  \; ,
\end{align}
where, $\bar{c}$ is determined from the classical registers via Eq.\ (\ref{FunctionTestData}), and $\{M_{a|x}\}$, $\{N_{b|y}\}$ are a family of POVM's. 
The bound from~\cite[Theorem~5.1 and Lemma~6.1]{arx_AHT24} on the {\Renyi} entropy takes the simple form
\begin{eqnarray}\label{eq:GREATbnd}
    \widetilde{H}_{\alpha}^{\uparrow}\left({A}_{1}^{n}\bar{{C}}_{1}^{n}|
{X}_{1}^{n}{Y}_{1}^{n}{T}_{1}^{n}\mathbf{E}\right)_{\rho_{|\Omega}} \geq n h_{\alpha} 
- \frac{\alpha}{\alpha-1} \log\frac{1}{p_\Omega} \; ,
\end{eqnarray}
where\footnote{Strictly speaking, in that work, the bound on the right-hand-side of Eq.~\eqref{eq:GREATbnd} is also in terms of sandwiched {\Renyi} divergences rather than Petz-{\Renyi} divergences. However, since the former is upper bounded by the latter~\cite{Tomamichel2015QuantumIP}, the bound we present here is also immediately valid.} 
\begin{align} \label{Eq: DefinitionhAlpha}
    h_{\alpha} = \inf_{q \in S_\Omega} \inf_{\sigma \in S_{=}(Q_AQ_BE)} \left( \frac{1}{\alpha-1}D\left(q \middle\Vert \mathcal{M}(\sigma)_{\bar{C}}\right)-\sum_{\bar{c}}q(\bar{c})D_{\alpha}\left(\mathcal{M}(\sigma)_{AXYTE \land \bar{c}} \middle\Vert \mathds{1}_{A}\otimes \mathcal{M}(\sigma)_{XYTE} \right) \right) \; ,
\end{align}
where in the first term, we interpret $\mathcal{M}(\sigma)_{\bar{C}}$ as a classical probability distribution on $C$ (since $C$ is a classical register).
Device-independent lower bounds on the accumulated entropy can then be achieved by minimizing this quantity over all possible sets of POVM's.

To summarize the above analysis, the final bound we obtained on $H_{\operatorname{min}}^\epsilon \left({A}_{1}^{n}{B}_{1}^{n}|
{X}_{1}^{n}{Y}_{1}^{n}{T}_{1}^{n}\mathbf{E}\right)_{\rho_{|\Omega}}$ is
\begin{align}\label{Eq: FinalHminBound}
H_{\operatorname{min}}^\epsilon \left({A}_{1}^{n}{B}_{1}^{n}|
{X}_{1}^{n}{Y}_{1}^{n}{T}_{1}^{n}\mathbf{E}\right)_{\rho_{|\Omega}} \geq n h_{\alpha} 
- \frac{\alpha}{\alpha-1} \log\frac{1}{p_\Omega} - \frac{\log (2/\epsilon^{2})}{\alpha - 1}, \; 
\end{align}
where the $h_{\alpha}$ term could be bounded by minimizing the expression in Eq.~\eqref{Eq: DefinitionhAlpha} over all POVMs, or relaxing it to various lower bounds that we discuss later. (This is with the understanding that the $B_1^n$ entropy contributions were dropped by invoking the optional inequality in Eq.~\eqref{Eq: FinalGREATStep}, to provide a fairer comparison to the previous works~\cite{Arnon-Friedman:2018aa,Liu_2021} which also dropped those contributions.)

\begin{rem}
To extend the above argument to handle DIQKD, we can first apply analogous arguments to see that the lower bound in Eq.~\eqref{eq:GREATbnd} also holds for\footnote{An alternative to considering $\widetilde{H}_{\alpha}^{\uparrow}\left({A}_{1}^{n}\bar{{B}}_{1}^{n}|
{X}_{1}^{n}{Y}_{1}^{n}{T}_{1}^{n}\mathbf{E}\right)_{\rho_{|\Omega}}$, as discussed in~\cite{murta2018realization}, is to directly consider $\widetilde{H}_{\alpha}^{\uparrow}\left({A}_{1}^{n}\bar{{C}}_{1}^{n}|
{X}_{1}^{n}{Y}_{1}^{n}{T}_{1}^{n}\mathbf{E}\right)_{\rho_{|\Omega}}$ in our subsequent steps. This can potentially yield better bounds if $\bar{C}_i$ has smaller dimension than $\bar{B}_i$, due to the dimension dependence in Eq.~\eqref{eq:DIQKDchain} (when considering $\bar{C}_i$ instead, the dimension dependence can be slightly improved by noting that $\bar{C}_i$ never takes the value $\perp$ in generation rounds; see~\cite[Remark~8.1]{arx_AHT24}). On the other hand, if the protocol anyway announces the $\bar{B}_i$ values, it is more useful to have $\widetilde{H}_{\alpha}^{\uparrow}\left({A}_{1}^{n}|\bar{{B}}_{1}^{n}{X}_{1}^{n}{Y}_{1}^{n}{T}_{1}^{n}\mathbf{E}\right)_{\rho_{|\Omega}} $ rather than $\widetilde{H}_{\alpha}^{\uparrow}\left({A}_{1}^{n}|\bar{{C}}_{1}^{n}{X}_{1}^{n}{Y}_{1}^{n}{T}_{1}^{n}\mathbf{E}\right)_{\rho_{|\Omega}}$ as an intermediate quantity in the subsequent proofs, so the question of which approach is better might have to be resolved on a case-by-case basis.} $\widetilde{H}_{\alpha}^{\uparrow}\left({A}_{1}^{n}\bar{{B}}_{1}^{n}|
{X}_{1}^{n}{Y}_{1}^{n}{T}_{1}^{n}\mathbf{E}\right)_{\rho_{|\Omega}}$.
To relate this to the previously mentioned quantity of interest $H_{\operatorname{min}}^\epsilon \left({A}_{1}^{n}|
{X}_{1}^{n}{Y}_{1}^{n}{T}_{1}^{n}\mathbf{E}\right)_{\rho_{|\Omega}}$ for DIQKD (ignoring some other public announcements), existing EAT security proofs~\cite{Arnon-Friedman:2018aa,TSB+22,arx_GLT+22,arx_CT23} used chain rules for smooth entropies or {\Renyi} entropies to remove the $\bar{{B}}_{1}^{n}$ registers. However, following a similar analysis as in our companion work~\cite{arx_AHT24}, we now point out that a much simpler bound can be obtained if the accept condition includes the condition that the observed frequency of test rounds (i.e.~the frequency of $T_i=1$) is at most $\gamma+\dtol$ for some fixed value $\dtol$.
Letting $\Omega_{T_1^n}$ denote the set of values on $T_1^n$ with nonzero probability in the conditional state $\rho_{|\Omega}$, for such protocols we can then write
\begin{align}
\widetilde{H}_{\alpha}^{\uparrow}\left({A}_{1}^{n}|
{X}_{1}^{n}{Y}_{1}^{n}{T}_{1}^{n}\mathbf{E}\right)_{\rho_{|\Omega}} 
&\geq \widetilde{H}_{\alpha}^{\uparrow}\left({A}_{1}^{n}|\bar{{B}}_{1}^{n}{X}_{1}^{n}{Y}_{1}^{n}{T}_{1}^{n}\mathbf{E}\right)_{\rho_{|\Omega}} \\
&\geq \widetilde{H}_{\alpha}^{\uparrow}\left({A}_{1}^{n}\bar{{B}}_{1}^{n}|
{X}_{1}^{n}{Y}_{1}^{n}{T}_{1}^{n}\mathbf{E}\right)_{\rho_{|\Omega}} - \max_{t_1^n \in \Omega_{T_1^n}} H_{0}(\bar{{B}}_{1}^{n})_{\rho_{|t_1^n}} \\
&\geq \widetilde{H}_{\alpha}^{\uparrow}\left({A}_{1}^{n}\bar{{B}}_{1}^{n}|
{X}_{1}^{n}{Y}_{1}^{n}{T}_{1}^{n}\mathbf{E}\right)_{\rho_{|\Omega}} - (\gamma+\dtol) n \log\dim(\bar{B}_i), \label{eq:DIQKDchain}
\end{align}
where the second line is proven in~\cite[Remark~8.1]{arx_AHT24}
and the third line follows from our above requirement that the accept condition is such that all $t_1^n \in \Omega_{T_1^n}$ have test-round frequency at most $\gamma+\dtol$, and the fact that $\bar{B}_i$ is deterministically set to a fixed value in all generation rounds. 
We can then convert this bound on $\widetilde{H}_{\alpha}^{\uparrow}\left({A}_{1}^{n}|
{X}_{1}^{n}{Y}_{1}^{n}{T}_{1}^{n}\mathbf{E}\right)_{\rho_{|\Omega}} $ to a bound on $H_{\operatorname{min}}^\epsilon \left({A}_{1}^{n}|
{X}_{1}^{n}{Y}_{1}^{n}{T}_{1}^{n}\mathbf{E}\right)_{\rho_{|\Omega}}$ the same way as in Eq.~\eqref{Eq: RenyitoHmineps} (or again, for better keyrates one could directly apply the {\Renyi} privacy amplification theorem of~\cite{arx_Dup21}).

In fact, since for our protocol we anyway publicly announce the $\bar{{B}}_{1}^{n}$ registers for simplicity, it is more expedient (and yields better final bounds) to instead use the fact that the above sequence of calculations was really just a lower bound on $\widetilde{H}_{\alpha}^{\uparrow}\left({A}_{1}^{n}|\bar{{B}}_{1}^{n}{X}_{1}^{n}{Y}_{1}^{n}{T}_{1}^{n}\mathbf{E}\right)_{\rho_{|\Omega}}$, and to use this quantity directly in the security proof.
\end{rem}

We now turn to the set $S_\Omega$ in the accept condition. In this work, we focus on the case where it has the following form.\footnote{For practical applications, we highlight that significant improvements can be obtained by using the extremely sharp binomial-distribution bounds described in~\cite{Liu_2021} (rather than the Chernoff bound used here), and also using different $\dtol$ values for each $\bar{c}$ value, as discussed in~\cite[Remark~7.1]{arx_AHT24}.
However, for the demonstrative examples in this work we will focus on the form described here for simplicity.}  First, we suppose that the honest protocol is IID and each round produces some distribution $\vect{q}_\mathrm{hon}$ on the alphabet $\{0,1,\perp\}$. Then we take some value $\dtol>0$ and set $S_\Omega$ to be the set of distributions $\vect{q}$ such that
\begin{align}\label{eq:acceptbox}
\forall \bar{c} \in \{0,1,\perp\}, \quad 
\left|q(\bar{c}) - q_\mathrm{hon}(\bar{c})\right| \leq \dtol
.
\end{align}
For this accept condition to be ``reasonable'' in a protocol, we should choose $\delta$ such that the probability of the honest behavior aborting is at most some small value $\eps^\mathrm{com}$ (often called the \term{completeness} parameter). In this work, we use the simple option of noting that by applying the Chernoff bound to the honest IID behavior, we have (writing $q_\mathrm{hon}(\text{not-}\bar{c})$ to denote the probability of \emph{not} getting $\bar{c}$ under distribution $q_\mathrm{hon}$):
\begin{align}
\forall \bar{c} \in \{0,1,\perp\}, \quad 
\Pr_{\bar{C}_1^n}\left[\left|\freq_{\bar{C}_1^n}(\bar{c}) - q_\mathrm{hon}(\bar{c})\right|> \dtol \right] \leq 2 e^{-\frac{\dtol^2 n}{3\min\{ q_\mathrm{hon}(\bar{c}) , q_\mathrm{hon}(\text{not-}\bar{c}) \}}} \leq 2 e^{-\frac{\dtol^2 n}{3\gamma}},
\end{align}
where for the first inequality we used the fact that $\left|\freq_{\bar{c}_1^n}(\bar{c}) - q_\mathrm{hon}(\bar{c})\right| = \left|\freq_{\bar{c}_1^n}(\text{not-}\bar{c}) - q_\mathrm{hon}(\text{not-}\bar{c})\right|$ to choose the smaller value out of $\{ q_\mathrm{hon}(\bar{c}) , q_\mathrm{hon}(\text{not-}\bar{c}) \}$ when applying the Chernoff bound, and to get the second inequality we noted that for $\bar{c} = \perp$ we have $q_\mathrm{hon}(\text{not-}\bar{c}) = \gamma$ and for the other $\bar{c}$ values we have $q_\mathrm{hon}(\bar{c}) \leq \gamma$.
With this, by simply applying the union bound we see that to achieve an upper bound of $\eps^\mathrm{com}$ on the honest abort probability, it suffices to take
\begin{align}\label{eq:dtolbox}
\dtol = \sqrt{\frac{3 \gamma}{n} \ln \frac{6}{\eps^\mathrm{com}}}.
\end{align}

We highlight however that the previous EAT versions in~\cite{Liu_2021,DF19} were not applied to $S_\Omega$ of the form~\eqref{eq:acceptbox}, because most of the constraints in that condition are not ``active'' in the bounds they derive (due to some properties of \emph{crossover min-tradeoff functions} defined in those works). Instead, they simply used a one-sided constraint on a single term, namely
\begin{align}\label{eq:accept1sided}
q(0) \leq q_\mathrm{hon}(0) + \dtol.
\end{align}
For $S_\Omega$ of this form, since there is only ``one condition'' for the protocol to accept, one can choose a smaller value of $\dtol$ than the formula~\eqref{eq:dtolbox} above (while maintaining the same $\eps^\mathrm{com}$ value); namely, by applying the Chernoff bound for a one-sided condition we find it suffices to take
\begin{align}\label{eq:dtol1sided}
\dtol = \sqrt{\frac{3 \gamma}{n} \ln \frac{1}{\eps^\mathrm{com}}}.
\end{align}
In this work, when generating data points for the previous EAT versions for comparison purposes, we used the above choices~\eqref{eq:accept1sided}--\eqref{eq:dtol1sided} for $S_\Omega$ and $\dtol$ to give a ``fairer'' comparison (since using $S_\Omega$ of the form~\eqref{eq:acceptbox} would not allow those versions to ``fully exploit'' the constraints).\footnote{Alternatively, we could also have used this choice of $S_\Omega$ when generating data points based on our new methods. However, we found that 
our methods usually (though not always) give better bounds when using $S_\Omega$ of the form~\eqref{eq:acceptbox}.
We stress however that from an applied perspective, this is still a fair comparison since we achieve the same operational parameter $\eps^\mathrm{com}$ in both the scenarios we are comparing.}

\subsection{Device-Independent Lower Bounds on the Single-Round Contribution}
\subsubsection{Device-Independent Approach Towards Directly Bounding \texorpdfstring{$h_{\alpha}$}{Single-Round Term}}\label{Section: SingleRoundCalculation}

Between two parties, the measurement devices for a single round of the protocol can be described by two families of POVM's, $\{M_{a|x}\}$ and $\{N_{b|y}\}$. Here, $\{M_{a|x}\}$ describes Alice's device, and $\{N_{b|y}\}$ describes Bob's. The goal of this section is to provide device-independent lower bounds on the single-round contribution, $h_{\alpha}$, i.e. $\inf_{\{M_{a|x}\}, \{N_{b|y}\}} h_{\alpha}$, where
\begin{align} \label{Eq: hhatSecion4.2}
      h_{\alpha} = \inf_{q \in S_\Omega} \inf_{\sigma \in S_{=}(Q_AQ_BE)} \left( \frac{1}{\alpha-1}D\left(q \middle\Vert \mathcal{M}(\sigma)_{\bar{C}}\right)-\sum_{\bar{c}}q(\bar{c})D_{\alpha}\left(\mathcal{M}(\sigma)_{AXYTE \land \bar{c}} \middle\Vert \mathds{1}_{A}\otimes \mathcal{M}(\sigma)_{XYTE} \right) \right) \; .
\end{align}
Due to the data-processing inequality, one may assume that $\sigma_{Q_AQ_BE}$ is always pure, i.e. the adversary holds a purification of Alice's and Bob's systems. 
For any tripartite input state $\sigma_{Q_AQ_BE}$, the single round mapping, $\mathcal{M}$, from Eq.\ (\ref{Eq: SingleRoundMapping}) outputs a state
\begin{align}
    \mathcal{M}(\sigma)_{A\bar{C}XYTE} = \sum_{x,y,t}  \operatorname{P}\left(X=x,Y=y,T=t\right)\mathcal{M}(\sigma)_{A\bar{C}E| X=x,Y=y,T=t } \otimes\ketbra{x,y,t}{x,y,t} \; ,
\end{align}
where $\operatorname{P}\left(X=x,Y=y,T=t\right)$ is determined by the protocol and 
\begin{align} \label{Eq: OutputStateForFixedInputs}
    \mathcal{M}(\sigma)_{A\bar{C}E| X=x,Y=y,T=t } = \sum_{a,b}\ketbra{a,\bar{c}}{a,\bar{c}} \otimes \operatorname{Tr}_{Q_AQ_B}\left[\left(M_{a|x} \otimes N_{b|y} \otimes \mathds{1}_{E}\right) \sigma_{Q_AQ_BE} \right] \; .
\end{align}
In Eq.\ (\ref{Eq: OutputStateForFixedInputs}), the value of the classical register $\bar{C}$ is determined by $A,B,X,Y$. Using the convention from Eq.\ (\ref{FunctionTestData}), we set $\bar{c} = f(a,b,x,y,t)$. Our first lemma in Section \ref{Section: SingleRoundCalculation} simplifies the expression for $h_{\alpha}$, by combining both terms.
\begin{restatable}{lem}{SinglehAlphaExpressionviaRelEntropy}
\label{Lemma: SinglehAlphaExpressionviaRelEntropy}
    For all $\alpha \in (1,2)$,
   \begin{eqnarray} \label{Eq: simptargetexpression}
      \frac{1}{\alpha-1}D\left(q \middle\Vert p\right) = \frac{1}{\alpha-1}D\left(q \middle\Vert \mathcal{M}(\sigma)_{\bar{C}}\right)-\sum_{\bar{c}}q(\bar{c})D_{\alpha}\left(\mathcal{M}(\sigma)_{AXYTE \land \bar{c}} \middle\Vert \mathds{1}_{A}\otimes \mathcal{M}(\sigma)_{XYTE} \right)  \; ,
\end{eqnarray}
where 
\begin{equation} \label{Eq: pbarcDefn}
    p(\bar{c}) := \operatorname{Tr}\left[\mathcal{M}(\sigma)_{AXYTE \land \bar{c}}^{\alpha} \cdot \mathds{1}_{A} \otimes \mathcal{M}(\sigma)_{XYTE}^{1-\alpha} \right]  \; .
\end{equation}
\end{restatable}
\begin{rem}
In this section, our goal will be to apply the techniques from Section \ref{Section: TightBoundsPetzRenyiVariational} to bound $p(\bar{c})$ and subsequently ``linearize" $p(\bar{c})$ using the NPA hiearchy from \cite{Pironio_2010}. The convexity property of the relative entropy~\cite[Corollary~4.3]{Tomamichel2015QuantumIP} ensures that the ensuing optimization problem remains convex.
\end{rem}
In order to provide lower bounds on 
   \begin{eqnarray} \label{Eq: OptimizationProblemSimplifiedExpression}
      \inf_{\{M_{a|x}\}, \{N_{b|y}\}}\inf_{q \in S_\Omega} \inf_{\sigma \in S_{=}(Q_AQ_BE)}  \frac{1}{\alpha-1}D\left(q \middle\Vert p\right) \; ,
\end{eqnarray}
using the NPA hierarchy, we upper bound each individual $p(\bar{c})$, using Theorem \ref{Theorem: App2}.\footnote{We could alternatively have used Theorem \ref{Theorem: App1Part2}.}

\begin{restatable}{lem}{LemmaFirstNonComOpt}
\label{Lemma:FirstNonComOpt}
For all $m \in \mathbb{N}$, there exists a choice of $t_1,\dots, t_{m} \in \left(0,1\right]$ and $w_1,\dots, w_{m} >0$ such that
 \begin{align}
  p(\bar{c}) \leq \sup_{Z_1,Z_2,\dots}  \operatorname{Tr}\left[\mathcal{M}(\sigma)_{AXYTE \land \bar{c}}P(Z_1,Z_2,\dots)\right]+\operatorname{Tr}\left[\mathds{1}_{A}\otimes\mathcal{M}(\sigma)_{XYTE}Q(Z_1,Z_2,\dots)\right]\; ,
\end{align}
where 
\begin{align}
    P(Z_1,Z_2,\dots) &:= 1 + \frac{\sin(\alpha \pi )}{\pi}\sum_{i=1}^{m} \frac{w_{i}}{t_i} \left( \mathds{1} + Z_{i}+Z_{i}^{\dagger} + \left(1-t_i\right)Z_{i}^{\dagger} Z_{i}\right)\\
    Q(Z_1,Z_2,\dots) &:= \frac{\sin(\alpha \pi )}{\pi}\sum_{i=1}^{m} w_i Z_iZ_{i}^{\dagger} \; .
\end{align}
\end{restatable}
Note that the $Z_i$ operators, which appear in Lemma \ref{Lemma:FirstNonComOpt}, currently act on multiple Hilbert spaces, most of them classical registers. Similar to \cite{brown2023deviceindependent}, one can reduce the action of each $Z_i$ to Eve's quantum side-information, at the cost of generating more such non-commuting variables. 
\begin{restatable}{prop}{PropFirstNonComOpt}
\label{Prop:FirstNonComOpt}
For every value $\bar{c}$,
\begin{equation}\label{Eq: ExtendedPTerm}
    p(\bar{c}) = \sum_{\substack{x ,y,t \ \mathrm{s.t.} \\\operatorname{P}\left(\bar{c}|x,y,t\right)>0}
    } \operatorname{P}\left(X=x,Y=y,T=t\right)\operatorname{Tr}\left[\mathcal{M}(\sigma)_{AE \land \bar{c}|X=x,Y=y,T=t}^{\alpha} \cdot \mathds{1}_{A} \otimes \mathcal{M}(\sigma)_{E}^{1-\alpha} \right]
 \; .
\end{equation}
Moreover, for all $m \in \mathbb{N}$, there exists a choice of $t_{1},\dots, t_{m} \in \left(0,1\right]$ and $w_{1},\dots, w_{m} >0$ such that every
$\operatorname{Tr}\left[\mathcal{M}(\sigma)_{AE \land \bar{c}|X=x,Y=y,T=t}^{\alpha} \cdot \mathds{1}_{A} \otimes \mathcal{M}(\sigma)_{E}^{1-\alpha} \right]$ can be upper bounded by
 \begin{align} \label{Eq: XYTConditionalPUpperBounds}
\sup_{ Z_{a,i}^{x,y,t}}  \sum_{\substack{a \ \mathrm{s.t.} \\ \operatorname{P}\left(a|\bar{c},x,y,t\right)>0}} \operatorname{Tr}\left[\sigma_{Q_AQ_BE}\left( \sum_{\substack{b \ \mathrm{s.t.} \\ \bar{c} = f(a,b,x,y,t)}} M_{a|x} \otimes N_{b|y} \otimes P(Z_{a,1}^{x,y,t},Z_{a,2}^{x,y,t},\dots) + \mathds{1}_{Q_AQ_B} \otimes Q(Z_{a,1}^{x,y,t},Z_{a,2}^{x,y,t},\dots)\right)\right]\; .
\end{align}
For any $(x,y,t)$, the sums are over all $a$ and $b$ which generate the test data $\bar{c}$, and 
\begin{align}
    P(Z_{a,1}^{x,y,t},Z_{a,2}^{x,y,t},\dots) &:= 1 +\frac{\sin(\alpha \pi )}{\pi}\sum_{i=1}^{m} \frac{w_{i}}{t_{i}} \left(\mathds{1} +Z_{a,i}^{x,y,t}+Z_{a,i}^{x,y,t \dagger} + \left(1-t_i\right)Z_{a,i}^{x,y,t \dagger} Z_{a,i}^{x,y,t}\right)\\
    Q(Z_{a,1}^{x,y,t},Z_{a,2}^{x,y,t},\dots) &:= \frac{\sin(\alpha \pi )}{\pi}\sum_{i=1}^{m} w_i Z_{a,i}^{x,y,t}Z_{a,i}^{x,y,t \dagger} \; .
\end{align}
\end{restatable}
\begin{rem}
    For every value $\bar{c}$, one generates a family of operators, $Z_{a,i}^{x,y,t}$, that now solely act on Eve's Hilbert space. If one instead considers $Eq.\ (\ref{Eq: InterimGREATStep})$, which also includes the classical $B$ register for DIRE, one will instead generate a slightly larger set of operators, $Z_{a,b,i}^{x,y,t}$, where $a,b$ go over all possible values such $\bar{c} = f(a,b,x,y,t)$.
\end{rem}
We mentioned at the beginning of Section \ref{Section: SingleRoundCalculation} that one may assume $\sigma_{Q_AQ_BE}$ is pure, without losing generality. As such, after replacing $p(\bar{c})$ with the polynomials generated by Proposition \ref{Prop:FirstNonComOpt}, Eq.\ (\ref{Eq: OptimizationProblemSimplifiedExpression}) can be lower bounded via the NPA hierarchy. In doing so, the NPA hierarchy will generate variables, $\{y_i\}_{i \in N}$, and express each $p(\bar{c})$ as a linear function of them \cite{Pironio_2010}. Moreover, the resulting optimization problem remains convex, due to the convexity of the relative entropy. The precision of the resulting optimization problem is, in part, dependent on the order of the NPA relaxation.
\begin{restatable}{theorem}{FirstConvexOptBound}
\label{FirstConvexOptBound}
    For all $\alpha \in (1,2)$ and $m \in \mathbb{N}$, let $\{y_i\}_{i \in N}$ be a family of variables of size $N$ that is generated by an NPA relaxation. 
     Then there exists a family of linear functions,
    \begin{eqnarray}
        f_{\bar{c}} \left( y_1,\dots, y_{N}\right) := \sum_{i=1}^{N} a_{i,\bar{c}}y_{i} \; ,
    \end{eqnarray}
    and a matrix, $M(y)$, with entries that are linear in $\{y_i\}_{i \in N}$,
    such that 
    \begin{align} \label{Eq: OptProbFirstConvexOptBound}
      \inf_{\{M_{a|x}\}, \{N_{b|y}\}}\inf_{q \in S_\Omega} \inf_{\sigma \in S_{=}(Q_AQ_BE)}  \frac{1}{\alpha-1}D\left(q \middle\Vert p\right) 
\end{align}
is lower bounded by the convex optimization problem
\begin{equation}
\begin{aligned}
\inf_{q , \{y_i\} } \quad &\frac{1}{\alpha-1}D\left(q \middle\Vert f\right) \\
\textrm{s.t.} \quad & y_1 = 1 \\
& M(y) \succeq 0 \\
& q \in S_\Omega \; .
\end{aligned}
\end{equation}
\end{restatable}
\subsubsection{Reducing the Number of Non-Commuting Variables} \label{Section: VariableReductionWorseBounds}
Using Proposition \ref{Prop:FirstNonComOpt}, every value $\bar{c}$ generates a new set $\{Z_{a,i}^{x,y,t} \}$ of non-commuting operators. Even for $\bar{c} \in \{\perp, 0,1\}$, the total number of operators can often be large and Theorem \ref{FirstConvexOptBound} then generates a convex optimization problem that has too many variables to be easily tractable. To deal with this issue, we introduce some simpler bounds that require fewer operators. Specifically, the following simplified bound on $h_{\alpha}$ was derived in~\cite[Lemma~5.1]{arx_AHT24}, 
for all $\alpha \in (1,\infty)$:
\begin{align}
      h_{\alpha} \geq \inf_{q \in S_\Omega} \inf_{\sigma \in S_{=}(Q_AQ_BE)} \left( \frac{1}{\alpha-1}D\left(q \middle\Vert \mathcal{M}(\sigma)_{\bar{C}}\right)-\sum_{\bar{c}}q(\bar{c})H_{\alpha}\left( A|\bar{C}= \bar{c},XYTE\right)_{\mathcal{M}(\sigma)} \right) \; . \label{Eq: secondtargetoptprob}
\end{align}

We further simplify the optimization problem, by bounding
\begin{eqnarray}
    H_{\alpha}\left( A|\bar{C}= \bar{c},XYTE\right)_{\mathcal{M}(\sigma)} \geq 0
\end{eqnarray}
for all $\bar{c} \neq \perp$.
To see why one shouldn't expect this approximation to greatly affect the final bound, we note that if $\bar{C}$ is just a copy of all inputs and outputs during test rounds, then we anyway have
\begin{eqnarray}
     H_{\alpha}\left( A|\bar{C}= \bar{c},XYTE\right)_{\mathcal{M}(\sigma)} = 0
\end{eqnarray}
for all $\bar{c} \neq \perp$. Moreover, to bound $h_{\alpha}$, this term would include the additional pre-factor $q(\bar{c})$, which scales as $\gamma$ for $\bar{c} \neq \perp$. For protocols with sufficiently small fractions of test rounds, removing these terms will therefore not greatly affect the final result. 

For the remaining term, $ H_{\alpha}\left( A|\bar{C}= \perp,XYTE\right)_{\mathcal{M}(\sigma)}$, we note that $\bar{C}= \perp$ is equivalent to setting the test bit to $T=0$. Moreover, one may remove the classical register $Y$ without changing the value of the entropic quantity. As such, one has that
\begin{eqnarray}
     H_{\alpha}\left( A|\bar{C}= \perp,XYTE\right)_{\mathcal{M}(\sigma)} = H_{\alpha}\left( A|T=0,XE\right)_{\mathcal{M}(\sigma)} \; ,
\end{eqnarray}
where $T=0$ ensures that $X$ will be sampled from the set of key-generating inputs according to some pre-set distribution.

After approximating the {\Renyi} divergences in Eq.\ (\ref{Eq: hhatSecion4.2}) with these bounds, one needs to ensure that the remaining optimization problem stays convex. One easily achieves this by bounding the pre-factor $q(\perp)$ in front of the {\Renyi} entropy with its lowest possible value within the acceptance set $S_\Omega$, which we shall denote as $q_{min}\left(\perp\right)$. The remaining optimization problem is then given by 
\begin{equation} \label{Eq: NonTightLowerBoundhAlpha}
    \inf_{\{M_{a|x}\}, \{N_{b|y}\}}\inf_{q \in S_\Omega} \inf_{\sigma \in S_{=}(Q_AQ_BE)}  \frac{1}{\alpha-1}D\left(q \middle\Vert \mathcal{M}(\sigma)_{\bar{C}}\right) + q_{min}(\perp)H_{\alpha}\left( A|T=0,XE\right)_{\mathcal{M}(\sigma)} \; ,
\end{equation}
and the {\Renyi} entropy will be bounded by a non-commuting polynomial optimization problem. 
\begin{restatable}{prop}{PropSecondNonComOpt}
\label{Prop: SecondNonComOpt}
For all $\alpha \in \left(1,2\right)$, $H_{\alpha}\left( A|T=0, XE\right)_{\mathcal{M}(\sigma)}$ is equal to 
\begin{equation} \label{Eq: EntropyCondOnClassicalStates}
    \frac{1}{1-\alpha}\log_{2}\left(\sum_{x \in \mathcal{X}_{K}} \operatorname{P}\left(X=x|T=0\right)\operatorname{Tr}\left[\mathcal{M}(\sigma)_{AE|X=x}^{\alpha} \cdot \mathds{1}_{A} \otimes \mathcal{M}(\sigma)_{E}^{1-\alpha} \right]
     \right) \; .
\end{equation}
Moreover, for all $m \in \mathbb{N}$, there exists a choice of $t_{1},\dots, t_{m} \in \left(0,1\right]$ and $w_{1},\dots, w_{m} >0$ such that $\operatorname{Tr}\left[\mathcal{M}(\sigma)_{AE|X=x}^{\alpha} \cdot \mathds{1}_{A} \otimes \mathcal{M}(\sigma)_{E}^{1-\alpha} \right]$ is less than or equal to 
 \begin{align} \label{Eq: XConditionalPUpperBounds2}
\sup_{ Z_{a,i}^{x}} \sum_{a} \operatorname{Tr}\left[\sigma_{Q_AQ_BE}\left(  M_{a|x} \otimes \mathds{1}_{Q_B} \otimes P(Z_{a,1}^{x},Z_{a,2}^{x},\dots) + \mathds{1}_{Q_AQ_B} \otimes Q(Z_{a,1}^{x},Z_{a,2}^{x},\dots)\right)\right]\; ,
\end{align}
for all $x \in \mathcal{X}_{K}$, and 
\begin{align}
    P(Z_{a,1}^{x},Z_{a,2}^{x},\dots) &:= 1 +\frac{\sin(\alpha \pi )}{\pi}\sum_{i=1}^{m} \frac{w_{i}}{t_{i}} \left(\mathds{1} +Z_{a,i}^{x}+Z_{a,i}^{x \dagger} + \left(1-t_i\right)Z_{a,i}^{x \dagger} Z_{a,i}^{x}\right)\\
    Q(Z_{a,1}^{x},Z_{a,2}^{x},\dots) &:= \frac{\sin(\alpha \pi )}{\pi}\sum_{i=1}^{m} w_i Z_{a,i}^{x}Z_{a,i}^{x \dagger} \; .
\end{align}
\end{restatable}
Compared to Proposition \ref{Prop:FirstNonComOpt}, Proposition \ref{Prop: SecondNonComOpt} generates fewer non-commuting variables. This effect is especially noticeable, when Alice has only one set of key-generating measurements. Analogous to Theorem \ref{FirstConvexOptBound}, one can bound Eq.\ (\ref{Eq: NonTightLowerBoundhAlpha}) using the NPA hierarchy.
\begin{restatable}{theorem}{SecondConvexOptBound}
\label{Theorem: SecondConvexOptBound}
    For all $\alpha \in (1,2)$ and $m \in \mathbb{N}$, let $\{y_i\}_{i \in N}$ be a family of variables of size $N$ that is generated by an NPA relaxation. Then there exists a set of linear functions,
    \begin{eqnarray}
        f_{\bar{c}} \left( y_1,\dots, y_{N}\right) &:=& \sum_{i=1}^{N} a_{i,\bar{c}}y_{i} \\
        g \left( y_1,\dots, y_{N}\right) &:=& \sum_{i=1}^{N} b_{i}y_{i} \; ,
    \end{eqnarray}
    and a matrix, $M(y)$, with entries that are linear in $\{y_i\}_{i \in N}$,
    such that 
    \begin{eqnarray} \label{Eq: OptProbSecondConvexOptBound}
      \inf_{\{M_{a|x}\}, \{N_{b|y}\}}\inf_{q \in S_\Omega} \inf_{\sigma \in S_{=}(Q_AQ_BE)}  \frac{1}{\alpha-1}D\left(q \middle\Vert \mathcal{M}(\sigma)_{\bar{C}}\right)+q_{min}(\perp)H_{\alpha}\left( A|T=0,XE\right)_{\mathcal{M}(\sigma)}
\end{eqnarray}
is lower bounded by the convex optimization problem
\begin{equation}
\begin{aligned}
\inf_{q , \{y_i \}} \quad &\frac{1}{\alpha-1}D\left(q \middle\Vert f\right) + \frac{q_{min}}{1-\alpha}\log_{2}\left(g\right)\\
\textrm{s.t.} \quad & y_1 = 1 \\
& M(y) \succeq 0 \\
& q \in S_\Omega \; .
\end{aligned}
\end{equation}
\end{restatable}
Alternatively, another simplified bound was derived in~\cite[Lemma~5.2]{arx_AHT24}; namely, for any $\alpha,\alpha',\alpha''\in (1,\infty)$ satisfying $\frac{\alpha}{\alpha-1} = \frac{\alpha'}{\alpha'-1} + \frac{\alpha''}{\alpha''-1}$, we have
\begin{align} 
h_\alpha &\geq \inf_{\{M_{a|x}\}, \{N_{b|y}\}}\inf_{q \in S_\Omega} \inf_{\sigma \in S_{=}(Q_AQ_BE)}   \frac{\alpha''}{\alpha''-1}D\left(q \middle\Vert \mathcal{M}(\sigma)_{\bar{C}}\right) + H_{\alpha'}\left( A\bar{C}|XYTE\right)_{\mathcal{M}(\sigma)} \label{Eq: simplified3Renyi} \\
&\geq \inf_{\{M_{a|x}\}, \{N_{b|y}\}}\inf_{q \in S_\Omega} \inf_{\sigma \in S_{=}(Q_AQ_BE)}  \frac{\alpha''}{\alpha''-1}D\left(q \middle\Vert \mathcal{M}(\sigma)_{\bar{C}}\right) + H_{\alpha'}\left( A|XTE\right)_{\mathcal{M}(\sigma)}
\; , \label{Eq: lastsimplified3Renyi} 
\end{align}
where the relaxation in the second line is again not strictly necessary but does make the final SDPs smaller.
This bound can then be relaxed to a non-commuting polynomial optimization using the same techniques as above.
Note that as compared to Eq.~\eqref{Eq: NonTightLowerBoundhAlpha}, in this bound there is no prefactor on the entropy term, and furthermore we do not condition on $\bar{C}= \perp$ and/or $T=0$.\footnote{For the CHSH game in particular, we have the convenient property that $H_{\alpha'}\left( A|XTE\right)_{\mathcal{M}(\sigma)}$ can still be lower bounded using the same non-commuting polynomial as for $H_{\alpha'}\left( A|T=0,XE\right)_{\mathcal{M}(\sigma)}$, due to ``symmetries'' of the CHSH game that ensure the same lower bound holds conditioned on any choice of inputs $(X,Y)$ (this fact was implicitly used in past works such as~\cite{Arnon-Friedman:2018aa,Liu_2021}). Technically, it was noted in~\cite{SGP+21} that this bound can in fact be fairly loose for some protocols; however, for the protocol we consider here (where fixed inputs are used in generation rounds, and test rounds have fairly small probability) it does not make much difference.} This means we have essentially preserved the ``entropy contributions from test rounds'', unlike the above bound; however, this comes at the cost of worse {\Renyi} parameters (see~\cite[Sec.~5.2]{arx_AHT24} for more detailed analysis). Hence we compute results for both of these versions to see which performs better in practice. 

In the next section, we will be comparing the performance of these bounds against the previous EAT version from~\cite{Liu_2021}. However, since that bound was based on von Neumann entropy (using a closed-form bound from~\cite{pironio2009deviceindependent}), one may be interested in seeing what can be achieved with the {\Renyi} EAT performs if one only has a method to bound the von Neumann entropy. In the interests of achieving as fair a comparison as possible to~\cite{Liu_2021} under these circumstances, we do so via the following (not necessarily optimal) chain of inequalities. First, recalling from Eq.~\eqref{Eq: InterimGREATStep} that for DIRE it suffices to bound\footnote{For DIQKD, we would perform analogous arguments for $\widetilde{H}_{\alpha}^{\uparrow}\left({A}_{1}^{n}{\bar{B}}_{1}^{n}\bar{{C}}_{1}^{n}|
{X}_{1}^{n}{Y}_{1}^{n}{T}_{1}^{n}\mathbf{E}\right)_{\rho_{|\Omega}}$; note that this yields the same final values for the various dimension-dependent terms because $\dim(\bar{B}_i) = \dim(B_i)$ in our case.} $\widetilde{H}_{\alpha}^{\uparrow}\left({A}_{1}^{n}{{B}}_{1}^{n}\bar{{C}}_{1}^{n}|
{X}_{1}^{n}{Y}_{1}^{n}{T}_{1}^{n}\mathbf{E}\right)_{\rho_{|\Omega}}$, we apply~\cite[Lemma~5.2]{arx_AHT24} to obtain a bound similar to Eq.~\eqref{Eq: simplified3Renyi} except with $\widetilde{H}_{\alpha'}\left( A{B}\bar{C}|XYTE\right)_{\mathcal{M}(\sigma)}$ in place of $\widetilde{H}_{\alpha'}\left( A{B}|XYTE\right)_{\mathcal{M}(\sigma)}$ (this is to allow a slight sharpening in the dimension-dependent terms in the next step). We then relax this {\Renyi} entropy to a von Neumann entropy via
\begin{align}
\widetilde{H}_{\alpha'}\left( A{B}\bar{C}|XYTE\right)_{\mathcal{M}(\sigma)} &= \widetilde{H}_{\alpha'}\left( A{B}|XYTE\right)_{\mathcal{M}(\sigma)} \\
&\geq H\left( A{B}|XYTE\right)_{\mathcal{M}(\sigma)} - (\alpha'-1)\log^2\left(1+2 \dim(A{B}) \right)
 \; , \label{Eq: RenyitovN}
\end{align}
where the first line is~\cite[Lemma~B.7]{DFR20} and the second line is a continuity bound in the {\Renyi} parameter,~\cite[Lemma~B.9]{DFR20} (assuming $\alpha'\in (1,1+{1}/{\log\left(1+2\dim(A{B})\right)})$). 

Putting things together, we have the following bound (for $\alpha',\alpha''>1$ satisfying $\frac{\alpha}{\alpha-1} = \frac{\alpha'}{\alpha'-1} + \frac{\alpha''}{\alpha''-1}$ and $\alpha'\in (1,1+{1}/{\log\left(1+2\dim(A{B})\right)})$), though it is not entirely optimal as we shall shortly discuss:
\begin{align} \label{Eq: HminviavN}
H_{\operatorname{min}}^{\epsilon}\left({A}_{1}^{n}{{B}}_{1}^{n}|
{X}_{1}^{n}{Y}_{1}^{n}{T}_{1}^{n}\mathbf{E}\right)_{\rho_{|\Omega}} &\geq n h_{\mathrm{vN},\alpha''} - n(\alpha'-1)\log^2\left(1+2 \dim(A{B}) \right) \nonumber\\
&\qquad - \frac{\log (2/\epsilon^{2})}{\alpha - 1}
- \frac{\alpha}{\alpha-1} \log\frac{1}{p_\Omega} \; ,
\end{align}
where 
\begin{align}
h_{\mathrm{vN},\alpha''} &= \inf_{\{M_{a|x}\}, \{N_{b|y}\}}\inf_{q \in S_\Omega} \inf_{\sigma \in S_{=}(Q_AQ_BE)}  \frac{\alpha''}{\alpha''-1}D\left(q \middle\Vert \mathcal{M}(\sigma)_{\bar{C}}\right) + H\left( A{B}|XYTE\right)_{\mathcal{M}(\sigma)} \nonumber\\
&\geq \inf_{\{M_{a|x}\}, \{N_{b|y}\}}\inf_{q \in S_\Omega} \inf_{\sigma \in S_{=}(Q_AQ_BE)}  \frac{\alpha''}{\alpha''-1}D\left(q \middle\Vert \mathcal{M}(\sigma)_{\bar{C}}\right) + H\left( A|XTE\right)_{\mathcal{M}(\sigma)}
 \; ,
\end{align}
where again the second inequality is not strictly necessary but yields a fairer comparison to~\cite{Liu_2021}. In particular, for this work we then bound $H\left( A|XTE\right)_{\mathcal{M}(\sigma)}$ in terms of the CHSH value using the closed-form bound from~\cite{pironio2009deviceindependent}, which reduces the above expression to a $3$-parameter minimization problem (over two entries of $q$ and the CHSH value of the state $\mathcal{M}(\sigma)$; see~\cite[Sec.~7.3]{arx_AHT24} for details). We then evaluate its value using heuristic optimization algorithms. (Unlike our preceding calculations, the resulting value is a less rigorously certified lower bound, since we only used heuristic algorithms here rather than convex optimization algorithms. However, we consider this bound to be just a rough reference point for comparison to our preceding calculations, and hence having a rigorously certified bound on it is a lower priority. Also, as the optimization only involves a small number of parameters and is individually convex in each parameter, there is less risk of a heuristic algorithm failing to find the true minimum.)

There are various improvements that could be made to the above bound. For instance, observe that in this work as well as~\cite{Liu_2021}, in the end we only invoked a bound on $H\left( A|XTE\right)_{\mathcal{M}(\sigma)}$ rather than $H\left( A{B}|XYTE\right)_{\mathcal{M}(\sigma)}$ (since these computations were based on applying the~\cite{pironio2009deviceindependent} bound on the former). Given this, the dimension dependence in Eq.~\eqref{Eq: RenyitovN} can be slightly improved by dropping the ${B}$ register \emph{before} applying~\cite[Lemma~B.9]{DFR20}, so the latter would only depend on $\dim(A)$. 
However, we choose not to do this here for a fairer comparison to~\cite{Liu_2021}, since the EAT bound in that work involves $\dim(A{B})$ rather than $\dim(A)$.
Other potential improvements include using the more elaborate continuity bound from~\cite[Corollary~IV.2]{DF19} (which should be sharper than~\cite[Lemma~B.9]{DFR20} in this context, as discussed in~\cite[Sec.~5.2]{arx_AHT24}), or starting from the~\cite[Lemma~5.1]{arx_AHT24} simplification rather than the~\cite[Lemma~5.2]{arx_AHT24} simplification (our results in the next section suggest that the former is better). We leave this for future work, since in the next section we compute some numerical examples and show that even this suboptimal bound can already outperform the~\cite{Liu_2021} bound.

\subsection{DIQKD Bounds for the Extended CHSH Set-up} \label{Section: DIQKDCHSHBounds}
We consider the case where Alice has two measurement inputs, $\mathcal{X}=\{0,1\}$, with binary outputs and Bob has three, $\mathcal{Y}=\{0,1,2\}$. Moreover, the key-generating and test measurements are given by $\mathcal{X}_{K}=\{0\}$, $\mathcal{X}_{T}=\{0,1\}$, $\mathcal{Y}_{K}=\{2\}$, $\mathcal{Y}_{T}=\{0,1\}$.
\begin{algorithm}[H]
\floatname{algorithm}{}
\caption{Extended CHSH Set-up}
\begin{tabular}{lll} 
 Parties: &Alice &Bob\\
 Inputs: & $\mathcal{X} = \{0,1\}$  & $\mathcal{Y} = \{0,1,2\}$\\
 Key-Inputs: & $\mathcal{X}_{K} = \{0\}$  & $\mathcal{Y}_{K} = \{2\}$\\
 Test-Inputs: & $\mathcal{X}_{T} = \{0,1\}$  & $\mathcal{Y}_{T} = \{0,1\}$\\
 Outputs: & $\mathcal{A} = \{0,1\}$  &$\mathcal{B} = \{0,1\}$\\
\end{tabular}
\end{algorithm} 

We capture the CHSH game within our protocol as follows, the same way as in~\cite{Arnon-Friedman:2018aa}.
The win condition during test rounds is that the inputs $X,Y$ and respective outputs $A,B$ satisfy
\begin{equation}
    A \oplus B = X \cdot Y \; .
\end{equation}
Whenever this occurs in a test round, one sets $\bar{C}=1$, otherwise one sets $\bar{C}=0$. If it is a generation round, then $\bar{C}=\perp$ as discussed previously. In the example we present, the set of accepted distributions is as follows. We take the honest behavior to be IID and to achieve an expected winning probability of $\omega_\mathrm{hon} = 0.8$ during test rounds. For any given test-round fraction, $\gamma$, and number of rounds, $n$, the corresponding tolerance level we allow on the test data outputs, $\bar{C}$, is given by Eq.~(\ref{eq:dtolbox}). 
For any {\Renyi} parameter $\alpha>1$, Theorem \ref{Theorem: SecondConvexOptBound} generates a convex optimization problem that can be used to lower bound $H_{\operatorname{min}}^\epsilon \left(A_{1}^{n} B_{1}^{n}|
X_{1}^{n}Y_{1}^{n}
T_{1}^{n}\mathbf{E}\right)_{\rho_{|\Omega}}$. We then find the optimal {\Renyi} parameter numerically. Using $m=6$ nodes, we generated the variational optimization problem via Theorem \ref{Theorem: App2}. With regards to the NPA hierarchy, we used NPA Level 1, and included the Local Level 1 set~\cite[Remark 2.6]{brown2023deviceindependent}. We used the same number of nodes and relaxation level when calculating the bounds from Eq.\ (\ref{Eq: lastsimplified3Renyi}). For simplicity, we set $\alpha^\prime = \alpha^{\prime \prime}$ for the latter two approaches and subsequently optimized numerically over $\alpha$. Figure \ref{fig:EATRate} contains the generated data and compares it to~\cite{Liu_2021}.

\begin{figure}[h]\centering

\definecolor{ao}{rgb}{0.0,0.0,1.0}
\definecolor{airforceblue}{rgb}{0.36,0.54,0.66}
\definecolor{aliceblue}{rgb}{0.94,0.97,1.00}
\definecolor{aqua}{rgb}{0.7,0.97,1.00}

\colorlet{MyColorOne}{yellow!50}

\newcommand{\lightercolor}[3]{
    \colorlet{#3}{#1!#2!white}
}

\newcommand{\darkercolor}[3]{
    \colorlet{#3}{#1!#2!orange}
}

\lightercolor{MyColorOne}{50}{MyColorOneLight}
\darkercolor{MyColorOne}{50}{MyColorOneDark}
\begin{tikzpicture} [scale = 0.4]

\begin{axis}[%
width=6.028in,
height=4.754in,
at={(0in,0.642in)},
scale only axis,
xmode=log,
xmin=5000000,
xmax=10000000000,
xminorticks=true,
xlabel={Rounds of the protocol, n},
ymin=0,
ymax=0.4,
ylabel={Lower Bound on $H_{\operatorname{min}}^{\epsilon}\left(A_{1}^{n}B_{1}^{n}|
X_{1}^{n}Y_{1}^{n}T_{1}^{n}\mathbf{E}\right)_{\rho_{|\Omega}}/n$},
label style={font=\Large},
legend pos = north east,
axis background/.style={fill=white}
]

\addplot[name path=S1, domain = 5000000:10000000000, line width=1pt, black, dash pattern={on 12pt }, forget plot] {0.346112};
\addplot [color=red, mark=o, mark options={solid, red}, forget plot]
  table[row sep=crcr]{%
10000000	0.140755047\\ 
30000000	0.221706660495319\\ 
100000000	0.274809708464348\\ 
300000000	0.302987829289995\\ 
1000000000	0.320752327343884\\ 
3000000000	0.329997118521976\\ 
10000000000	0.335685027934794\\ 
};
\addplot [color=red, dashed, mark=o, mark options={solid, red}, forget plot]
  table[row sep=crcr]{%
100000000	0.142076956023595\\ 
300000000	0.223753698471279\\ 
1000000000	0.27732972553353\\ 
3000000000	0.305758876399186\\ 
10000000000	0.323680419260892\\ 
};
\addplot [color=teal, mark=o, mark options={solid, teal}, forget plot]
  table[row sep=crcr]{%
10000000	0.110985295\\
30000000	0.204817329\\
100000000	0.266639553\\
300000000	0.299578484\\
1000000000	0.320394338\\
3000000000	0.331196889\\
10000000000	0.337925187\\
};
\addplot [color=teal, dashed, mark=o, mark options={solid, teal}, forget plot]
  table[row sep=crcr]{%
100000000	0.111077426\\
300000000	0.204847111\\
1000000000	0.266644146\\
3000000000	0.299585462\\
10000000000 0.320397523\\
};
\addplot [color=magenta, mark=o, mark options={solid, magenta}, forget plot]
  table[row sep=crcr]{%
10000000	0.10053265397342\\
30000000	0.199993950057734\\
100000000	0.260788136814316\\
300000000	0.296405030665986\\
1000000000	0.318730153578763\\
3000000000	0.330256000758258\\
10000000000	0.337411311416491\\
};
\addplot [color=magenta, dashed, mark=o, mark options={solid, magenta}, forget plot]
  table[row sep=crcr]{%
100000000	0.101673427210379\\
300000000	0.200652597803476\\
1000000000	0.264809370640531\\
3000000000	0.296570531853093\\
10000000000	0.318820802698528\\
};

\addplot [color=blue, mark=o, mark options={solid, blue}, forget plot]
  table[row sep=crcr]{%
10000000	0.066617134\\
30000000	0.174604741\\
100000000	0.248162198\\
300000000	0.288390996\\
1000000000	0.314093159\\
3000000000	0.327492571\\
10000000000	0.33585066\\
};
\addplot [color=blue, dashed, mark=o, mark options={solid, blue}, forget plot]
  table[row sep=crcr]{%
100000000	0.071566138\\
300000000	0.177730766\\
1000000000	0.249694918\\
3000000000	0.289271602\\
10000000000	0.314586338\\
};
\end{axis};
\end{tikzpicture}%
\caption{
Comparison of various bounds on $H_{\operatorname{min}}^{\epsilon}\left(A_{1}^{n}B_{1}^{n}|
X_{1}^{n}Y_{1}^{n}T_{1}^{n}\mathbf{E}\right)_{\rho_{|\Omega}}/n$.
Solid lines indicate a test-round fraction of $\gamma = 10^{-2}$, and dashed lines use $\gamma = 10^{-3}$. The best bounds we obtain (shown in red) use Eq.~\eqref{Eq: FinalHminBound} with Theorem \ref{Theorem: SecondConvexOptBound}. The second-best bounds (shown in green) arise from using Eq.~\eqref{Eq: FinalHminBound} with Eq.\ (\ref{Eq: lastsimplified3Renyi}). The third-best bounds (shown in orange) can be achieved with only a bound on single-round von Neumann entropy, from Eq.~\eqref{Eq: HminviavN}. All three approaches improve upon the best bound from previous work, namely~\cite{Liu_2021}, which we evaluated for these scenarios and display as well (in blue). The horizontal, dashed line represents the asymptotic value our results should converge to. For these plots, we impose $\epsilon= 10^{-5}$ and $p_{\Omega}\geq 10^{-5}$ (the latter condition can be imposed from the security proof structure; see e.g.~\cite{Arnon-Friedman:2018aa,Liu_2021,TSB+22}).}
\label{fig:EATRate}
\end{figure}
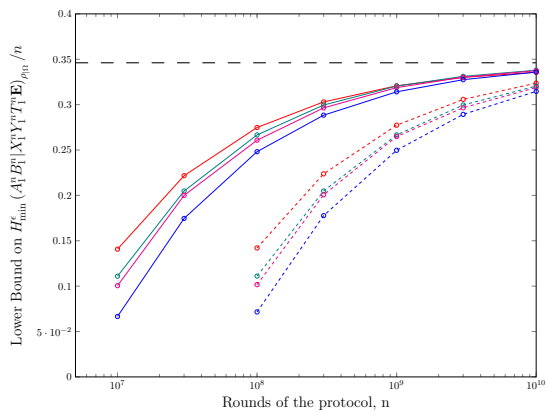

In the above figure, we constrained all the values of $q(\bar{c})$ as described in Eq.~\eqref{eq:acceptbox} when computing the bounds described in this work (i.e.~the red, green and orange curves), using the completeness parameter $\eps^\mathrm{com}=10^{-3}$. However, a noteworthy point is that for the data points shown here in the red curves (i.e.~based on Theorem \ref{Theorem: SecondConvexOptBound}), we found that in fact the bounds are very slightly better if we instead impose a single-sided constraint on a single component $q(0)$ as described in Eq.~\eqref{eq:accept1sided} (with the appropriately reduced $\dtol$ value in Eq.~\eqref{eq:dtol1sided}), though the improvement is very small, on the order of $10^{-3}$ -- $10^{-4}$, depending on the value of $n$. This observation is somewhat surprising, in that the latter approach generally performed worse in previous {\Renyi} EAT security proofs for other protocols (see~\cite[Remark~7.1]{arx_AHT24}), and furthermore we found that this was also the case for the data points shown here in the green and orange curves (i.e.~based on Eq.\ (\ref{Eq: lastsimplified3Renyi}), or based on single-round von Neumann entropy via Eq.~\eqref{Eq: HminviavN}), though in some cases only to a small extent. Hence the question of which approach yields better results appears to be somewhat complicated, though we point out that since the single-sided constraint performed significantly worse in many cases and only gave a very slight improvement in the cases where we found it to be better, a loose guiding principle to simplify applications might be to just always use the version constraining all $q(\bar{c})$ values, where possible.

\section{Device-Independent Advantage Distillation} \label{Section: AdvantageDistill}
\subsection{Introduction}
For simplicity, we consider set-ups, where both parties, Alice and Bob, each have one key-generating measurement, $X=0$ and $Y=0$, and have measurement devices which always output a binary value. Moreover, we let Alice and Bob use all inputs during test rounds. This allows them to estimate the total input-output distribution after the parameter estimation step.
\begin{algorithm}[H]
\floatname{algorithm}{}
\caption{Set-up}
\begin{tabular}{lll} 
 Parties: &Alice &Bob\\
 Inputs: & $\mathcal{X} = \{0,\dots,M_{A}\}$  & $\mathcal{Y} = \{0,\dots,M_{B}\}$\\
 Key-Inputs: & $\mathcal{X}_{K} = \{0\}$  & $\mathcal{Y}_{K} = \{0\}$\\
 Test-Inputs: & $\mathcal{X}_{T} = \mathcal{X}$  & $\mathcal{Y}_{T} = \mathcal{Y}$\\
 Outputs: & $\mathcal{A}_{X} = \{0,1\}$  &$\mathcal{B}_{Y} = \{0,1\}$\\
\end{tabular}
\end{algorithm} 
We consider the repetition-code protocol, which is the most well-known version of advantage distillation~\cite{256484,20.500.11850/144113,20.500.11850/72791,bae2006key}. This means that, after all rounds of the protocol have been concluded, Alice and Bob divide their raw key into blocks of $k$ bits. For each block, we denote Alice's and Bob's output bits by $\mathbf{A}$ and $\mathbf{B}$, respectively. The goal will be to generate (up to) $1$ highly correlated bit from each block. To achieve this, for each block, Alice generates a uniformly random bit $C$. She then sends the message $\mathbf{M} = \mathbf{A} \oplus \left(C,\dots, C\right)$ to Bob via an authenticated channel.\footnote{This means that an adversary can read the message, but cannot alter it without the alerting the honest parties.} Bob decodes this message by adding his output bits to $\mathbf{M}$. If 
\begin{equation} \label{RepetitionCodeAcceptCon}
    \mathbf{B} \oplus \mathbf{M} = \left(C^\prime,\dots, C^\prime\right) \; ,
\end{equation}
for some $C^\prime \in \{0,1\}$, Bob will communicate to Alice that the block is accepted and the parties will replace their block bits, $\mathbf{A}$ and $\mathbf{B}$, with the single bits $C$ and $C^\prime$. If Eq.\ (\ref{RepetitionCodeAcceptCon}) does not hold, Bob will communicate that the block is rejected and both parties will throw away the bits from that block. This process generates highly correlated bits, as it only accepts the block if Bob's initial block bits are equal to Alice's or if every bit is flipped. Since the latter possibility is much less likely, $C$ and $C^\prime$ will almost always be the same. This intuition was formalized into a full security proof (accounting for the DI setting) in e.g.~\cite{Tan_2020,Hahn_2022,stasiuk2022quantum}, as we describe below.

\begin{algorithm}[H]
\floatname{algorithm}{}
\caption{Cryptographic Protocol including Repetition-Code Protocol }
\begin{algorithmic}[1] 
\State For all rounds, $i \in \left[n\right]$:
\begin{algsubstates}
\State Alice and Bob generate a common random bit $T_i$, such that $\operatorname{P}\left(T_i=0\right) = 1-\gamma$ and $\operatorname{P}\left(T_i=1\right) = \gamma$.
\State If $T_i=0$, both parties will choose 
key-generating inputs $\left(X_i,Y_i\right) = (0,0)$, obtaining outputs $A_i$ and $B_i$, respectively. They then set $ \bar{C}_i = \perp$. 
\State If $T_i=1$, both parties will choose test inputs $\left(X_i,Y_i\right) \in \mathcal{X}_{T} \times \mathcal{Y}_{T}$ according to some distribution, obtaining outputs $A_i$ and $B_i$, respectively. They set $\bar{C}_i = \left(A_i,B_i,X_i,Y_i\right)$.
\State (Symmetrization Step) Both parties generate a common uniformly random bit $F_i$, and change their outputs to $\tilde{A}_i = A_i \oplus F_i$ and $\tilde{B}_i = B_i \oplus F_i$, respectively.
\end{algsubstates}
\State Parameter estimation (a.k.a.~acceptance test): The protocol aborts if the observed frequency distribution $\freq_{\bar{c}_1^n}$ lies outside some predetermined set $S_\Omega$.
\State Advantage Distillation Step: Alice and Bob perform the repetition-code protocol on the raw key that is produced from the key-generating rounds.
\State Additional Classical Postprocessing: Alice and Bob perform some additional classical operations, including for instance error correction in the case of DIQKD, ending with a privacy amplification step to generate their final keys.
\end{algorithmic}
\label{alg:ProtocolWithAD}
\end{algorithm} 
In addition to applying advantage distillation, we also assume they include a symmetrization step during each round of the initial protocol. This ensures that the input-output distribution is symmetric, thus simplifying the theoretical analysis. Alice and Bob can achieve this by generating a common uniformly random bit $F$ during each round, and adding it bit-wise to their respective outputs. One way of ensuring that both parties know the value $F$ is for Alice to generate it, and then send it to Bob over an authenticated channel. In such cases, Eve also has knowledge of the value of $F$ as well, and as such will be included in her side-information.

We consider the security of this protocol against collective (IID) attacks, where an adversary, Eve, is restricted to using the same attack every round. In particular, this means that, in every round, Alice, Bob, and Eve share the same tripartite state, $\rho_{Q_AQ_BE}$.\footnote{Here, we again use $Q_A$ and $Q_B$ to denote Alice's and Bob's initial quantum systems. The registers that then contain their classical measurement outputs are denoted by $A$ and $B$, respectively.} Moreover, the measurement devices can be described by a family of POVMs, $\{M_{a|x}\}$, $\{N_{b|y}\}$, where $x \in \{0,\dots, M_{A} \}$ and $y \in \{0,\dots, M_{B} \}$ are Alice's and Bob's measurement inputs and $a, b \in \{0,1 \}$ their corresponding measurement outputs. As such, before the symmetrization step, the input-output distribution is given by
\begin{align}
    \operatorname{P}(ab|xy) = \operatorname{Tr}\left[\left(M_{a|x} \otimes N_{b|y} \otimes \mathds{1}_{E}\right) \rho_{Q_AQ_BE} \right] \; ,
\end{align}
and, for any inputs $x,y$, the post-measurement state is
\begin{equation}
    \sum_{a,b \in \{0,1\}} \operatorname{P}(ab|xy) \ketbra{ab}{ab}_{AB} \otimes \rho_{E|ab} \; , 
\end{equation}
where 
\begin{equation} \label{EveSideInfoAdDist}
    \rho_{E|ab} := \frac{1}{\operatorname{P}(ab|xy)} \operatorname{Tr}_{AB}\left[\left(M_{a|x} \otimes N_{b|y} \otimes \mathds{1}_{E}\right) \rho_{Q_AQ_BE} \right] \; .
\end{equation}
After symmetrization, Alice's and Bob's new outputs, $\tilde{a}$, $\tilde{b}$, will be distributed according to 
\begin{equation}
    \operatorname{P}(\tilde{a}\tilde{b}|xy) = \frac{1}{2} \left[\operatorname{Tr}\left[\left(M_{a|x} \otimes N_{b|y} \otimes \mathds{1}_{E}\right) \rho_{Q_AQ_BE} \right] + \operatorname{Tr}\left[\left(M_{\bar{a}|x} \otimes N_{\bar{b}|y} \otimes \mathds{1}_{E}\right) \rho_{Q_AQ_BE} \right]
     \right]\; ,
\end{equation}
where $\bar{a}=a \oplus 1$ and $\bar{b}=b \oplus 1$. For inputs $x$ and $y$, the new post-measurement state will also include the random bit $F$, and is given by
\begin{equation}
    \sum_{\tilde{a},\tilde{b} \in \{0,1\}} \operatorname{P}(\tilde{a}\tilde{b}|xy) \ketbra{\tilde{a}\tilde{b}}{\tilde{a}\tilde{b}}_{\tilde{A}\tilde{B}} \otimes \rho_{EF|\tilde{a}\tilde{b}} \; , 
\end{equation}
where 
\begin{equation}
    \rho_{EF|\tilde{a}\tilde{b}} := \frac{1}{2}\left[ \rho_{E|ab} \otimes \ketbra{0}{0}_{F} + \rho_{E|\bar{a}\bar{b}} \otimes \ketbra{1}{1}_{F}\right] \; .
\end{equation}
Conditioned on the repetition-code protocol accepting a block, it is very likely that Alice's and Bob's key-generating outputs were identical on the block the protocol was acting on. As such, for each of those rounds, Eve's quantum side-information will most likely either be $\rho_{EF|\tilde{0}\tilde{0}}$ or $\rho_{EF|\tilde{1}\tilde{1}}$. Intuitively, as long as these two states are hard to distinguish from Eve's perspective, she will not know what Alice's and Bob's measurement outcomes were, and advantage distillation will thus generate bits for Alice and Bob, which are heavily correlated, but look sufficiently random to an adversary. Based on this intuition, a sufficient security condition was proven in \cite{stasiuk2022quantum}, which is based on the Quantum Chernoff Coefficient $Q(\rho,\sigma) := \inf_{0<s<1} \operatorname{Tr}\left[ \rho^s \sigma^{1-s}\right]$.
\begin{theorem}[\cite{stasiuk2022quantum}, Theorem~1] \label{Theorem: CBSuffCond}
    A secret key can be generated using the repetition-code protocol if, for a key-generating round,
\begin{equation} \label{Eq: CBOptimizationProblem}
    Q(\rho_{EF|\tilde{0}\tilde{0}},\rho_{EF|\tilde{1}\tilde{1}}) > \frac{\epsilon}{1-\epsilon} \; ,
\end{equation}
where $\epsilon := \operatorname{P}(01|00)+\operatorname{P}(10|00)= \operatorname{P}(\tilde{0}\tilde{1}|00)+\operatorname{P}(\tilde{1}\tilde{0}|00)$ denotes the quantum bit error rate (QBER).
\end{theorem}
Given any observed distribution $\operatorname{P}(\tilde{a}\tilde{b}|xy)$ (in the IID asymptotic limit), one can use Theorem \ref{Theorem: CBSuffCond} to prove a secret key can be generated, in a relatively straightforward way. One simply verifies that there exist no states, $\rho_{ABE}$, and POVMs, 
$\{M_{a|x}\}$, $\{N_{b|y}\}$, which can simultaneously generate $\operatorname{P}(\tilde{a}\tilde{b}|xy)$, and violate Eq.\ (\ref{Eq: CBOptimizationProblem}).
\subsection{Device-Independent Lower Bounds for Advantage Distillation}
Here, our first main result is that the security condition from Theorem \ref{Theorem: CBSuffCond} can be rephrased in terms of the pretty good fidelity, $F_{pg}(\rho,\sigma) := \operatorname{Tr}\left[ \rho^{1/2} \sigma^{1/2}\right]$.
\begin{restatable}{thm}{TheoremPGFCB}
\label{Theorem: PGFCB}
    A secret key can be generated using the repetition-code protocol if, for a key-generating round,
\begin{equation} \label{Eq: PGFOptimizationProblem}
    F_{pg}(\rho_{E|00},\rho_{E|11}) > \frac{\epsilon}{1-\epsilon} \; ,
\end{equation}
where $\epsilon := \operatorname{P}(01|00)+\operatorname{P}(10|00)$ denotes the quantum bit error rate (QBER).
\end{restatable}
This simplifies the problem in two key ways. The states, which appear in Eq.\ (\ref{Eq: PGFOptimizationProblem}), do not contain the symmetrization bit, $F$, and describe Eve's side-information before the symmetrization step. As such, apart from assuming that the input-output distribution is symmetric, the symmetrization step will not further affect the theoretical analysis. Moreover, the definition of $Q(\rho,\sigma)$ contains an extra optimization over a parameter, $s$. By instead considering $F_{pg}(\rho,\sigma)$, this optimization will not be necessary. It's clear from 
\begin{align}
    Q(\rho_{EF|\tilde{0}\tilde{0}},\rho_{EF|\tilde{1}\tilde{1}}) &:= \inf_s \operatorname{Tr}\left[ \rho_{EF|\tilde{0}\tilde{0}}^s \rho_{EF|\tilde{1}\tilde{1}}^{1-s}\right] \\
    &\leq \operatorname{Tr}\left[ \rho_{EF|\tilde{0}\tilde{0}}^{1/2} \rho_{EF|\tilde{1}\tilde{1}}^{1/2}\right] \\
    &= \frac{1}{2}\left(  \operatorname{Tr}\left[ \rho_{E|00}^{1/2} \rho_{E|11}^{1/2}\right] + \operatorname{Tr}\left[ \rho_{E|11}^{1/2} \rho_{E|00}^{1/2}\right] \right)\\
    &= F_{pg}(\rho_{E|00},\rho_{E|11}) \; ,
\end{align}
that Eq.\ (\ref{Eq: PGFOptimizationProblem}) holds whenever the security condition from Theorem \ref{Theorem: CBSuffCond} is satisfied. Conversely, one can use the fact that $\operatorname{Tr}\left[ \rho^s \sigma^{1-s}\right]$ is convex in $s$ \cite{Audenaert_2007} to show that
\begin{align}
     Q(\rho_{EF|\tilde{0}\tilde{0}},\rho_{EF|\tilde{1}\tilde{1}}) &:= \inf_s \operatorname{Tr}\left[ \rho_{EF|\tilde{0}\tilde{0}}^s \rho_{EF|\tilde{1}\tilde{1}}^{1-s}\right] \\
     &= \inf_s \frac{1}{2}\left(  \operatorname{Tr}\left[ \rho_{E|00}^{s} \rho_{E|11}^{1-s}\right] + \operatorname{Tr}\left[ \rho_{E|00}^{1-s} \rho_{E|11}^{s}\right] \right) \\
     &\geq  \inf_s \operatorname{Tr}\left[ \rho_{E|00}^{1/2} \rho_{E|11}^{1/2}\right] \\
     &=  \operatorname{Tr}\left[ \rho_{E|00}^{1/2} \rho_{E|11}^{1/2}\right] \\
     &=: F_{pg}(\rho_{E|00},\rho_{E|11}) \; .
\end{align}
In order for both inequalities to hold, 
\begin{align}
    Q(\rho_{EF|\tilde{0}\tilde{0}},\rho_{EF|\tilde{1}\tilde{1}}) =  F_{pg}(\rho_{E|00},\rho_{E|11}) \; ,
\end{align}
and Theorem \ref{Theorem: PGFCB} is therefore equivalent to Theorem \ref{Theorem: CBSuffCond}.

For any observed (symmetric) distribution, $\operatorname{P}(ab|xy)$, Eve's attack can never achieve a lower value than
\begin{equation} \label{Eq: SimplifiedPGFOptProb}
\begin{aligned}
\inf_{\rho_{ABE},\{M_{a|x}\},\{N_{b|y}\}}  \quad & F_{pg} (\rho_{E|00},\rho_{E|00})\\
\textrm{s.t.} \quad & \operatorname{Tr}\left[\left(M_{a|x} \otimes N_{b|y} \otimes \mathds{1}_{E}\right) \rho_{ABE} \right]  = \operatorname{P}(ab|xy) \; .
\end{aligned}
\end{equation}
If this value satisfies the condition from Theorem \ref{Theorem: PGFCB}, then Alice and Bob will be able to generate a secret key.
\begin{rem}
    We recall that the way the protocol is described, the constraint on the system would a priori be on the post-symmetrization distribution, $\operatorname{P}(\tilde{a}\tilde{b}|xy)$, i.e.\
    \begin{equation}
     \frac{1}{2} \left[\operatorname{Tr}\left[\left(M_{a|x} \otimes N_{b|y} \otimes \mathds{1}_{E}\right) \rho_{Q_AQ_BE} \right] + \operatorname{Tr}\left[\left(M_{\bar{a}|x} \otimes N_{\bar{b}|y} \otimes \mathds{1}_{E}\right) \rho_{Q_AQ_BE} \right]
     \right] = \operatorname{P}(\tilde{a}\tilde{b}|xy) \; . \label{Eq: PostSymConstraint}
\end{equation}
At an informal level, the constraint from Eq.\ (\ref{Eq: SimplifiedPGFOptProb}) is equivalent to this, as one can assume without loss of generality that Eve has already applied the symmetrization step before sending the tripartite state. More precisely, for any feasible state $\rho_{Q_AQ_BE}$ and POVMs $\{M_{a|x}\}$, $\{N_{b|y}\}$, there exists the state 
\begin{eqnarray}
    \rho^{\prime}_{Q_A Q_{\hat{A}} Q_B Q_{\hat{B}} E\hat{E} } := \rho_{Q_AQ_BE} \otimes \ketbra{\psi}{\psi}_{Q_{\hat{A}}Q_{\hat{B}}\hat{E}}
\end{eqnarray}
with $\ket{\psi}_{Q_{\hat{A}}Q_{\hat{B}}\hat{E}}:= \frac{1}{\sqrt{2}} \left[\ket{000}_{Q_{\hat{A}}Q_{\hat{B}}\hat{E}} + \ket{111}_{Q_{\hat{A}}Q_{\hat{B}}\hat{E}} \right]$ and POVMS
\begin{eqnarray}
    M_{a|x}^\prime &:=& \frac{1}{2} \left[ M_{a|x} \otimes \ketbra{0}{0}_{Q_{\hat{A}}} + M_{\bar{a}|x} \otimes \ketbra{1}{1}_{Q_{\hat{A}}} \right] \\
    N_{b|y}^\prime &:=& \frac{1}{2} \left[ N_{b|y} \otimes \ketbra{0}{0}_{Q_{\hat{B}}} + N_{\bar{b}|y} \otimes \ketbra{1}{1}_{Q_{\hat{B}}} \right] \; ,
\end{eqnarray}
which leave the objective function from Eq.\ (\ref{Eq: SimplifiedPGFOptProb}) invariant and satisfy 
\begin{equation}
    \operatorname{Tr}\left[\left(M_{a|x}^\prime \otimes N_{b|y}^\prime \otimes \mathds{1}_{E\hat{E}}\right) \rho^{\prime}_{Q_A Q_{\hat{A}} Q_B Q_{\hat{B}} E\hat{E} } \right]  = \operatorname{P}(\tilde{a}\tilde{b}|xy) \; .
\end{equation}
Conversely, any state $\rho_{Q_AQ_BE}$ and POVMs $\{M_{a|x}\}$, $\{N_{b|y}\}$ that satisfy the constraint from $Eq.\ (\ref{Eq: SimplifiedPGFOptProb})$ (and produce a symmetric output distribution) must also satisfy $Eq.\ (\ref{Eq: PostSymConstraint})$. We may thus, without loss of generality, assume that the initial distribution is symmetric and ignore the symmetrization step.
\end{rem}
Theorem \ref{Theorem: App1Part1} can be used to bound the pretty good fidelity using a variational expression, i.e.\ there exist polynomials $P_m$ and $Q_m$ such that $F_{pg} (\rho_{E|00},\rho_{E|00})$ is greater than or equal to
\begin{align}
      \inf_{Z_{1},Z_2,\dots} \operatorname{Tr}\left[\rho_{Q_AQ_BE} \left[  \left(M_{0|0} \otimes N_{0|0} \otimes P_m(Z_{1},Z_2,\dots) \right)  + \left(M_{1|0} \otimes N_{1|0} \otimes  Q_m(Z_{1},Z_2,\dots)\right)\right] \right] \; .
\end{align}
\begin{rem}
    Each of the polynomials include an additional constant factor $\frac{1}{\operatorname{P}(00|00)}$. This is due to the fact that 
    \begin{align}
        \rho_{E|ii} &= \frac{\operatorname{Tr}_{Q_AQ_B} \left[\rho_{Q_AQ_BE}   \left(M_{i|0} \otimes N_{i|0} \otimes \mathds{1_E}\right)\right]}{P(ii|00)} \\
        &= \frac{\operatorname{Tr}_{Q_AQ_B}\left[\rho_{Q_AQ_BE}   \left(M_{i|0} \otimes N_{i|0} \otimes \mathds{1_E}\right)\right]}{P(00|00)} \; ,
    \end{align}
    where the last line holds due to $P(ab|xy)$ being symmetric.
\end{rem}
By assuming, without loss of generality, that Eve holds a purification, one can again invoke the NPA hierarchy to lower bound Eq.\ (\ref{Eq: SimplifiedPGFOptProb}) via an SDP.

\subsection{Improved Tolerance Bounds} \label{Subsection: ADExplicitSetUps}
We consider 
two DIQKD set-ups, and assume that the measurement outputs are affected by depolarizing noise. This means that, for any target distribution $P_{\text{target}}$ Alice and Bob aim for, the observed distribution will be
\begin{align}
    P(ab|xy) = \left(1-2q\right) P_{\text{target}}(ab|xy) + q/2 \; ,
\end{align}
where $q$ is a depolarization parameter. The 
two target distributions, $P_{\text{target}}$, are described by these set-ups.
\begin{enumerate}
    \item (CHSH) Alice and Bob each have two possible measurement settings. \label{SetUP: 1}
    \begin{itemize}
        \item Target State: $\ket{\psi} = \frac{1}{\sqrt{2}}\left( \ket{00}+ \ket{11}\right)$
        \item Alice's Measurements: $A_0 = Z$, $A_1 = X$
        \item Bob's Measurements: $B_0 = \frac{1}{\sqrt{2}}\left(X+Z\right)$, $B_1 = \frac{1}{\sqrt{2}}\left(X-Z\right)$
    \end{itemize}
    \item (Extended CHSH) Alice has two possible measurement settings, and Bob has three possible measurement settings. \label{SetUP: 2}
    \begin{itemize}
        \item Target State: $\ket{\psi} = \frac{1}{\sqrt{2}}\left( \ket{00}+ \ket{11}\right)$
        \item Alice's Measurements: $A_0 = Z$, $A_1 = X$
        \item Bob's Measurements: $B_0 = Z$, $B_1 = \frac{1}{\sqrt{2}}\left(X+Z\right)$, $B_2 = \frac{1}{\sqrt{2}}\left(X-Z\right)$
    \end{itemize}
\end{enumerate}
Set-up \ref{SetUP: 1} represents the smallest possible set-up which can self-test the target state (via the CHSH game). Here, the target key-generating measurements $A_0$ and $B_0$ do not produce perfectly correlated outcomes for the ideal state. This issue can be resolved by giving one party an extra measurement setting. In Set-up \ref{SetUP: 2}, we provide Bob with this extra measurement setting, which he will use during key-generating rounds.  
In the following, we present our improved bounds and compare them to those achieved by \cite{stasiuk2022quantum}.

   \begin{algorithm}[H] \label{Figure: ADPGFBounds}
\floatname{algorithm}{} 
\caption{Bounds on Noise Tolerance}
\label{Results}
\begin{tabular}{lll} 
 Set-up &Tolerance Bounds From \cite{stasiuk2022quantum} &New Tolerance Bounds\\
 \ref{SetUP: 1} & $q \approx 7.71\%$  & $q \approx 8.78 \%$\\
\ref{SetUP: 2} & $q \approx 8.18\%$  & $q \approx 8.38 \%$\\
\end{tabular}
\end{algorithm} 

We used $m=8$ nodes, and generated the variational optimization problem via Theorem \ref{Theorem: App1Part2}. Moreover, we optimized over each node individually~(See \cite[Remark 2.6]{brown2023deviceindependent}), using NPA Level $5$ for Set-up \ref{SetUP: 1}, and NPA Level $4$ for Set-up \ref{SetUP: 2}.
\section{Discussion} 
In Section \ref{Section: TightBoundsPetzRenyiVariational}, we discussed two variational approaches towards bounding Petz-{\Renyi} divergences. The bounds found in Section \ref{SubSection: App2} have the benefit of being structurally similar to the main results from~\cite{brown2023deviceindependent}, where they provide bounds on the relative divergence. There, in addition to considering finite-dimensional Hilbert spaces, they prove that their results extend to the infinite-dimensional setting~\cite[Section 3]{brown2023deviceindependent}. As such, the results from Section \ref{SubSection: App2} can  straightforwardly be extended to hold for this case as well.


We show how Theorem \ref{Theorem: App2} from Section \ref{SubSection: App2} can be 
applied to provide numerically tractable lower bounds on the smooth min-entropy bounds derived by the {\Renyi} EAT~\cite{arx_AHT24}. However, this approach is not directly applicable when considering DI advantage distillation, as the constraint $\rho <\!\!<\sigma$ no longer generally holds. In principle, one can circumvent this issue as follows. For $\alpha \in (0,1)$, one can consider $\operatorname{Tr}\left[ \rho^\alpha \tau^{1-\alpha}\right]$ instead of $\operatorname{Tr}\left[ \rho^\alpha \sigma^{1-\alpha}\right]$, where $\tau:= \left(1-\delta\right)\sigma + \delta \rho$. By construction, $\rho <\!\!<\tau$ holds and one can then apply Theorem \ref{Theorem: App2}. One practical issue here relates to the convergence of the resulting bound. As an example, if $\rho$ and $\sigma$ are orthogonal, then as $\delta \to 0$, the end node's contribution will also start to diverge. To counteract this, the smaller one chooses $\delta$ the more nodes should to be included, such that one attains a better approximation of $\operatorname{Tr}\left[ \rho^\alpha \tau^{1-\alpha}\right]$ (in particular, additional nodes will lead to a decrease in the end node's weight, thus reducing its effect).  The error that arises from considering $\operatorname{Tr}\left[ \rho^\alpha \tau^{1-\alpha}\right]$ instead of the original quantity of interest can be bounded in a relatively straightforward way. By the H{\"o}lder inequality (See e.g.~\cite[Lemma 3.1]{Tomamichel2015QuantumIP} or~\cite[Corollary IV.2.6]{BhatiaMatrixAnalysis}), it follows that
\begin{align}
  |\operatorname{Tr}\left[ \rho^\alpha \tau^{1-\alpha}\right] - \operatorname{Tr}\left[ \rho^\alpha \sigma^{1-\alpha}\right]|\leq \operatorname{Tr}\left[ \rho\right]^{\alpha} 
  \lVert \tau^{1-\alpha} -\sigma^{1-\alpha} \rVert_{\frac{1}{1-\alpha}} \; ,
\end{align}
Moreover, the right-hand-side can be upper bounded using~\cite[Theorem X.1.3]{BhatiaMatrixAnalysis} to get
\begin{align}
    |\operatorname{Tr}\left[ \rho^\alpha \tau^{1-\alpha}\right] - \operatorname{Tr}\left[ \rho^\alpha \sigma^{1-\alpha}\right]|\leq \operatorname{Tr}\left[ \rho\right]^{\alpha} \lVert \tau -\sigma\rVert_{1}^{1-\alpha} \; .
\end{align}
In Section \ref{Section: AdvantageDistill}, we only consider normalized states, which further simplifies this inequality. After plugging in the definition for $\tau$ and using the triangle inequality, one finds that in this case
\begin{align}
    |\operatorname{Tr}\left[ \rho^\alpha \tau^{1-\alpha}\right] - \operatorname{Tr}\left[ \rho^\alpha \sigma^{1-\alpha}\right]|\leq \left(2\delta\right)^{1-\alpha}  \; .
\end{align}
It should be noted that whenever $\alpha > 1/2$, tighter bounds on the resulting error can by achieved by instead replacing $\rho$ with $\tau := \left(1-\delta\right)\rho + \delta \sigma $, and then applying analogous arguments to above.

The variational bound found in Section \ref{SubSection: App1} seems to be more general. Here, the bounds achieved by Theorem \ref{Theorem: App1Part1} provide tight variational bounds even when $\rho <\!\!<\sigma$ does not hold. As such, one can directly apply it in the context of DI advantage distillation, which is discussed in Section \ref{Section: AdvantageDistill}. Moreover, both Theorem \ref{Theorem: App1Part1} and Theorem \ref{Theorem: App1Part2} from Section \ref{SubSection: App1} generate polynomials which can be made to be positive semi-definite.\footnote{For Theorem \ref{Theorem: App1Part2} the polynomials are either positive or negative semi-definite, depending on whether one includes the additional factor $\sin(\alpha \pi )$ into their definition.} As we discuss later in this section, this property may be helpful when considering other information theoretic problems. Similar to Theorem \ref{Theorem: App2}, the bounds from Theorem \ref{Theorem: App1Part1} and Theorem \ref{Theorem: App1Part2} should also straightforwardly be extendable to the infinite-dimensional case.

In Section \ref{Section: GREATBounds}, we discuss how the variational bounds from Theorem \ref{Theorem: App2} can be applied to the {\Renyi} EAT.\footnote{Using Theorem \ref{Theorem: App1Part2} would also have been equally possible.} The tightest bounds we derive are found in Section \ref{Section: SingleRoundCalculation}, but applying the resulting optimization problem does not seem to be feasible due to the large number of variables that arise from applying the NPA hierarchy. The simpler bounds derived in Section \ref{Section: VariableReductionWorseBounds} are much more tractable. As can be seen in Section \ref{Section: DIQKDCHSHBounds}, there are multiple ways to use the {\Renyi} EAT and attain tighter bounds on the smooth min-entropy than were previously known. Unsurprisingly, Theorem \ref{Theorem: SecondConvexOptBound} seems to provide the best results. 
For the Extended CHSH Set-up considered in Section \ref{Section: DIQKDCHSHBounds}, it should be possible to provide analytical bounds on the relevant {Renyi} entropy terms that appear, similar to~\cite{pironio2009deviceindependent}. It would be of interest to verify how close the achieved numerical results are to the corresponding analytical bounds. We leave this question for future work.

We consider DI advantage distillation in Section \ref{Section: AdvantageDistill}. Here, we derive a security condition based on the pretty good fidelity that is equivalent to the previously known (Quantum Chernoff Condition-based) condition from~\cite{stasiuk2022quantum}. As the pretty good fidelity can be expressed via Petz-{\Renyi} divergences, one can use Theorem \ref{Theorem: App1Part1} to bound achievable noise tolerances via a variational optimization problem. When implementing the numerical simulations, we simplified the complexity of optimization problem by optimizing over each node individually (See \cite[Remark 2.6]{brown2023deviceindependent}).
Even with this reduction, we managed to improve upon the best previously known noise tolerances found in~\cite{stasiuk2022quantum}. In addition, it should be noted that in order to achieve satisfactory bounds on the pretty-good fidelity, we used an NPA relaxation level of at least four for the set-ups that were considered. For the simplest case described in Set-up \ref{SetUP: 1}, we managed to increase the NPA level to five and noticed a slight, but non-negligible, improvement on the resulting noise tolerances. It would be interesting to know if there exist smaller local sets which can be added so that it is possible to reduce the required NPA level without significantly decreasing (or even increasing) the achieved bounds. 


We conclude with the open question of whether or not the tools from this work can be extended to provide variational bounds on a different class of closely related {\Renyi} entropies~\cite[Definition~5.2]{Tomamichel2015QuantumIP}
\begin{align}
    H_{\alpha}^{\uparrow} (A|E)_{\rho}:= \sup_{\sigma_{E} \in S_{=}\left(E\right)} -D_{\alpha} \left(\rho_{AE} || \mathds{1}_{A}\otimes \sigma_{E} \right) \; .
\end{align}
These entropies are known to have a closed expression~\cite{Tomamichel_2014,Sharma_2013}, i.e.\
\begin{align}
    H_{\alpha}^{\uparrow} (A|E)_{\rho} = \frac{\alpha}{1-\alpha}\log \left(\operatorname{Tr}\left[ \left(\operatorname{Tr}_{A} \left[ \rho_{AE}^{\alpha}\right] \right)^{\frac{1}{\alpha}} \right]\right) \; .
\end{align}
For $\alpha \in (1/2,1) \cup (1,2)$, one potential approach towards bounding this is by using the results from Section \ref{Section: TightBoundsPetzRenyiVariational}. To see where the main difficulties lie in this approach, let us consider $\alpha \in (1/2,1)$. By Eq.\ (\ref{Eq: App1Part2LowerBounds}), lower bounds on 
\begin{align}
    \operatorname{Tr}\left[ \left(\operatorname{Tr}_{A} \left[ \rho_{AE}^{\alpha}\right] \right)^{\frac{1}{\alpha}} \right] = \operatorname{Tr}\left[ \left(\operatorname{Tr}_{A} \left[ \rho_{AE}^{\alpha}\right] \right)^{\frac{1}{\alpha}} \cdot \mathds{1}_{E}^{1-\frac{1}{\alpha}}\right]
\end{align}
are achievable, and of the form 
\begin{align} \label{Eq: UpArrowVarProblem}
    \sup_{Z_1,  \dots} C_m^\prime \operatorname{Tr}\left[ \rho_{AE}^{\alpha}\right] + \frac{\sin(\pi/\alpha  )}{\pi}  \operatorname{Tr}\left[ \rho_{AE}^{\alpha} \cdot \mathds{1}_{A} \otimes P_{m}(Z_1,\dots)\right]+\operatorname{Tr}\left[Q_{m}(Z_1,\dots)\right] \; ,
\end{align}
where $P_m,Q_m$ are positive semi-definite and $C_m^\prime$ is a positive constant. Let us first focus on the second term, which needs to be further simplified. Section \ref{Section: TightBoundsPetzRenyiVariational} contains no results that can directly be applied to it. However, here, one should be able to leverage the fact that $P_m$ is positive semi-definite. The resulting approach towards (upper) bounding terms of the form $\operatorname{Tr}\left[ \rho^{\alpha}\cdot \sigma \right]$,
where $\alpha \in (1/2,1)$ and $\rho, \sigma$ are positive semi-definite, should in principle be similar to the methods used in Section \ref{Section: TightBoundsPetzRenyiVariational}. 

The last term in Eq.\ (\ref{Eq: UpArrowVarProblem}) is slightly problematic, as it does not contain any dependence on the initial state $\rho_{AE}$. As such, it is not possible to directly apply the NPA hierarchy from~\cite{Pironio_2010}. The last difficulty we wish to discuss is that, in DIQKD, one would ideally like to minimize $H_{\alpha}^{\uparrow} (A|E)_{\rho}$ over all states that satisfy some distribution constraints. The resulting optimization problem contains both a minimization (over all states and measurements) and a maximization (over the operators $Z_i$). As such, once again, one cannot directly apply the NPA hierarchy. These issues are further highlighted by the fact that, after applying e.g.\  Theorem \ref{Theorem: App1Part1},
\begin{align}
    \operatorname{Tr}\left[ \rho_{AE}^{\alpha}\right]  \geq 
    \frac{\sin(\alpha \pi )}{\pi} \inf_{Z_1^\prime,  \dots}    \operatorname{Tr}\left[ \rho_{AE} \cdot P_{m}^\prime(Z_1^\prime,\dots)\right]+\operatorname{Tr}\left[Q_{m}^\prime(Z_1^\prime,\dots)\right] \; ,
\end{align}
which introduces further \textit{minimizations} over the operators $Z_i^\prime$ and an additional term that does not depend on the initial state.

\section*{Acknowledgements}
E.T.\ conducted research at the Institute for Quantum Computing, at the University of Waterloo, which is supported by Innovation, Science, and Economic Development Canada. Support was also provided by NSERC under the Discovery Grants Program, Grant No. 341495. 
T. H. acknowledges support from the Marshall and Arlene Bennett Family Research Program, the Minerva foundation with funding from the Federal German Ministry for Education and Research and the Israel Science Foundation (ISF), and the Directorate for Defense Research and Development (DDR\&D), grant No. 3426/21.3.
P.B. acknowledges funding from the European Union's Horizon Europe research and innovation programme under the project ``Quantum Secure Networks Partnership'' (QSNP, grant agreement No. 101114043).

The Python code in this work generated the semi-definite programs and convex optimization problems by using the package NCPOL2SDPA~\cite{Wittek_2015}. The resulting optimization problems were solved using MOSEK~\cite{mosek} and CVXPY~\cite{diamond2016cvxpy,agrawal2018rewriting}.
\section*{Code Availability}
The Python code used to generate the data can be found at:
\url{https://github.com/Thomas0501/Petz-Renyi-Entropies}
\section{Appendix} \label{Section: Appendix}
\subsection{Proofs}
\ApponeLemone*
\begin{proof}
    For $x \geq 0$ and $\alpha \in \left(0,1\right)$, it is known that~\cite[Eq.~(V.4)]{BhatiaMatrixAnalysis}
\begin{align} 
    x^{\alpha} = \frac{\sin (\alpha \pi)}{\pi} \int_{0}^{\infty} \frac{x}{x+t}t^{\alpha-1}dt \; .
\end{align}
As such
\begin{align}
    x^{\alpha} &= \frac{\sin (\alpha \pi)}{\pi}\left[ \int_{0}^{1} \frac{x}{x+t}t^{\alpha-1}dt + \int_{1}^{\infty} \frac{x}{x+t}t^{\alpha-1}dt  \right] \\
    &= \frac{\sin (\alpha \pi)}{\pi}\left[ \int_{0}^{1} \frac{x}{x+t}t^{\alpha-1}dt + \int_{0}^{1} \frac{x }{x+z^{-1}}z^{ -(\alpha+1)}dz  \right] \\
    &= \frac{\sin (\alpha \pi)}{\pi}\left[ \int_{0}^{1} \frac{x}{x+t}t^{\alpha-1}dt + \int_{0}^{1} \frac{x }{xz+1} z^{ -\alpha}dz  \right] \; ,
\end{align}
where we used the substitution $z = 1/t$ in the second line. The result then follows after relabeling the variable of integration, $z$, with $t$.
\end{proof}
\BoundingIntviaQuad* 
\begin{proof}
    For any $m$-point quadrature with nodes $t_1,\dots,t_m$ and weights $w_1,\dots, w_m$, one has that 
    \begin{align}
    \int_{0}^{1}f(t)t^{\beta}dt = \sum_{i=1}^{m} w_{i} f(t_i) +E_{m} \; , \label{Eq: ProofIntQuadErrorEq}
\end{align}
where $E_m$ denotes the error or difference between the integral and the corresponding quadrature. Moreover, it follows from~\cite[Eq.\ (8.9.8)]{Hildebrand56} that there exists a $\xi \in (0,1)$ such that
\begin{align} \label{Eq: ProofErrorTerm}
    E_{m} = \frac{\Gamma(m +1)\Gamma(m+\beta +1)^2}{(2m+\beta+1)\Gamma(2m+\beta+1)^2}\frac{m!}{(2m)!} f^{(2m)}\left(\xi\right) \geq 0 \; ,
\end{align}
where the inequality holds due the fact that $f^{(2m)}\left(\xi\right) \geq 0$. As such, any $m$-point GJ quadrature must provide a lower bound on the integral. 

For GRJ quadratures with endpoint $t_1=0$, it follows from~\cite[Eq.\ (8.10.22)]{Hildebrand56} that there exists a $\xi \in (0,1)$ such that
\begin{align} \label{Eq: ProofErrorTerm}
    E_{m} = \frac{\gamma_{m-1}}{A_{m-1}^{2}}\frac{f^{(2m-1)}\left(\xi\right)}{(2m-1)!}  \leq 0\; ,
\end{align}
where $\gamma_{m-1}, A_{m-1}$ are real constants and $\gamma_{m-1}\geq 0$. The inequality then holds due to the fact that $f^{(2m-1)}\left(\xi\right) \leq 0$. Any $m$-point GRJ quadrature (with endpoint $t_1=0$) must thus provide an upper bound on the integral.
\end{proof}
\begin{rem}
    For GRJ quadratures which instead use the endpoint $t_1=1$, $\gamma_{m-1} \leq 0$ (see~\cite[Eq.\ (8.10.15)]{Hildebrand56}). As such, these would also provide lower bounds for the integrals we are considering.
\end{rem}
\SinglehAlphaExpressionviaRelEntropy*
\begin{proof}
    For any normalized distribution, $q$, one has
    \begin{align}
        \frac{1}{\alpha-1}D\left(q \middle\Vert p\right) &= \frac{1}{\alpha-1} \sum_{\bar{c}}q(\bar{c}) \log \left( \frac{q(\bar{c})}{p(\bar{c})}\right) \\
        & = \frac{1}{\alpha-1} \sum_{\bar{c}}q(\bar{c}) \log \left( \frac{q(\bar{c})}{\mathcal{M}(\sigma)_{\land \bar{c}}} \cdot \frac{\mathcal{M}(\sigma)_{\land \bar{c}}}{p(\bar{c})}\right) \\
        & = \frac{1}{\alpha-1} \sum_{\bar{c}}q(\bar{c}) \log \left( \frac{q(\bar{c})}{\mathcal{M}(\sigma)_{\land \bar{c}}} \cdot \frac{\mathcal{M}(\sigma)_{\land \bar{c}}}{p(\bar{c})}\right) \\
        & = \frac{1}{\alpha-1} \sum_{\bar{c}}q(\bar{c}) \log \left( \frac{q(\bar{c})}{\mathcal{M}(\sigma)_{\land \bar{c}}} \right) -  \frac{1}{\alpha-1} \sum_{\bar{c}}q(\bar{c}) \log \left( \frac{p(\bar{c})}{\mathcal{M}(\sigma)_{\land \bar{c}}}\right) \\
        & = \frac{1}{\alpha-1}D\left(q \middle\Vert \mathcal{M}(\sigma)_{\bar{C}}\right)-\sum_{\bar{c}}q(\bar{c})D_{\alpha}\left(\mathcal{M}(\sigma)_{AXYTE \land \bar{c}} \middle\Vert \mathds{1}_{A}\otimes \mathcal{M}(\sigma)_{XYTE} \right)  \; ,
        \end{align}
        where we used Eq.\ (\ref{Eq: pbarcDefn}) and the fact that $\mathcal{M}(\sigma)_{\land \bar{c}}= \operatorname{Tr}\left[\mathcal{M}(\sigma)_{AXYTE \land \bar{c}}\right]$ to prove the last equality. 
\end{proof}
\LemmaFirstNonComOpt*
\begin{proof}
We consider any $m\in \mathbb{N}$ and $\alpha \in (1,2)$. By Theorem \ref{Theorem: App2}, one has that there exist nodes, $t_1,\dots, t_{m} \in \left(0,1\right]$, and weights, $w_1,\dots, w_{m} >0$,  such that $p(\bar{c}) := \operatorname{Tr}\left[\mathcal{M}(\sigma)_{AXYTE \land \bar{c}}^{\alpha} \cdot \mathds{1}_{A} \otimes \mathcal{M}(\sigma)_{XYTE}^{1-\alpha} \right]$ is upper bounded by
\begin{align}
\begin{split}
    p(\bar{c})
    &\leq C_m + \frac{\sin(\alpha \pi )}{\pi} \sum_{i} \inf_{Z_i} \left( \frac{w_{i}}{t_i}\operatorname{Tr}\left[\mathcal{M}(\sigma)_{AXYTE \land \bar{c}}\left( Z_{i}+Z_{i}^{\dagger} + \left(1-t_i\right)Z_{i}^{\dagger} Z_{i}\right)\right] \right. \\
&+\left. w_{i}\operatorname{Tr}\left[\mathds{1}_{A}\otimes\mathcal{M}(\sigma)_{XYTE}Z_iZ_{i}^{\dagger}\right] \right)\; ,
\end{split}
\end{align}
where $C_m = \operatorname{Tr}\left[\mathcal{M}(\sigma)_{AXYTE \land \bar{c}}\right] \left(1 + \frac{\sin(\alpha \pi )}{\pi}\sum_{i=1}^{m} \frac{w_{i}}{t_i} \right)$. Using the fact that $C_m$ is independent of $Z_i$ and $\alpha \in (1,2)$, this can then be expressed as 
\begin{align}
\begin{split}
    p(\bar{c})
    &\leq \sup_{Z_1, Z_2, \dots} \left( C_m + \frac{\sin(\alpha \pi )}{\pi} \sum_{i} \left(  \frac{w_{i}}{t_i}\operatorname{Tr}\left[\mathcal{M}(\sigma)_{AXYTE \land \bar{c}}\left( Z_{i}+Z_{i}^{\dagger} + \left(1-t_i\right)Z_{i}^{\dagger} Z_{i}\right)\right] \right. \right. \\
&+\left. \left. w_{i}\operatorname{Tr}\left[\mathds{1}_{A}\otimes\mathcal{M}(\sigma)_{XYTE}Z_iZ_{i}^{\dagger}\right] \right)\right) \; .
\end{split}
\end{align}
We then achieve the desired result by noting that 
\begin{align}
    & C_m + \frac{\sin(\alpha \pi )}{\pi} \sum_{i}   \frac{w_{i}}{t_i}\operatorname{Tr}\left[\mathcal{M}(\sigma)_{AXYTE \land \bar{c}}\left( Z_{i}+Z_{i}^{\dagger} + \left(1-t_i\right)Z_{i}^{\dagger} Z_{i}\right)\right] \\
    = \ &  \operatorname{Tr}\left[\mathcal{M}(\sigma)_{AXYTE \land \bar{c}}\right] + \frac{\sin(\alpha \pi )}{\pi} \sum_{i}   \frac{w_{i}}{t_i}\operatorname{Tr}\left[\mathcal{M}(\sigma)_{AXYTE \land \bar{c}}\left( \mathds{1} + Z_{i}+Z_{i}^{\dagger} + \left(1-t_i\right)Z_{i}^{\dagger} Z_{i}\right)\right] \\
     = \ &  \operatorname{Tr}\left[\mathcal{M}(\sigma)_{AXYTE \land \bar{c}} P(Z_1,Z_2,\dots)\right] \; ,
\end{align}
where 
\begin{align}
    P(Z_1,Z_2,\dots) :=  \mathds{1}+\frac{\sin(\alpha \pi )}{\pi}\sum_{i=1}^{m} \frac{w_{i}}{t_i} \left( \mathds{1} + Z_{i}+Z_{i}^{\dagger} + \left(1-t_i\right)Z_{i}^{\dagger} Z_{i}\right) \; ,
\end{align}
as well as 
\begin{align}
    & \frac{\sin(\alpha \pi )}{\pi} \sum_{i}   w_{i}\operatorname{Tr}\left[\mathcal{M}(\sigma)_{AXYTE \land \bar{c}} Z_{i}Z_{i}^{\dagger} \right] = \operatorname{Tr}\left[\mathcal{M}(\sigma)_{AXYTE \land \bar{c}} Q(Z_1,Z_2,\dots)\right] \; ,
\end{align}
where
\begin{align}
    Q(Z_1,Z_2,\dots) := \frac{\sin(\alpha \pi )}{\pi}\sum_{i=1}^{m} w_i Z_iZ_{i}^{\dagger} \; .
\end{align}
\end{proof}
\TheoremPGFCB*
\begin{proof}
    From Theorem \ref{Theorem: CBSuffCond}, it is known that a secret key can be produced whenever 
\begin{align} \label{Eq: ProofCBSecCond}
Q(\rho_{EF|\tilde{0}\tilde{0}},\rho_{EF|\tilde{1}\tilde{1}}) > \frac{\epsilon}{1-\epsilon}
\end{align}    
holds. Using the fact that $\operatorname{Tr}\left[ \rho^s \sigma^{1-s}\right]$ is convex in $s$ \cite{Audenaert_2007},
\begin{align}
    Q(\rho_{EF|\tilde{0}\tilde{0}},\rho_{EF|\tilde{1}\tilde{1}}) &:= \inf_s \operatorname{Tr}\left[ \rho_{EF|\tilde{0}\tilde{0}}^s \rho_{EF|\tilde{1}\tilde{1}}^{1-s}\right] \\
     &= \inf_s \frac{1}{2}\left(  \operatorname{Tr}\left[ \rho_{E|00}^{s} \rho_{E|11}^{1-s}\right] + \operatorname{Tr}\left[ \rho_{E|00}^{1-s} \rho_{E|11}^{s}\right] \right) \\
     &\geq  \inf_s \operatorname{Tr}\left[ \rho_{E|00}^{1/2} \rho_{E|11}^{1/2}\right] \\
     &=  \operatorname{Tr}\left[ \rho_{E|00}^{1/2} \rho_{E|11}^{1/2}\right] \\
     &=: F_{pg}(\rho_{E|00},\rho_{E|11}) \; .
\end{align}
As such, whenever 
\begin{align}
F_{pg}(\rho_{E|00},\rho_{E|11}) > \frac{\epsilon}{1-\epsilon} 
\end{align}  
is satisfied, Eq.\ (\ref{Eq: ProofCBSecCond}) is true as well. It then follows from Theorem \ref{Theorem: CBSuffCond} that a secret key can be produced.
\end{proof}
\subsubsection{Proofs of Proposition \ref{Prop: Approach1Part1VarOpt} and Proposition \ref{Prop: Approach1Part2VarOpt}}
We will follow the proof method used in \cite{brown2023deviceindependent}. Let $\rho = \sum q_{k}\ketbra{\phi_{k}}{\phi_{k}}$ and $\sigma = \sum p_{j}\ketbra{\psi_{j}}{\psi_{j}}$ be positive semi-definite operators on a finite-dimensional Hilbert space, $\mathcal{H}$, where $p_{j},q_{k} \geq 0$ for all $j,k$, and both $\{\ket{\psi_{j}}\}$ as well as $\{\ket{\phi_{k}}\}$ form an orthonormal basis of $\mathcal{H}$.
For any operator, $a$, on the Hilbert space, we will consider the two commuting mappings 
\begin{align}
    \mathcal{L}_{\sigma}(a) &=  \sigma a  \\
    \mathcal{R}_{\rho}(a) 
    &=  a \rho \; .
\end{align}
Moreover, we consider any F-quasi-relative entropy\footnote{See \cite[Section 3]{brown2023deviceindependent} or \cite{Pusz1978WYDL} for further background on such terms.}
\begin{align}
    D_{F} \left( \rho || \sigma\right) := \sum_{j,k} F(q_k,p_j) \left| \braket{\psi_{j} | \phi_{k}}\right|^2 \; ,
\end{align}
such that $F: \mathcal{R}^{2}_{\geq 0} \to \mathcal{R}$
\begin{enumerate}
  \item is (Borel) measurable.
  \item is positive homogeneous.
  \item is locally bounded from below.
\end{enumerate}
\begin{rem} \label{Remark: ExpandingSum}
    In Section \ref{Section: TightBoundsPetzRenyiVariational}, this is not exactly the sum we initially consider, as one only summed over all $j,k$ such that $p_j,q_k >0$. For Proposition \ref{Prop: Approach1Part1VarOpt}, we use $F_{1}(q_k,p_j) := \frac{q_k \cdot p_j}{q_k+p_j t}$ and ${F_{2}(q_k,p_j) := \frac{q_k \cdot p_j}{p_j+q_k t}}$. Whenever $p_j=0$ or $q_k=0$ (or both), these functions are well-defined and output $0$ (or may be continuously extended to return $0$). As such, we may expand to summation, to include all eigenvalues. For Proposition \ref{Prop: Approach1Part2VarOpt}, we consider $F_{3}(q_k,p_j) := \frac{q_k^2}{q_k+ p_j t}$, and $F_{4}(q_k,p_j) := \frac{q_k^2}{p_j + q_k t}$. These are non-zero whenever $q_k > 0$. As such, it is not a priori clear that the sum can similarly be expanded. Any issues that arise here can be avoided by additionally requiring that $\rho <\!\!<\sigma$.\footnote{This is not completely unexpected, as Proposition \ref{Prop: Approach1Part2VarOpt} relates to the case when $\alpha >1$. For {\Renyi} divergences, this is always required.} To see why this resolves the problem, note that whenever $p_j=0$, one of the following two properties must now hold for every value $k$. Either $q_k=0$ or $\left| \braket{\psi_{j} | \phi_{k}}\right| = 0$. As such, the sum can be expanded again without changing its value. 
    These functions are clearly Borel measurable, as they are composed of sums, products, and quotients of linear functions.\footnote{Moreover, there is no issue when $q_k,p_j \to 0$, as the quotient is still well-defined and the function is continuous.} Moreover, they satisfy
\begin{align}
    F_{i}(\lambda q_k,\lambda  p_j) = \lambda F_{i}(q_k,p_j)
\end{align}
    and are thus positive homogeneous. The output of these two functions is always non-negative and they are therefore locally bounded from below. 
\end{rem}
Such quasi-relative entropies can be expressed using the Hilbert-Schmidt inner product \cite[Remark 3.4]{brown2023deviceindependent}, i.e.\
\begin{align}
    D_{F} \left( \rho || \sigma\right) = \braket{\mathds{1},F\left(\mathcal{R}_{\rho},\mathcal{L}_{\sigma}\right)\mathds{1}}_{HS} \; ,
\end{align}
where $\braket{a,b}_{HS} = \operatorname{Tr}\left[a^\dagger b\right]$ and we replaced $q_k,p_j$ with commuting mappings. We are now ready to prove the individual Propositions. 
\ApproachOnePartOneVarOpt*
\begin{proof}
   We use \cite[Lemma]{PUSZ1975159} (or see \cite[Proof of Proposition 3.5]{brown2023deviceindependent}), which restates these inner products as optimization problems, i.e.\
\begin{align}
\braket{\mathds{1},F_{1}\left(\mathcal{R}_{\rho},\mathcal{L}_{\sigma}\right)\mathds{1}}_{HS} &= \inf_{Z} \frac{1}{t}\left( \braket{Z,\mathcal{R}_{\rho}Z }_{HS} + t \braket{\mathds{1}-Z,\mathcal{L}_{\sigma}\left(\mathds{1}-Z\right) }_{HS}\right)
\end{align} 
Explicitly writing out the Hilbert-Schmidt inner product, we get
\begin{align}   \braket{\mathds{1},F_{1}\left(\mathcal{R}_{\rho},\mathcal{L}_{\sigma}\right)\mathds{1}}_{HS} &= \frac{1}{t} \inf_{Z} \left( \operatorname{Tr}\left[Z^\dagger \mathcal{R}_{\rho} \left(Z\right)\right] + t \operatorname{Tr}\left[\left(\mathds{1}-Z\right)^\dagger \mathcal{L}_{\sigma} \left(\mathds{1}-Z\right)\right]\right) \\
&= \frac{1}{t}\inf_{Z} \left( \operatorname{Tr}\left[Z^\dagger Z\rho\right] + t \operatorname{Tr}\left[\left(\mathds{1}-Z^\dagger\right) \sigma \left(\mathds{1}-Z\right)\right]\right) \\
&=\frac{1}{t} \inf_{Z}\left( \operatorname{Tr}\left[\rho Z^\dagger Z\right] + t \operatorname{Tr}\left[ \sigma \left(\mathds{1}-Z\right) \left(\mathds{1}-Z^\dagger\right)\right]\right) \\
&= \frac{1}{t}\inf_{Z} \left( \operatorname{Tr}\left[\rho Z^\dagger Z\right] + t \operatorname{Tr}\left[ \sigma \left(\mathds{1}+Z\right) \left(\mathds{1}+Z^\dagger \right)\right]\right) \; ,
\end{align}
where we applied the cyclic property of the trace in the penultimate line. In the last line, we used the fact that replacing $Z$ with $-Z$ does not affect the overall optimization problem.

Eq.\ (\ref{Eq: Approach1Part1VarOptEq2}) can be proven analogously. Alternatively, one may use the fact that Eq.\ (\ref{Eq: Approach1Part1VarOptEq2}) is equivalent to Eq.\ (\ref{Eq: Approach1Part1VarOptEq1}), after switching $\rho$ and $\sigma$. It then follows immediately that
\begin{align}
\braket{\mathds{1},F_{2}\left(\mathcal{R}_{\rho},\mathcal{L}_{\sigma}\right)\mathds{1}}_{HS} &= \frac{1}{t} \inf_{Z} \left( \operatorname{Tr}\left[\sigma Z^\dagger Z\right] + t \operatorname{Tr}\left[ \rho \left(\mathds{1}+Z\right) \left(\mathds{1}+Z^\dagger\right)\right]\right) \\
&= \frac{1}{t} \inf_{Z}\left( \operatorname{Tr}\left[\sigma Z Z^\dagger\right] + t \operatorname{Tr}\left[ \rho \left(\mathds{1}+Z^\dagger\right) \left(\mathds{1}+Z\right)\right]\right) \; ,
\end{align} 
where we use the fact that substituting $Z$ with $Z^\dagger$ does not affect the overall optimization problem.
\end{proof}
\ApproachOnePartTwoVarOpt*
\begin{proof}
Let us first consider Eq.\ (\ref{Eq: Approach2VarOptEq1}). We will use the fact that
   \begin{align}
       \frac{q_k^2}{q_k+ p_j t} &= q_{k} \frac{q_k+ p_j t}{q_k+ p_j t} - q_{k} \frac{ p_j t}{q_k+ p_j t} \\
       &=q_{k} - \frac{ q_k \cdot p_j t}{q_k+ p_j t} \; .
   \end{align} 
We consider the two terms separately. The first contribution can be simplified via
\begin{align}
     \sum_{j,k} q_k \left| \braket{\psi_{j} | \phi_{k}}\right|^2  = \operatorname{Tr}\left[\rho\right] \; .
\end{align}
The remaining contribution has been discussed in Proposition \ref{Prop: Approach1Part1VarOpt}, i.e.\
\begin{align}
    \sum_{j,k} \frac{ q_k \cdot p_j t}{q_k+ p_j t} \left| \braket{\psi_{j} | \phi_{k}}\right|^2  = \frac{1}{t} \inf_{Z} \left( \operatorname{Tr}\left[\rho Z^\dagger Z\right] + t \operatorname{Tr}\left[ \sigma \left(\mathds{1}+Z\right) \left(\mathds{1}+Z^\dagger \right)\right]\right) \; .
\end{align}
Combining these two expressions, one finds that
\begin{align}
\sum_{j,k}  \frac{q_k^2}{q_k+p_j t} \left| \braket{\psi_{j} | \phi_{k}}\right|^2 &= \operatorname{Tr}\left[\rho\right] - 
   \inf_{Z} \left(\operatorname{Tr}\left[\rho Z^\dagger Z\right]+t \operatorname{Tr}\left[\sigma \left(\mathds{1} + Z \right) \left(\mathds{1} + Z^\dagger \right)
       \right]\right) \; .
\end{align} 
We prove Eq.\ (\ref{Eq: Approach2VarOptEq2}) similarly. Using
   \begin{align}
       \frac{q_k^2}{p_j+ q_k t} &= \frac{q_{k}}{t} \frac{p_j+ q_k t}{p_j+ q_k t} - \frac{q_{k}}{t} \frac{ p_j}{p_j + q_k t} \\
       &=\frac{1}{t}\left[q_{k} - \frac{ q_k \cdot p_j }{p_j+ q_k t}\right] \; .
   \end{align} 
   and Proposition \ref{Prop: Approach1Part1VarOpt}, one gets that
   \begin{align}
       \sum_{j,k} \frac{ q_k^2}{p_j+ q_k t} \left| \braket{\psi_{j} | \phi_{k}}\right|^2  &=\frac{1}{t} \left[ \operatorname{Tr}\left[\rho\right] - \frac{1}{t}\inf_{Z} \left(  t \operatorname{Tr}\left[ \rho \left(\mathds{1}+Z^\dagger\right) \left(\mathds{1}+Z \right)\right]+ \operatorname{Tr}\left[\sigma Z Z^\dagger \right]\right) \right] \\
        &=\frac{1}{t} \left[ \operatorname{Tr}\left[\rho\right] - \inf_{Z} \left( \operatorname{Tr}\left[ \rho \left(\mathds{1}+Z^\dagger\right) \left(\mathds{1}+Z \right)\right]+ \frac{1}{t}\operatorname{Tr}\left[\sigma Z Z^\dagger \right]\right) \right]\; .
   \end{align}
\end{proof}
\subsubsection{Proofs of Theorem \ref{Theorem: App1Part1}, Theorem \ref{Theorem: App1Part2}, and Theorem \ref{Theorem: App2}}
For any $\beta \in (0,1)$, an $m$-point GJ/GRJ quadrature of an integral 
\begin{align}
    \int_{0}^{1}f(t)t^{\beta}dt  \;  ,
\end{align}
where $f(t)$ is some smooth function, generates nodes $t_1, \dots, t_{m}$ and weights $w_1,\dots, w_m$ such that~\cite[Eq.\ (8.9.6)]{Hildebrand56}
\begin{align}
    \int_{0}^{1}f(t)t^{\beta}dt = \sum_{i=1}^{m} w_{i} f(t_i) +E_{m} \; ,
\end{align}
where $E_{m}$ denotes the error term. For proving convergence of the quadratures, we use the approach from~\cite[Appendix A]{brown2023deviceindependent}
\TheoremApponePartone*

\begin{proof}
    By applying Lemma \ref{Lem: BoundingIntviaQuad} to Eq.\ (\ref{Eq: Approach1DoubleIntegral}) one directly finds that the corresponding m-point GJ quadratures satisfy
    \begin{align}
\operatorname{Tr}\left[\rho^{\alpha}\sigma^{1-\alpha}\right] &\geq  \frac{\sin (\alpha \pi)}{\pi}   \sum_{i=1}^{m} \left[ \sum_{\substack{j,k \ \mathrm{s.t.} \\ q_k,p_j >0}}w_i\frac{q_k \cdot p_j}{q_k+p_j t_i} \left| \braket{\psi_{j} | \phi_{k}}\right|^2  +  \sum_{\substack{j,k \ \mathrm{s.t.} \\ q_k,p_j >0}}w_i^\prime\frac{q_k \cdot p_j}{p_j+q_k t_i^\prime} \left| \braket{\psi_{j} | \phi_{k}}\right|^2 \right] \; .
    \end{align}
    The desired lower bounds then follow from Eq.\ (\ref{Eq: IncludingAllSumsApproach1Part1}) and Proposition \ref{Prop: Approach1Part1VarOpt}.
    To prove convergence, one has to show that for every $q_k,p_j>0$
    \begin{align}
        \lim_{m \to \infty }E_m &=  \lim_{m \to \infty } \int_{0}^{1}\frac{q_k \cdot p_j}{q_k + p_j t}t^{\alpha-1}dt - \sum_{i=1}^{m} w_i \frac{q_k \cdot p_j}{q_k + p_j t_i} = 0 \\
        \lim_{m \to \infty }E_m^{\prime} &=  \lim_{m \to \infty } \int_{0}^{1}\frac{q_k \cdot p_j}{q_k + p_j t}t^{-\alpha}dt - \sum_{i=1}^{m} w_i^{\prime} \frac{q_k \cdot p_j}{ p_j+ q_k t_i^{\prime}} = 0 \; .
    \end{align}
Since both $F_1 (t) :=\frac{q_k \cdot p_j}{q_k + p_j t}$ and $F_2 (t) :=\frac{q_k \cdot p_j}{p_j + q_k t}$ are continuous (for fixed $q_k,p_j >0$) on the bounded interval $t \in [0,1]$, these functions can be approximated by polynomials arbitrarily well.\footnote{This follows from the Weierstrass approximation theorem.} As such, for all $\epsilon > 0 $, there exist polynomials $P_{1}(t)$ and $P_{2}(t)$ of degree $n_1$ and $n_2$, respectively, such that
\begin{align} \label{Eq: Proof1PolApprox}
    \max_{t \in [0,1]} \left|F_{l}(t)-P_{l}(t) \right| < \epsilon \; ,
\end{align}
where $l \in \{1,2\}$. Then, for all $m \geq \max \frac{1}{2}\{n_1+1,n_2 +1\}$ and 
    \begin{align}
       \left| \int_{0}^{1}F_{1}(t)t^{\alpha-1}dt - \sum_{i=1}^{m} w_i F_{1}(t_i) \right| &= \left| \int_{0}^{1}\left(F_{1}(t)-P_{1}(t)\right)t^{\alpha-1}dt + \sum_{i=1}^{m} w_i \left(P_{1}(t_i)-F_{1}(t_i)\right) \right| \\
       &\leq \left| \int_{0}^{1}\left(F_{1}(t)-P_{1}(t)\right)t^{\alpha-1}dt  \right| + \left| \sum_{i=1}^{m} w_i \left(P_{1}(t_i)-F_{1}(t_i)\right) \right| \\
       &\leq  \epsilon \left| \int_{0}^{1}t^{\alpha-1}dt  \right| + \epsilon \left| \sum_{i=1}^{m} w_i  \right| \\
        &= \frac{2\epsilon}{\alpha} \; .
    \end{align}
The first line follows from the fact that m-point GJ quadratures are exact for polynomials of degree up to $2m-1$. The second line follows from the triangle inequality. The third line is due to Eq.~(\ref{Eq: Proof1PolApprox}), and the last line follows from 
\begin{align} \label{Eq: EqualityweightsApproximation}
    \int_{0}^{1}t^{\alpha-1}dt = \sum_{i=1}^{m} w_i = \frac{1}{\alpha}
\end{align}
    Similarly,
     \begin{align}
       \left| \int_{0}^{1}F_{2}(t)t^{-\alpha}dt - \sum_{i=1}^{m} w_i^\prime F_{2}(t_i^\prime) \right| \leq \frac{2\epsilon}{1-\alpha} \; .
    \end{align}
    For any fixed $\alpha \in (0,1)$, the quadrature thus converges as $m \to \infty$.
\end{proof}
\TheoremApponeParttwo*

\begin{proof}
    Applying Lemma \ref{Lem: BoundingIntviaQuad} to Eq.\ (\ref{Eq: Approach1Part2Integral}), one directly finds that the corresponding m-point GRJ quadratures satisfy
    \begin{align}
\operatorname{Tr}\left[\rho^{\alpha}\sigma^{1-\alpha}\right] &\leq  \frac{\sin (\left(\alpha -1\right)\pi)}{\pi}   \sum_{i=1}^{m} \left[ \sum_{\substack{j,k \ \mathrm{s.t.} \\ q_k,p_j >0}}w_i\frac{q_k^2}{q_k+p_j t_i} \left| \braket{\psi_{j} | \phi_{k}}\right|^2  +  \sum_{\substack{j,k \ \mathrm{s.t.} \\ q_k,p_j >0}}w_i^\prime\frac{q_k^2}{p_j+q_k t_i^\prime} \left| \braket{\psi_{j} | \phi_{k}}\right|^2 \right] \; .
    \end{align}
    From Remark \ref{Rem: CompletingSumApp1Part2} and Proposition \ref{Prop: Approach1Part2VarOpt}, it follows that for every $t_i,t_i^\prime >0$
    \begin{align}
         \sum_{\substack{j,k \ \mathrm{s.t.} \\ q_k,p_j >0}}w_i\frac{q_k^2}{q_k+p_j t_i} \left| \braket{\psi_{j} | \phi_{k}}\right|^2 &= w_i \left[ \operatorname{Tr}\left[\rho\right] - 
   \inf_{Z} \left(\operatorname{Tr}\left[\rho Z^\dagger Z\right]+t_i \operatorname{Tr}\left[\sigma \left(\mathds{1} + Z \right) \left(\mathds{1} + Z^\dagger \right)
       \right]\right) \right]\\
         \sum_{\substack{j,k \ \mathrm{s.t.} \\ q_k,p_j >0}}w_i^\prime\frac{q_k^2}{p_j+q_k t_i^\prime} \left| \braket{\psi_{j} | \phi_{k}}\right|^2 &= \frac{w_i^\prime}{t_i^\prime} \left[ \operatorname{Tr}\left[\rho\right] - \inf_{Z} \left(\operatorname{Tr}\left[\rho 
   \left(\mathds{1} + Z^\dagger \right) \left(\mathds{1} + Z \right)
       \right] + \frac{1}{t_i^\prime}\operatorname{Tr}\left[\sigma Z Z^\dagger \right]\right) \right] \; .
    \end{align}
    Moreover, by Remark \ref{Rem: EndNodeApp1Part2}, the contributions for $t_1,t_1^\prime =0$ are
    \begin{align}
         \sum_{\substack{j,k \ \mathrm{s.t.} \\ q_k,p_j >0}}w_1 q_k \left| \braket{\psi_{j} | \phi_{k}}\right|^2 &= w_1  \operatorname{Tr}\left[\rho\right] \\
         \sum_{\substack{j,k \ \mathrm{s.t.} \\ q_k,p_j >0}}w_1^\prime\frac{q_k^2}{p_j} \left| \braket{\psi_{j} | \phi_{k}}\right|^2 &= - w_1^\prime\inf_{Z} \left(\operatorname{Tr}\left[\rho\left( Z^\dagger  + Z \right)\right] + \operatorname{Tr}\left[\sigma Z Z^\dagger\right]\right) \; .
    \end{align}
    The desired bound then follows from the fact that\footnote{This equality holds, for the same reason as Eq.\ (\ref{Eq: EqualityweightsApproximation}).}
    \begin{align}
        \sum_{i=1}^{m}w_i \operatorname{Tr}\left[\rho\right] = \frac{1}{\alpha-1} \operatorname{Tr}\left[\rho\right]
    \end{align}
    and $\sin (\alpha\pi) = -\sin (\left(\alpha -1\right)\pi)$.
  To prove convergence, one has to show that, for every $q_k,p_j>0$,
    \begin{align}
        \lim_{m \to \infty }E_m &=  \lim_{m \to \infty } \int_{0}^{1}\frac{q_k^2}{q_k + p_j t}t^{\alpha-2}dt - \sum_{i=1}^{m} w_i \frac{q_k^2}{q_k + p_j t_i} = 0 \\
        \lim_{m \to \infty }E_m^{\prime} &=  \lim_{m \to \infty } \int_{0}^{1}\frac{q_k^2}{q_k + p_j t}t^{1-\alpha}dt - \sum_{i=1}^{m} w_i^{\prime} \frac{q_k^2}{ p_j+ q_k t_i^{\prime}} = 0 \; .
    \end{align}
Since both $F_3 (t) :=\frac{q_k^2}{q_k + p_j t}$ and $F_4 (t) :=\frac{q_k^2}{p_j + q_k t}$ are continuous (for fixed $q_k,p_j >0$) on the bounded interval $t \in [0,1]$, these functions can again be approximated by polynomials arbitrarily well. As such, for all $\epsilon > 0 $, there exist polynomials $P_{3}(t)$ and $P_{4}(t)$ of degree $n_3$ and $n_4$, respectively, such that
\begin{align} \label{Eq: Proof2PolApprox}
    \max_{t \in [0,1]} \left|F_{l}(t)-P_{l}(t) \right| < \epsilon \; ,
\end{align}
where $l \in \{3,4\}$. Then, for all $m \geq \max \frac{1}{2}\{n_3+2,n_4+2 \}$ and 
    \begin{align}
       \left| \int_{0}^{1}F_{3}(t)t^{\alpha-2}dt - \sum_{i=1}^{m} w_i F_{3}(t_i) \right| &= \left| \int_{0}^{1}\left(F_{3}(t)-P_{3}(t)\right)t^{\alpha-2}dt + \sum_{i=1}^{m} w_i \left(P_{3}(t_i)-F_{3}(t_i)\right) \right| \\
       &\leq \left| \int_{0}^{1}\left(F_{3}(t)-P_{3}(t)\right)t^{\alpha-2}dt  \right| + \left| \sum_{i=1}^{m} w_i \left(P_{3}(t_i)-F_{3}(t_i)\right) \right| \\
       &\leq  \epsilon \left| \int_{0}^{1}t^{\alpha-2}dt  \right| + \epsilon \left| \sum_{i=1}^{m} w_i  \right| \\
        &= \frac{2\epsilon}{\alpha-1} \; .
    \end{align}
The first line follows from the fact that m-point GRJ quadratures are exact for polynomials of degree up to $2m-2$. The second line follows from the triangle inequality. The third line is due to Eq.~(\ref{Eq: Proof2PolApprox}), and the last line follows from 
\begin{align}
    \int_{0}^{1}t^{\alpha-2}dt = \sum_{i=1}^{m} w_i = \frac{1}{\alpha-1}
\end{align}
    Similarly,
     \begin{align}
       \left| \int_{0}^{1}f_{2}(t)t^{1-\alpha}dt - \sum_{i=1}^{m} w_i^\prime f_{2}(t_i^\prime) \right| \leq \frac{2\epsilon}{2-\alpha} \; .
    \end{align}
    For any fixed $\alpha \in (1,2)$, the quadrature thus converges as $m \to \infty$.
\end{proof}
\TheoremApptwo*
\begin{proof}
    Applying Proposition \ref{Prop: App2QuadBounds} to Eq.\ (\ref{Eq: App2Integral}), one directly finds that the  m-point GRJ quadrature (with end node $t_m=1$)  satisfies
    \begin{align}
\operatorname{Tr}\left[\rho^{\alpha}\sigma^{1-\alpha}\right] &\geq \operatorname{Tr}\left[\rho\right] + \sum_{i=1}^{m} w_{i}\sum_{\substack{j,k \ \mathrm{s.t.} \\ q_k,p_j >0}}  q_k \left[ \frac{\sin \left(\alpha\pi\right)}{\pi }\frac{p_j-q_k}{t_{i}\left(p_j-q_k\right) +q_k} \right] \left| \braket{\psi_{j} | \phi_{k}}\right|^2 
    \end{align}
    for $\alpha \in (0,1)$. The reverse inequality holds for $\alpha \in (1,2)$. The desired bounds then follow from  Eq.~(\ref{Eq: App2CompletionSum}) and Proposition \ref{Prop: Approach2VarOpt}.  To prove convergence, one has to show that, for every $q_k,p_j>0$,
    \begin{align}
        \lim_{m \to \infty }E_m &=  \lim_{m \to \infty } \int_{0}^{1} q_k  \frac{p_j-q_k}{t\left(p_j-q_k\right) +q_k}t^{\alpha-1}\left(1-t\right)^{1-\alpha}dt - \sum_{i=1}^{m} w_i q_k  \frac{p_j-q_k}{t_{i}\left(p_j-q_k\right) +q_k} = 0 \; . \label{Eq: ProofApp2QuadConvergence}
    \end{align}
It follows from~\cite[Theorem 1]{faust2023rational} that, for any $x>0$,
\begin{align}
    \lim_{m \to \infty } \int_{0}^{1}   \frac{x-1}{t\left(x-1\right) +1}t^{\alpha-1}\left(1-t\right)^{1-\alpha}dt - \sum_{i=1}^{m} w_i \frac{x-1}{t_{i}\left(x-1\right) +q_k} = 0 \;  . \label{Eq: ProofFaustQuadConvergence}
\end{align}
Setting $x = \frac{p_j}{q_k}$ and multiplying Eq.\ (\ref{Eq: ProofFaustQuadConvergence}) with $q_k$ yields the desired result.
\end{proof}

\subsubsection{Proofs of Proposition \ref{Prop:FirstNonComOpt} and Proposition \ref{Prop: SecondNonComOpt}}
These proofs will closely follow the proof-structure from~\cite[Lemma~2.3]{brown2023deviceindependent}. In particular they use the fact that for any two (potentially non-normalized) classical-quantum states ${\rho_{AE}=\sum_{a}\ketbra{a}{a} \otimes \rho_{E|A=a}}$ and ${\rho_{AE}^{\prime}=\sum_{a}\ketbra{a}{a} \otimes \rho_{E|A=a}^{\prime}}$, the optimization problem 
\begin{align}
    \inf_{Z}  \frac{1}{t}\operatorname{Tr}\left[\rho_{AE} \left( \mathds{1} + Z+Z^\dagger + \left(1-t\right)Z^{ \dagger} Z\right)\right]+\operatorname{Tr}\left[\rho_{AE}^{\prime} Z Z^\dagger\right]
\end{align}
is equivalent to 
\begin{align} \label{Eq: ProofReducingActionofZ}
    \inf_{Z_{a}}\sum_{a}  \frac{1}{t}\operatorname{Tr}\left[\rho_{E|A=a} \left( \mathds{1} + Z_{a}+Z_{a}^\dagger + \left(1-t\right)Z_{a}^{ \dagger} Z_{a}\right)\right]+\operatorname{Tr}\left[\rho_{E|A=a}^{\prime} Z_{a} Z_{a}^\dagger\right]
\end{align}
for all $t\in \left(0,1\right]$.
\PropFirstNonComOpt*
\begin{proof}
    By definition,
    \begin{align}
        \mathcal{M}(\sigma)_{AXYTE \land \bar{c}}& = \sum_{x ,y,t} \operatorname{P}\left(X=x,Y=y,T=t\right) \ketbra{x,y,t}{x,y,t} \otimes \mathcal{M}(\sigma)_{AE \land \bar{c}|X=x,Y=y,T=t} \\
        \mathcal{M}(\sigma)_{XYTE}& = \sum_{x ,y,t} \operatorname{P}\left(X=x,Y=y,T=t\right) \ketbra{x,y,t}{x,y,t} \otimes \mathcal{M}(\sigma)_{E |X=x,Y=y,T=t} \; .
    \end{align}
   For any $x,y,t$ such that $\operatorname{P}\left(\bar{C} =\bar{c}|X=x,Y=y,T=t\right)=0$, $\mathcal{M}(\sigma)_{AE\land \bar{c} |X=x,Y=y,T=t}=0$. As such, those terms can be naturally be neglected. Therefore,
    \begin{align}
        \mathcal{M}(\sigma)_{AXYTE \land \bar{c}}^{\alpha}& = \sum_{\substack{x ,y,t \ \mathrm{s.t.} \\\operatorname{P}\left(\bar{c}|x,y,t\right)>0}} \operatorname{P}\left(X=x,Y=y,T=t\right)^{\alpha} \ketbra{x,y,t}{x,y,t} \otimes \mathcal{M}(\sigma)_{AE \land \bar{c}|X=x,Y=y,T=t}^{\alpha} \\
        \mathcal{M}(\sigma)_{XYTE}^{1-\alpha}& = \sum_{x ,y,t } \operatorname{P}\left(X=x,Y=y,T=t\right)^{1-\alpha} \ketbra{x,y,t}{x,y,t} \otimes \mathcal{M}(\sigma)_{E |X=x,Y=y,T=t}^{1-\alpha} \; .
    \end{align}
    Eq.\ (\ref{Eq: ExtendedPTerm}) then directly follows, by using the fact that Eve's inital quantum side information is independent of the classical registers $X,Y,T$, i.e.
    \begin{align}
        \mathcal{M}(\sigma)_{E |X=x,Y=y,T=t} = \mathcal{M}(\sigma)_{E} \; ,
    \end{align}
    and that $\braket{x,y,t|x^\prime,y^\prime, t^\prime}=\delta_{x,x^\prime} \delta_{y,y^\prime} \delta_{t,t^\prime}$, where $\delta_{i,j}$ refers to the Kronecker delta function.
Lemma \ref{Lemma:FirstNonComOpt} can naturally be extended to bound
    \begin{align} \label{Eq: ConditionalTraceDefn}
    \operatorname{Tr}\left[\mathcal{M}(\sigma)_{AE \land \bar{c}|X=x,Y=y,T=t}^{\alpha} \cdot \mathds{1}_{A} \otimes \mathcal{M}(\sigma)_{E}^{1-\alpha} \right] \; .
    \end{align}
As such, there exist nodes, $t_1,\dots, t_{m} \in \left(0,1\right]$, and weights, $w_1,\dots, w_{m} >0$, such that Eq.\ (\ref{Eq: ConditionalTraceDefn}) is upper bounded by
    \begin{align} \label{Eq: ProofIntermediateVarOptProb}
\sup_{Z_{i}^{x,y,t}}  \operatorname{Tr}\left[\mathcal{M}(\sigma)_{AE \land \bar{c}|X=x,Y=y,T=t}P(Z_{1}^{x,y,t},\dots)\right]+\operatorname{Tr}\left[\mathds{1}_{A}\otimes \mathcal{M}(\sigma)_{E} Q(Z_{1}^{x,y,t},\dots)\right]\; ,
    \end{align}
where 
\begin{align}
    P(Z_{1}^{x,y,t},Z_{2}^{x,y,t},\dots) &:= 1 + \frac{\sin(\alpha \pi )}{\pi}\sum_{i=1}^{m} \frac{w_{i}}{t_i} \left( \mathds{1} + Z_{i}^{x,y,t}+Z_{i}^{x,y,t \dagger} + \left(1-t_i\right)Z_{i}^{x,y,t \dagger} Z_{i}^{x,y,t}\right)\\
    Q(Z_{1}^{x,y,t},Z_{2}^{x,y,t},\dots) &:= \frac{\sin(\alpha \pi )}{\pi}\sum_{i=1}^{m} w_i Z_{i}^{x,y,t}Z_{i}^{x,y,t \dagger} \; .
\end{align}
Currently, each $Z_{i}^{x,y,t}$ acts on the joint Hilbert space $\mathcal{H}_{AE}$. However, analogous to~\cite[Lemma~2.3]{brown2023deviceindependent}, we may reduce their actions to Eve's quantum register, using Eq.\ (\ref{Eq: ProofReducingActionofZ}). That is, there exists a family, $\{Z_{a,i}^{x,y,t}\}$, such that Eq.\ (\ref{Eq: ProofIntermediateVarOptProb}) is equivalent to 
\begin{align} 
\sup_{Z_{a,i}^{x,y,t}} \sum_{a} \operatorname{Tr}\left[\mathcal{M}(\sigma)_{E \land \bar{c}|A=a,X=x,Y=y,T=t}P(Z_{a,1}^{x,y,t},\dots)\right]+\operatorname{Tr}\left[\mathcal{M}(\sigma)_{E} Q(Z_{a,1}^{x,y,t},\dots)\right]\; .
    \end{align}
Eq.\ (\ref{Eq: XYTConditionalPUpperBounds}) then directly follows after using both
    \begin{align}
       \mathcal{M}(\sigma)_{E \land \bar{c}|A=a, X=x,Y=y,T=t} &=  \sum_{\substack{b \ \mathrm{s.t.} \\ \bar{c} = f(a,b,x,y,t)}} \operatorname{Tr}_{Q_AQ_B}\left[\left(M_{a|x} \otimes N_{b|y} \otimes \mathds{1}_{E} \right) \sigma_{Q_AQ_BE}\right] \\
       \mathcal{M}(\sigma)_{E} &=  \operatorname{Tr}_{Q_AQ_B}\left[\left(\mathds{1}_{Q_AQ_B} \otimes \mathds{1}_{E} \right) \sigma_{Q_AQ_BE}\right] \; ,
    \end{align}
    and the identity~\cite[Eq.~3.2.19]{khatri2024principles}
    \begin{align}       \operatorname{Tr}_{Q_A Q_B}\left[ X_{Q_AQ_BE} \left(\mathds{1}_{Q_AQ_B} \otimes Z\right) \right] = \operatorname{Tr}_{Q_A Q_B}\left[ X_{Q_AQ_BE} \right] Z \; .
    \end{align}
\end{proof}
\PropSecondNonComOpt*
\begin{proof}
    Eq.\ (\ref{Eq: EntropyCondOnClassicalStates}) is a direct consequence of the relation
    \begin{align}
        \mathcal{M}(\sigma)_{E |X=x,T=0} = \mathcal{M}(\sigma)_{E} 
    \end{align}
    and~\cite[Proposition~5.1]{Tomamichel2015QuantumIP}, which decomposes the conditional entropy into the desired sum. Analogous to the upper bounds on Eq.\ (\ref{Eq: ConditionalTraceDefn}), one may again extend Lemma~\ref{Lemma:FirstNonComOpt} to generate nodes, $t_1,\dots, t_{m} \in \left(0,1\right]$, and weights, $w_1,\dots, w_{m} >0$, such that $\operatorname{Tr}\left[\mathcal{M}(\sigma)_{AE|X=x}^{\alpha} \cdot \mathds{1}_{A} \otimes \mathcal{M}(\sigma)_{E}^{1-\alpha} \right]$ is upper bounded by
    \begin{align} \label{Eq: ProofIntermediateVarOptProb2}
\sup_{Z_{i}^{x}}  \operatorname{Tr}\left[\mathcal{M}(\sigma)_{AE |X=x}P(Z_{1}^{x},Z_{2}^{x},\dots)\right]+\operatorname{Tr}\left[\mathds{1}_{A}\otimes \mathcal{M}(\sigma)_{E} Q(Z_{1}^{x},Z_{2}^{x},\dots)\right]\; ,
    \end{align}
    where 
\begin{align}
    P(Z_{1}^{x},Z_{2}^{x},\dots) &:= 1 + \frac{\sin(\alpha \pi )}{\pi}\sum_{i=1}^{m} \frac{w_{i}}{t_i} \left( \mathds{1} + Z_{i}^{x}+Z_{i}^{x \dagger} + \left(1-t_i\right)Z_{i}^{x \dagger} Z_{i}^{x}\right)\\
    Q(Z_{1}^{x},Z_{2}^{x},\dots) &:= \frac{\sin(\alpha \pi )}{\pi}\sum_{i=1}^{m} w_i Z_{i}^{x}Z_{i}^{x \dagger} \; .
\end{align}
To reduce the action of $Z_{i}^{x}$ which currently acts on the joint Hilbert space $\mathcal{H}_{AE}$, we again use Eq.\ (\ref{Eq: ProofReducingActionofZ}). This generates a family of operators, $\{Z_{a,i}^{x,y,t}\}$, such that Eq.\ (\ref{Eq: ProofIntermediateVarOptProb2}) is equivalent to 
\begin{align} 
\sup_{Z_{a,i}^{x}} \sum_{a} \operatorname{Tr}\left[\mathcal{M}(\sigma)_{E |A=a,X=x}P(Z_{a,1}^{x},Z_{a,2}^{x},\dots)\right]+\operatorname{Tr}\left[\mathcal{M}(\sigma)_{E} Q(Z_{a,1}^{x},Z_{a,2}^{x},\dots)\right]\; .
    \end{align}
Eq.\ (\ref{Eq: XConditionalPUpperBounds2}) then follows, using both
    \begin{align}
       \mathcal{M}(\sigma)_{E|A=a, X=x} &= \operatorname{Tr}_{Q_AQ_B}\left[\left(M_{a|x} \otimes \mathds{1}_{Q_B} \otimes \mathds{1}_{E} \right) \sigma_{Q_AQ_BE}\right] \\
       \mathcal{M}(\sigma)_{E} &=  \operatorname{Tr}_{Q_AQ_B}\left[\left(\mathds{1}_{Q_AQ_B} \otimes \mathds{1}_{E} \right) \sigma_{Q_AQ_BE}\right] \; ,
    \end{align}
    and the identity~\cite[Eq.~3.2.19]{khatri2024principles}
    \begin{align}       \operatorname{Tr}_{Q_A Q_B}\left[ X_{Q_AQ_BE} \left(\mathds{1}_{Q_AQ_B} \otimes Z\right) \right] = \operatorname{Tr}_{Q_A Q_B}\left[ X_{Q_AQ_BE} \right] Z \; .
    \end{align}
\end{proof}

\subsubsection{Proofs of Theorem \ref{FirstConvexOptBound} and Theorem \ref{Theorem: SecondConvexOptBound}}
Both Theorem \ref{FirstConvexOptBound} and Theorem \ref{Theorem: SecondConvexOptBound} convert an optimization problem into a \textit{convex} problem, using the NPA hierarchy from~\cite{Pironio_2010}.
We will consider a family of expressions of the form 
\begin{align}
    \bra{\psi}p_{j}\left( X_1,X_2,\dots \right)\ket{\psi} \; ,
\end{align}
where $\ket{\psi}$ is a normalized, pure state on some Hilbert space, $H$, and $p_j$ are polynomials of operators, $X_1,X_2,\dots$, where one requires these operators to satisfy multiple constraints of the form
\begin{align}
    q_{k}\left( X_1,X_2,\dots \right) \succeq 0\; ,
\end{align}
where each $q_{k}$ is again a polynomial.
\begin{rem}
    If one requires a constraint of the form 
    \begin{align}
    q\left( X_1,X_2,\dots \right) = 0\; ,
\end{align}
this is not an issue. This can be trivially enforced, using the constraints
\begin{align}
    q\left( X_1,X_2,\dots \right) &\succeq 0 \\
    -q\left( X_1,X_2,\dots \right) &\succeq 0 \; .
\end{align}
For more details regarding such constraints, see~\cite[Section~3.5]{Pironio_2010}.
\end{rem}
Any given NPA relaxation will generate a family of variables $\{ y_i\}_{i \in N}$ (for some value $N$) and a matrix $M(y)$, with entries that are linear in $y_i$. A key feature of such a relaxation is that for any state $\ket{\psi}$ and operators $X_1,X_2,\dots$, which satisfy the relevant constraints, there exists a family of linear functions
\begin{align} \label{Eq: NPA1}
    f_j(y_1,y_2,\dots ):= \sum_{i=1}^{N}a_{ij}y_i \; ,
\end{align}
such that there are values of $\{ y_i\}_{i \in N}$ that both satisfy
\begin{align} \label{Eq: NPA2}
    \bra{\psi}p_{j}\left( X_1,X_2,\dots \right)\ket{\psi} = f_j(y_1,y_2,\dots ) \; ,
\end{align}
for all $j$, as well as the constraints 
\begin{align}
    y_1 &= 1 \label{Eq: NPA3}\\
    M(y) &\succeq 0 \; . \label{Eq: NPA4}
\end{align}
\FirstConvexOptBound*
\begin{proof}
    For any pure state $\sigma = \ketbra{\psi}{\psi}$, Proposition \ref{Prop:FirstNonComOpt} upper bounds each $p(\bar{c})$ by
\begin{align}
    p(\bar{c}) \leq \sup_{ Z_{\bar{c},a,i}^{x,y,t}}\bra{\psi}\sum_{\substack{a,x,y,t \ \mathrm{s.t.} \\ \operatorname{P}\left(a,\bar{c},x,y,t\right)>0}} \operatorname{P}\left(X=x,Y=y,T=t\right)p_{\bar{c},a,x,y,t}\left( Z_{\bar{c},a,1}^{x,y,t},Z_{\bar{c},a,2}^{x,y,t},\dots \right)\ket{\psi} \; ,
\end{align}
where 
\begin{align} \label{Eq: ProofIntermediateTensorproducts}
    p_{\bar{c},a,x,y,t}\left( Z_{\bar{c},a,1}^{x,y,t},Z_{\bar{c},a,2}^{x,y,t},\dots \right) :=  \sum_{\substack{b \ \mathrm{s.t.} \\ \bar{c} = f(a,b,x,y,t)}} M_{a|x} \otimes N_{b|y} \otimes P(Z_{\bar{c},a,1}^{x,y,t},\dots) + \mathds{1}_{Q_AQ_B} \otimes Q(Z_{\bar{c},a,1}^{x,y,t},\dots) \; ,
\end{align}
and the extra index $\bar{c}$ on the operators $Z_{\bar{c},a,i}^{x,y,t}$, indicates the fact that for each different value of $\bar{c}$, one is generating a new family of such operators. 

One way to generate lower bounds on Eq.\ (\ref{Eq: OptProbFirstConvexOptBound}), is by relaxing the tensor product constraints seen in Eq.\ (\ref{Eq: ProofIntermediateTensorproducts}) to (more general) commutation relations, i.e.\ one considers instead
   \begin{align} 
    p_{\bar{c},a,x,y,t}^{\prime}\left( Z_{\bar{c},a,1}^{x,y,t},Z_{\bar{c},a,2}^{x,y,t},\dots \right) :=  \sum_{\substack{b \ \mathrm{s.t.} \\ \bar{c} = f(a,b,x,y,t)}} M_{a|x} \cdot N_{b|y} \cdot P(Z_{\bar{c},a,1}^{x,y,t},\dots) + Q(Z_{\bar{c},a,1}^{x,y,t},\dots) \; ,
\end{align} 
and includes the constraints
\begin{align}
    \left[M_{a|x}, N_{b|y}\right] = \left[M_{a^\prime|x^\prime}, Z_{\bar{c},a,i}^{x,y,t (\dagger)}\right] = \left[N_{b^\prime|y^\prime}, Z_{\bar{c},a,i}^{x,y,t (\dagger)}\right] = 0
\end{align}
for all relevant $a,a^\prime, b,b^\prime, x,x^\prime, y,y^\prime, t$, and $i$. If one now includes to the remaining additional constraints that $\{M_{a|x}\}$ and $\{N_{b|y}\}$ must both be POVM's, one finds that 
    \begin{equation}
\begin{aligned} \label{Eq: ProofIntermediateOptimizationProblemCommCond}
 \inf_{\{M_{a|x}\}, \{N_{b|y}\}}\inf_{q} \inf_{\ket{\psi}} \inf_{ Z_{\bar{c},a,i}^{x,y,t}} \quad &\frac{1}{\alpha-1}\sum_{\bar{c}}q\left(\bar{c} \right) \log \left( \frac{q\left(\bar{c} \right)}{\bra{\psi} \sum_{\substack{a,x,y,t \ \mathrm{s.t.} \\ \operatorname{P}\left(a,\bar{c},x,y,t\right)>0}} \operatorname{P}\left(x,y,t\right)p_{\bar{c},a,x,y,t}^{\prime}\left( Z_{\bar{c},a,1}^{x,y,t},\dots \right)\ket{\psi}} \right)  \\
\textrm{s.t.} \quad & \sum_{a} M_{a|x} -  1 = \sum_{b} N_{b|y} -  1 = 0 \\
&  M_{a|x} \succeq 0 \\
&  N_{b|y} \succeq 0 \\
&  \left[M_{a|x}, N_{b|y}\right] = \left[M_{a^\prime|x^\prime}, Z_{\bar{c},a,i}^{x,y,t (\dagger)}\right] = \left[N_{b^\prime|y^\prime}, Z_{\bar{c},a,i}^{x,y,t (\dagger)}\right] = 0 \\
& q \in S_\Omega 
\end{aligned}
\end{equation}
    is a lower bound on Eq.\ (\ref{Eq: OptProbFirstConvexOptBound}).
\begin{rem}
    In principle, one may w.l.o.g.\ assume that $\{M_{a|x}\}$ and $\{N_{b|y}\}$ are projectors. This constraint can be added as follows:
    \begin{align}
        M_{a|x}^2 - M_{a|x} & = 0\\
        N_{b|y}^2 - N_{b|y} & = 0 \; .
    \end{align}
\end{rem}
Let $\{ y_i\}_{i \in N}$, $M(y)$ be the variables and matrix generated by a given NPA relaxation. Then, by Eqs.\ (\ref{Eq: NPA1})--(\ref{Eq: NPA4}), for every value of $\bar{c}$, one generates a linear function 
   \begin{align} 
    f_{\bar{c}}(y_1,y_2,\dots ):= \sum_{i=1}^{N}a_{i,\bar{c}}y_i \; .
\end{align} 
Moreover, for every feasible point in the optimization problem described in Eq.\ (\ref{Eq: ProofIntermediateOptimizationProblemCommCond}), there exist specific values of $\{ y_i\}_{i \in N}$ that satisfy 
\begin{align}
    f_{\bar{c}}(y_1,y_2,\dots ) = \bra{\psi} \sum_{\substack{a,x,y,t \ \mathrm{s.t.} \\ \operatorname{P}\left(a,\bar{c},x,y,t\right)>0}} \operatorname{P}\left(X=x,Y=y,T=t\right)p_{\bar{c},a,x,y,t}^{\prime}\left( Z_{\bar{c},a,1}^{x,y,t}, Z_{\bar{c},a,2}^{x,y,t},\dots \right)\ket{\psi} 
\end{align}
for all $\bar{c}$, as well as
\begin{align} 
y_1 &= 1 \label{Eq: OurNPARel1}\\
 M(y) &\succeq 0 \; . \label{Eq: OurNPARel2}
\end{align}
As such, replacing $\bra{\psi} \sum_{\substack{a,x,y,t \ \mathrm{s.t.} \\ \operatorname{P}\left(a,\bar{c},x,y,t\right)>0}} \operatorname{P}\left(x,y,t\right)p_{\bar{c},a,x,y,t}^{\prime}\left( Z_{\bar{c},a,1}^{x,y,t}, Z_{\bar{c},a,2}^{x,y,t},\dots \right)\ket{\psi} $ with $f_{\bar{c}}(y_1,y_2,\dots )$, and subsequently minimizing over all $\{ y_i\}_{i \in N}$, which satisfy both Eq.\ (\ref{Eq: OurNPARel1}) and Eq.\ (\ref{Eq: OurNPARel2}), will yield a lower bound on Eq.\ (\ref{Eq: ProofIntermediateOptimizationProblemCommCond}). The resulting bound is described by the optimization problem
        \begin{equation}
\begin{aligned}
\inf_{q , \{y_i\} } \quad &\frac{1}{\alpha-1}D\left(q \middle\Vert f\right) \\
\textrm{s.t.} \quad & y_1 = 1 \\
& M(y) \succeq 0 \\
& q \in S_\Omega \; .
\end{aligned}
\end{equation}
\end{proof}
\SecondConvexOptBound*
\begin{proof}
    For any value $\bar{C} = \bar{c}$ and pure state $\sigma = \ketbra{\psi}{\psi}$, one has that
    \begin{align} \label{Eq: ProofTestOutputDistribution}
     \mathcal{M}(\sigma)_{\bar{C}=\bar{c}} = \bra{\psi} \sum_{\substack{a,b,x,y,t \ \mathrm{s.t.} \\ \bar{c} = f(a,b,x,y,t)}}  \operatorname{P}\left(X=x,Y=y,T=t\right) M_{a|x} \otimes N_{b|y} \otimes \mathds{1}_{E}\ket{\psi} \; .
    \end{align}
    Moreover, by Proposition \ref{Prop: SecondNonComOpt},
    \begin{align} \label{Eq: ProofCondEntBound}
        H_{\alpha}\left( A|T=0,XE\right)_{\mathcal{M}(\sigma)} \geq  \inf_{ Z_{a,i}^{x}} \frac{1}{1-\alpha}\log_{2}\left(\bra{\psi} \sum_{x \in \mathcal{X}_{K}} \sum_{a \in A_{x}} \operatorname{P}\left(X=x|T=0\right) p_{a,x} \left( Z_{a,1}^{x},Z_{a,2}^{x},\dots\right) \ket{\psi}
     \right) \; ,
    \end{align}
    where 
 \begin{align} 
p_{a,x} \left( Z_{a,1}^{x},Z_{a,2}^{x},\dots\right)=   M_{a|x} \otimes \mathds{1}_{Q_B} \otimes P(Z_{a,1}^{x},Z_{a,2}^{x},\dots) + \mathds{1}_{Q_AQ_B} \otimes Q(Z_{a,1}^{x},Z_{a,2}^{x},\dots)\; ,
\end{align}
and 
\begin{align}
    P(Z_{a,1}^{x},Z_{a,2}^{x},\dots) &:= 1 +\frac{\sin(\alpha \pi )}{\pi}\sum_{i=1}^{m} \frac{w_{i}}{t_{i}} \left(\mathds{1} +Z_{a,i}^{x}+Z_{a,i}^{x \dagger} + \left(1-t_i\right)Z_{a,i}^{x \dagger} Z_{a,i}^{x}\right)\\
    Q(Z_{a,1}^{x},Z_{a,2}^{x},\dots) &:= \frac{\sin(\alpha \pi )}{\pi}\sum_{i=1}^{m} w_i Z_{a,i}^{x}Z_{a,i}^{x \dagger} \; .
\end{align}
One potential lower bound on Eq.\ (\ref{Eq: OptProbSecondConvexOptBound}) is achieved by relaxing the implicit tensor product constraints seen in both Eq.\ (\ref{Eq: ProofTestOutputDistribution}) and Eq.\ (\ref{Eq: ProofCondEntBound}) to instead requiring the commutation relations
    \begin{align}
    \left[M_{a|x}, N_{b|y}\right] = \left[M_{a^\prime|x^\prime}, Z_{a,i}^{x(\dagger)}\right] = \left[N_{b^\prime|y^\prime}, Z_{a,i}^{x (\dagger)}\right] = 0
\end{align}
for all relevant $a,a^\prime, b,b^\prime, x,x^\prime$, and $i$, which are more general constraints. Let $\{ y_i\}_{i \in N}$, $M(y)$ be the variables and matrix generated by a given NPA relaxation. Then, by Eqs.\ (\ref{Eq: NPA1})--(\ref{Eq: NPA4}), one can generate a linear function 
 \begin{align} 
    g(y_1,y_2,\dots ):= \sum_{i=1}^{N}a_{i,\bar{c}}y_i \; ,
\end{align} 
as well as an extra linear function 
   \begin{align} 
    f_{\bar{c}}(y_1,y_2,\dots ):= \sum_{i=1}^{N}a_{i,\bar{c}}y_i 
\end{align} 
for every value of $\bar{c}$.
Moreover, for every feasible point in the relaxed optimization problem, there exist specific values of $\{ y_i\}_{i \in N}$ that satisfy 
\begin{align}
   g(y_1,y_2,\dots ) = \bra{\psi} \sum_{x \in \mathcal{X}_{K}} \sum_{a \in A_{x}} \operatorname{P}\left(X=x|T=0\right) p^{\prime}_{a,x} \left( Z_{a,1}^{x},Z_{a,2}^{x},\dots\right) \ket{\psi} \; ,
\end{align}
where 
\begin{align} 
p^{\prime}_{a,x} \left( Z_{a,1}^{x},Z_{a,2}^{x},\dots\right)=   M_{a|x} \cdot P(Z_{a,1}^{x},Z_{a,2}^{x},\dots)  + Q(Z_{a,1}^{x},Z_{a,2}^{x},\dots)\; ,
\end{align}
and
\begin{align}
    f_{\bar{c}}(y_1,y_2,\dots ) = \bra{\psi} \sum_{\substack{a,b,x,y,t \ \mathrm{s.t.} \\ \bar{c} = f(a,b,x,y,t)}}  \operatorname{P}\left(X=x,Y=y,T=t\right) M_{a|x} \cdot N_{b|y}\ket{\psi} 
\end{align}
for all $\bar{c}$, as well as
\begin{align} 
y_1 &= 1 \\
 M(y) &\succeq 0 \; . 
\end{align}
As such, a lower bound on Eq.\ (\ref{Eq: OptProbSecondConvexOptBound}) is given by the optimization problem
        \begin{equation}
\begin{aligned}
\inf_{q , \{y_i\} } \quad &\frac{1}{\alpha-1}D\left(q \middle\Vert f\right) + \frac{q_{min}}{1-\alpha}\log_{2}\left(g\right) \\
\textrm{s.t.} \quad & y_1 = 1 \\
& M(y) \succeq 0 \\
& q \in S_\Omega \; .
\end{aligned}
\end{equation}

\end{proof}

\pagebreak
\printbibliography


\end{document}